\documentclass[final]{siamltex}

\usepackage{amsfonts,amsmath,amssymb}
\usepackage[ruled,linesnumbered,algosection,vlined]{algorithm2e}
\usepackage[noend]{algorithmic}
\renewcommand{\CommentSty}[1]{\footnotesize\itshape{#1}\unskip}
\renewcommand{\ArgSty}[1]{#1\unskip}
\SetKwComment{tcc}{[}{]}
\SetKwFor{ForAll}{for all}{do}{end for}%
\SetVlineSkip{2pt}
\renewcommand{\AlFnt}{\footnotesize}
\renewcommand{\NlSty}[1]{#1:}
\newcommand{\Comment}[1]{\tcc*[f]{#1}}
\newcommand{\algand}{\textbf{and}\xspace}
\newcommand{\algor}{\textbf{or}\xspace}
\newcommand{\algnot}{\textbf{not}\xspace}
\newcommand{\algto}{\textbf{to}\xspace}
\newcommand{\algtrue}{\textbf{true}\xspace}
\newcommand{\algfalse}{\textbf{false}\xspace}
\SetKwInOut{Input}{Input}
\SetKwInOut{Result}{Result}
\DontPrintSemicolon
\usepackage{array}
\usepackage{calc}
\usepackage{cite}
\usepackage{color}
\usepackage{comment}
\usepackage{graphicx}
\usepackage{subfigure}
\usepackage{tabularx}
\usepackage{tikz}
\usetikzlibrary{trees}
\usetikzlibrary{quadtree}
\usetikzlibrary{shadows}
\usetikzlibrary{decorations.pathreplacing}
\usepackage{multirow}
\usepackage{url}
\usepackage{xspace}
\usepackage{pgfplots}
\usepgfplotslibrary{groupplots}
\usepackage{nth}
\usepackage{microtype}

\newcommand{\query}{\mathrm{query}}

\newcommand{\cO}{\mathcal{O}}
\newcommand{\cT}{\mathcal{T}}

\newcommand{\oc}[1]{\text{\upshape\texttt{#1}}\xspace}
\newcommand{\ix}[1]{#1}
\newcommand{\pr}[1]{\mathfrak{#1}}
\newcommand{\st}[1]{\mathcal{#1}}

\newcommand{\rfirst}{\mathrm{first}}
\newcommand{\rlast}{\mathrm{last}}

\newcommand{\rmatch}{\mathrm{match}}

\newcommand{\rshare}{\mathrm{share}}

\newcommand{\dealii}{\texttt{deal.II}\xspace}

\newcommand{\pforest}{\texttt{p4est}\xspace}

\newcommand{\localtreeoctants}{\arry{O}^{\ix t}_{\pr p}}
\newcommand{\localoctants}{\mathcal{O}_{\pr p}}

\newcommand{\partoctants}{\mathcal{P}_{\Omega}}
\newcommand{\localpart}{\mathcal{P}_{\pr p}}
\newcommand{\treeoctants}{\mathcal{O}^{\ix t}}
\newcommand{\ternary}[3]{(#1) \;?\; #2 \;:\; #3}
\newcommand{\bitwand}{\;\&\;}
\newcommand{\bitwor}{\;\vert\;}

\newtheorem{theo}{Theorem}[section]
\newtheorem{defn}[theo]{Definition}
\newtheorem{propn}[theo]{Proposition}
\newtheorem{remark}[theo]{Remark}

\newcommand{\eqnlab}[1]{\label{eqn:#1}}

\newcommand{\figlab}[1]{\label{fig:#1}}
\newcommand{\figref}[1]{Figure~\ref{fig:#1}}
\newcommand{\alglab}[1]{\label{alg:#1}}
\newcommand{\algref}[1]{Algorithm~\ref{alg:#1}}
\newcommand{\seclab}[1]{\label{sec:#1}}
\newcommand{\secref}[1]{Section~\ref{sec:#1}}

\newcommand{\tablab}[1]{\label{tab:#1}}
\newcommand{\tabref}[1]{Table~\ref{tab:#1}}

\newcommand{\applab}[1]{\label{app:#1}}
\newcommand{\appref}[1]{Appendix~\ref{app:#1}}
\newcommand{\remlab}[1]{\label{rem:#1}}
\newcommand{\remref}[1]{Remark~\ref{rem:#1}}

\newcommand{\dtype}[1]{{#1}\xspace}
\newcommand{\Int}{\dtype{integer}}
\newcommand{\Octant}{\dtype{octant}}

\newcommand{\Octants}{\dtype{octants}}
\newcommand{\Array}{\dtype{array}}

\newcommand{\Boolean}{\dtype{boolean}}
\newcommand{\Callback}{\dtype{callback}\xspace}

\newcommand{\dfield}[1]{{\normalfont\texttt{#1}}}
\newcommand{\df}[1]{\dfield{#1}}

\newcommand{\arry}{\mathbf}
\newcommand{\fxn}[1]{\textup{\texttt{#1}}}

\newcommand{\asupp}{\mathrm{atom\,supp}}
\newcommand{\lsupp}{\mathrm{leaf\,supp}}

\newcommand{\bound}{\mathrm{bound}}
\newcommand{\child}{\mathrm{child}}
\newcommand{\desc}{\mathrm{desc}}
\newcommand{\clos}{\mathrm{clos}}
\newcommand{\dom}{\mathrm{dom}}
\newcommand{\Dom}{\mathrm{Dom}}
\newcommand{\level}{\mathrm{level}}
\newcommand{\locate}{\mathrm{locate}}
\newcommand{\owner}{\mathrm{owner}}
\newcommand{\range}{\mathrm{range}}
\newcommand{\Root}{\mathrm{root}}
\newcommand{\part}{\mathrm{part}}
\newcommand{\lmax}{\ix l_{\max}}

\newcommand{\algorule}{%
  \vspace{0.5ex}
  \hrule
  \vspace{0.5ex}
}

\pgfdeclareplotmark{halftriangle*}
{%
  \pgfpathmoveto{\pgfqpoint{0pt}{\pgfplotmarksize}}
  \pgfpathlineto{\pgfqpoint{0pt}{-0.5\pgfplotmarksize}}
  \pgfpathlineto{\pgfqpointpolar{-150}{\pgfplotmarksize}}
  \pgfpathclose
  \pgfusepathqfillstroke
  \pgfpathmoveto{\pgfqpoint{0pt}{\pgfplotmarksize}}
  \pgfpathlineto{\pgfqpointpolar{-30}{\pgfplotmarksize}}
  \pgfpathlineto{\pgfqpointpolar{-150}{\pgfplotmarksize}}
  \pgfpathclose
  \pgfusepathqstroke
}

\pgfdeclareplotmark{halfpentagon*}
{%
  \pgfpathmoveto{\pgfqpoint{0pt}{\pgfplotmarksize}}
  \pgfpathlineto{\pgfqpointpolar{18}{\pgfplotmarksize}}
  \pgfpathlineto{\pgfqpointpolar{-54}{\pgfplotmarksize}}
  \pgfpathlineto{\pgfqpointpolar{234}{\pgfplotmarksize}}
  \pgfpathlineto{\pgfqpointpolar{162}{\pgfplotmarksize}}
  \pgfpathclose
  \pgfusepathqstroke
  \pgfpathmoveto{\pgfqpoint{0pt}{\pgfplotmarksize}}
  \pgfpathlineto{\pgfqpointpolar{18}{\pgfplotmarksize}}
  \pgfpathlineto{\pgfqpointpolar{-54}{\pgfplotmarksize}}
  \pgfpathlineto{\pgfqpoint{0pt}{-0.809\pgfplotmarksize}}
  \pgfpathclose
  \pgfusepathqfillstroke
}

\pgfplotscreateplotcyclelist{p4est2}{%
every mark/.append style={draw=blue!80!black},mark=o\\%
every mark/.append style={draw=red!80!black},mark=triangle\\%
every mark/.append style={draw=brown!80!black},mark=square\\%
every mark/.append style={draw=black},mark=diamond\\%
every mark/.append style={draw=green!80!black},mark=pentagon\\%
every mark/.append style={draw=orange!80!black,rotate=90},mark=diamond\\%
every mark/.append style={draw=blue!80!black,fill=blue!80!black},mark=halfcircle*\\%
every mark/.append style={draw=red!80!black,fill=red!80!black},mark=halftriangle*\\%
every mark/.append style={draw=brown!80!black,fill=brown!80!black,rotate=45,scale=1.414},mark=halfsquare*\\%
every mark/.append style={draw=black,fill=black},mark=halfdiamond*\\%
every mark/.append style={draw=green!80!black,fill=green!80!black},mark=halfpentagon*\\%
every mark/.append style={draw=orange!80!black,fill=orange!80!black,rotate=90},mark=halfdiamond*\\%
every mark/.append style={draw=blue!80!black,fill=blue!80!black},mark=*\\%
every mark/.append style={draw=red!80!black,fill=red!80!black},mark=triangle*\\%
every mark/.append style={draw=brown!80!black,fill=brown!80!black},mark=square*\\%
every mark/.append style={draw=black,fill=black},mark=diamond*\\%
every mark/.append style={draw=green!80!black,fill=green!80!black},mark=pentagon*\\%
every mark/.append style={draw=orange!80!black,fill=orange!80!black,rotate=90},mark=diamond*\\%
}

\pgfplotscreateplotcyclelist{p4est2search}{%
every mark/.append style={fill=blue!80!black,draw=blue!80!black},mark=o\\%
every mark/.append style={fill=blue!80!black,draw=blue!80!black},mark=halfcircle*\\%
every mark/.append style={fill=red!80!black,draw=red!80!black},mark=triangle\\%
every mark/.append style={fill=red!80!black,draw=red!80!black},mark=halftriangle*\\%
every mark/.append style={fill=red!80!black,draw=red!80!black},mark=triangle*\\%
every mark/.append style={fill=brown!80!black,draw=brown!80!black},mark=square\\%
every mark/.append style={fill=brown!80!black,draw=brown!80!black,rotate=45,scale=1.414},mark=halfsquare*\\%
every mark/.append style={fill=brown!80!black,draw=brown!80!black},mark=square*\\%
}

\title{Recursive Algorithms for Distributed\\ Forests of Octrees}

\author{%
Tobin Isaac\footnotemark[1]
\and
Carsten Burstedde\footnotemark[2]
\and
Lucas C.\ Wilcox\footnotemark[3]
\and
Omar Ghattas\footnotemark[1] \footnotemark[4] \footnotemark[5]}

\begin{document}

\thispagestyle{plain}

\maketitle

\renewcommand{\thefootnote}{\fnsymbol{footnote}}
\footnotetext[1]{Institute for Computational Engineering and Sciences,
  The University of Texas at Austin, USA}
\footnotetext[2]{Institut f\"ur Numerische Simulation,
  Rheinische Friedrich-Wilhelms-Universit\"at Bonn, Germany}
\footnotetext[3]{Department of Applied Mathematics,
  Naval Postgraduate School, USA}
\footnotetext[4]{Jackson School of Geosciences,
  The University of Texas at Austin, USA}
\footnotetext[5]{Department of Mechanical Engineering,
  The University of Texas at Austin, USA}
\renewcommand{\thefootnote}{\arabic{footnote}}

\begin{abstract}
  The forest-of-octrees approach to parallel adaptive mesh refinement and
coarsening (AMR) has recently been demonstrated in the context of a number of
large-scale PDE-based applications.  Efficient reference software has been
made freely available to the public both in the form of the standalone
\pforest library and more indirectly by the general-purpose finite element
library \dealii, which has been equipped with a \pforest backend. 

Although linear octrees, which store only leaf octants, have an underlying
tree structure by definition, it is not fully exploited in previously
published mesh-related algorithms.  This is because tree branches are not
explicitly stored, and because the topological relationships in meshes, such
as the adjacency between cells, introduce dependencies that do not respect the
octree hierarchy.  In this work we combine hierarchical and topological
relationships between octants to design efficient recursive algorithms
that operate on distributed forests of octrees.

We present three important algorithms with recursive implementations.  The
first is a parallel search for leaves matching any of a set of multiple search
criteria, such as leaves that contain points or intersect polytopes.  The
second is a ghost layer construction algorithm that handles arbitrarily
refined octrees that are not covered by previous algorithms, which require a
2:1 condition between neighboring leaves.  The third is a universal mesh
topology iterator.  This iterator visits every cell in a partition, as well as
every interface (face, edge and corner) between these cells.  The iterator
calculates the local topological information for every interface that it
visits, taking into account the nonconforming interfaces that increase the
complexity of describing the local topology.  To demonstrate the utility of
the topology iterator, we use it to compute the numbering and encoding of
higher-order $C^0$ nodal basis functions used for finite elements.

We analyze the complexity of the new recursive algorithms theoretically, and
assess their performance, both in terms of single-processor efficiency and in
terms of parallel scalability, demonstrating good weak and strong scaling up
to 458k cores of the JUQUEEN supercomputer.



\end{abstract}

\begin{keywords}
  forest of octrees, parallel adaptive mesh refinement, Morton code,
  recursive algorithms, large-scale scientific computing
\end{keywords}

\begin{AMS}
  65M50, 
  68W10, 
  65Y05, 
  65D18  
\end{AMS}

\section{Introduction}

The development of efficient and scalable parallel algorithms that modify
computational meshes is necessary for resolving features in large-scale
simulations.  These features may vanish and reappear, and/or evolve in shape
and location, which stresses the dynamic and in-situ aspects of adaptive mesh
refinement and coarsening (AMR).  Both stationary and time-dependent
simulations benefit from flexible and fast remeshing and repartitioning
capabilities, for example when using a-posteriori error estimation, building
mesh hierarchies for multilevel solvers for partial differential equations
(PDEs), or tracking of non-uniformly distributed particles by using an
underlying adaptive mesh.

Three main algorithmic approaches to AMR have emerged over time, which we may
call unstructured (U), block-structured (S), and hierarchical or tree-based
(T) AMR.  Just some examples that integrate parallel processing are (U)
\cite{NortonLouCwik01, LawlorChakravortyWilmarthEtAl06,
ZhouSahniDevineEtAl10}, (S)
\cite{MacNeiceOlsonMobarryEtAl00, GoodaleAllenLanfermannEtAl03,
  ColellaBellKeenEtAl07, LuitjensBerzinsHenderson08}, and (T)
\cite{Popinet03, TuOHallaronGhattas05, SundarSampathBiros08}.  While
these approaches have been developed independently of one another, there has
been a definite crossover of key technologies.  The graph-based partitioning
algorithms traditionally used in UAMR have for instance been supplemented by
fast algorithms based on coordinate partitioning and space-filling curves
(SFCs) \cite{DevineBomanHeaphyEtAl02}.  Hierarchical ideas and SFCs have also
been applied in SAMR packages to speed up and improve the partitioning
\cite{DreherGrauer05, ColellaGravesKeenEtAl07}.  Last but not least, the
unstructured meshing paradigm can be employed to create a root mesh of
connected trees when a nontrivial geometry needs to be meshed by
forest-of-octrees TAMR
\cite{StewartEdwards04,
BangerthHartmannKanschat07, BursteddeWilcoxGhattas11}.

The three approaches mentioned above differ in the way that the mesh topology
information is passed to applications.  With UAMR, the mesh is usually stored
in memory as an adjacency graph, and the application traverses the graph to
compute residuals, assemble system matrices, etc.  This approach has the
advantages that local graph traversal operations typically have constant
runtime complexity and that the AMR library can remain oblivious of the
details of the application, but the disadvantages of less efficient global
operations, such as locating the cell containing a point, and of unpredictable
memory access.  On the other hand, the SAMR approach allows for common
operations to be optimized and to use regular memory access patterns, but
requires more integration between the AMR package and the application, which
may not have access to the topology in a way not anticipated by the AMR
package.

Tree-based AMR can be integrated with an application for convenience
\cite{SampathAdavaniSundarEtAl08}, but can also be kept strictly modular
\cite{TuOHallaronGhattas05}.  Most TAMR packages implement
logarithmic-complexity algorithms for both global operations, such as point
location, and local operations, such as adjacency queries.  The paper
\cite{BursteddeWilcoxGhattas11} introduces the \pforest library, which
implements distributed forest-of-octree AMR with an emphasis on geometric and
topological flexibility and parallel scalability, and connects with
applications through a minimal interface.

The implementation of \pforest does not explicitly build a tree data
structure, so tree-based, recursive algorithms are largely absent from the
original presentation in \cite{BursteddeWilcoxGhattas11}.  Many topological
operations on octrees and quadtrees, however, are naturally expressed as
recursive algorithms, which have simple descriptions and often have good,
cache-oblivious memory access patterns.  In this paper, we present, analyze,
and demonstrate the efficiency of algorithms for important hierarchical and
topological operations: searching for leaves matching multiple criteria in
parallel, identifying neighboring (ghost) domains from minimal information, and
iterating over mesh cells and interfaces.  We include an example algorithm
that uses this iteration to create a high-order finite element node numbering.
Each algorithm has a key recursive
component that gives it an advantage over previously developed non-recursive
algorithms, such as improved efficiency, coverage of additional use cases, or
both.  We demonstrate the per-process efficiency of these algorithms, as well
as their parallel scalability, on the full size of JUQUEEN \cite{juqueen}, a
Blue Gene/Q \cite{HaringOhmachtFoxEtAl12} supercomputer.

\section{Forest of octree types and operations}
\seclab{prelim}

Here we present the important concepts on which we build our algorithms.  We
review the data structures for octants and distributed forests of octrees that
were presented in \cite{BursteddeWilcoxGhattas11}.  We also define a data type
to handle both octants and octant boundaries that will allow us to describe
the topology of forests of octrees.\footnote{%
  This data type is a notational convenience for this work, not part of the
  \pforest interface.%
} %
The definitions in this section are summarized in \tabref{assumeddata}.  For
the sake of correctness, the definitions in this section are given formally,
but the reader may find that the figures are just as helpful in understanding
our concepts, as they often correspond to geometrically intuitive ideas.

\begin{remark}[Notation]\upshape
  If we have defined an operation $\mathrm{op}(\cdot)$ for every $a\in\st A$,
  then $\mathrm{op}(\mathcal{A}) := \{\mathrm{op}(a): a\in\mathcal{A}\}$.
  $|\st A|$ is the cardinality of set $\st A$.  If $\{\st A_i\}_{i\in \st I}$
  are disjoint, their union is written $\bigsqcup_{i\in\st I}\st A_i$.  For a
  subset $A$ of a manifold, $\overline{A}$, $\partial A$, and $A^{\circ}$ are
  the closure, boundary, and interior of $A$.  We distinguish variable types
  with fonts:
  \begin{itemize}
    \item standard lower-case for integers and index sets
                                              ($\ix a,\ix b,\ix c,\dots$),
          except for $K,$ $N,$ and $P,$ which are the number of octrees,
          octants, and processes,
    \item typewriter for compound data types  ($\oc a,\oc b,\oc c,\dots$),
    \item Fraktur for MPI processes           ($\pr a,\pr b,\pr c,\dots$),
    \item upper-case for subsets of $\mathbb{R}^d$ and manifolds
                                              ($A, B, C, \dots$),
    \item calligraphic for finite sets        ($\st A,\st B, \st C,\dots$),
      and
    \item bold for finite sets represented as indexable arrays ($\arry A,
      \arry B, \arry C, \dots$).
  \end{itemize}
\end{remark}

\begin{table}
  \caption{A summary of \secref{prelim} and the locations of the definitions
    in the text.}
  \tablab{assumeddata}
  \renewcommand{\arraystretch}{1.1}
  \centering
  \begin{tabularx}{\textwidth}{|l|X|l|}
    \hline \multicolumn{1}{|c}{} &
    \multicolumn{2}{l|}{\S~\ref{sec:octant} Octants and points} \\ \hline
    octant $\oc o$, $\dom(\oc o)$ & Data type and the cube in
    $\mathbb{R}^d$ it represents & (\ref{eqn:octdomain}) \\
    $\Root(\ix t)$ & Root of the $\ix t$-th octree: side length $2^{\lmax}$ &
    (\ref{eqn:root})
    \\
    atom $\oc a$ & level-$\lmax$ octant: side length $1$ &
    \\
    %
    %
    $\{\dom_b(\oc o)\}_{b\in \st B \cup \{v_0\}}$ & Octant boundary domains and
    their indices &
    Fig. \ref{fig:boundarysets}
    \\
    point $\oc c = (\oc o,b)$ & Common data type for octants and interfaces &
    \\

    $\dom(\oc c)$ & A point's $n$-dim.\ ($n \leq d$) hypercube domain &
    (\ref{eqn:point})
    \\
    $\dim(\oc c)$ & Topological dimension of a point &
    \\
    $\level(\oc c)$ & Refinement level for points &
    (\ref{eqn:pointlevel})
    \\
    \hline \multicolumn{1}{|c}{} &
    \multicolumn{2}{l|}{\S~\ref{sec:relations} Hierarchical and topological relationships} \\ \hline
    $\desc(\oc c)$ & Descendants of point $\oc c$               &
    (\ref{eqn:pointdesc}) \\
    $\child(\oc c)$ & Children of point $\oc c$ \hfill(\tabref{sets}, 3rd
    column)               &
    (\ref{eqn:pointchildren}) \\
    $\part(\oc c)$  & Child partition of point $\oc c$ \hfill(", 4th
    column)       & (\ref{eqn:pointpart}) \\
    $\clos(\oc c)$  & Closure set of point $\oc c$            &
    (\ref{eqn:pointclos}) \\
    $\bound(\oc c)$ & Boundary set of point $\oc c$ \hfill(", 2nd
    column)           & (\ref{eqn:pointbound}) \\
    $\supp(\oc c)$  & Support set of point $\oc c$ \hfill(", 5th
    column)           & (\ref{eqn:pointsupp}) \\
    $\asupp(\oc c)$ & Atomic support set of $0$-point $\oc c$ &
    (\ref{eqn:pointasupp}) \\
    \hline \multicolumn{1}{|c}{} &
    \multicolumn{2}{l|}{\S~\ref{sec:forest} Forests of octrees} \\ \hline
    $\cT:=\{(T^{\ix t},\varphi^{\ix t})\}_{0\leq t < K}$ & Conformal
    macro mesh of $\Omega$ & (\ref{eqn:macro}) \\
    $\Dom(\oc c)$ & Point domain mapped by $\varphi^{\ix t}$ into $\Omega$ &
    (\ref{eqn:mapped})
    \\
    $\cO:= \bigsqcup_{0\leq t < K}\cO^{\ix t}$ & Non-conformal mesh via octree
  refinement & (\ref{eqn:micro}) \\
    \hline \multicolumn{1}{|c}{} &
    \multicolumn{2}{l|}{\S~\ref{sec:distributed} Distributed forests of octrees} \\ \hline
    $\oc o \leq \oc r$ & SFC-based total octant order & Alg.~\ref{alg:comparison} \\
    $\localoctants := \bigsqcup_{0 \leq t < K} \localtreeoctants$ & Sorted arrays of leaves owned
    by $\pr p$ for each tree & (\ref{eqn:localoctants}) \\
    $\Omega_{\pr p}$, $\Omega^{\ix t}$, $\Omega_{\pr p}^{\ix t}$ & Subdomains
    of $\localoctants$, $\treeoctants$, $\localtreeoctants$ &
    (\ref{eqn:subdomains}) \\
    $\locate(\oc a)$ & Process $\pr q$ such that $\Dom(\oc
      a)\subseteq\Omega_{\pr q}$ for atom $\oc a$ & (\ref{eqn:locate}) \\
    $(\oc f_{\pr q},\oc l_{\pr q}) := \range(\pr q)$ & First and
    last atoms located in $\Omega_{\pr q}$ & (\ref{eqn:procrange}) \\
    $(\oc f_{\oc o},\oc l_{\oc o}) := \range(\oc o)$ & First and last atoms
      in octant $\oc o$'s descendants & (\ref{eqn:octrange}) \\
    $\arry f:=\{\oc f_{\pr q}\}_{0\leq \pr q < P}$ & Sorted array of the
    first atoms of all processes & (\ref{eqn:firstatoms}) \\
    $\oc F_{\pr p}:=(\st T,\localoctants,\arry f)$ &
    Distributed forest of octrees & (\ref{eqn:distributed}) \\
    \hline
  \end{tabularx}
\end{table}

\subsection{Octants and points}
\seclab{octant}

\begin{figure}
  \begin{center}
    \includegraphics[scale=1.]{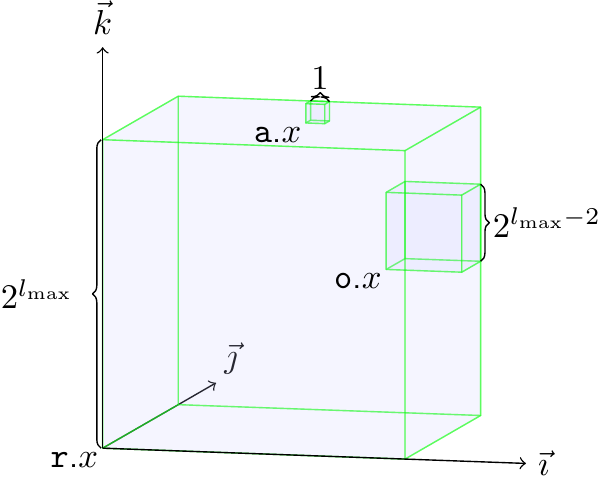}
    \hfill
    \includegraphics[scale=1.]{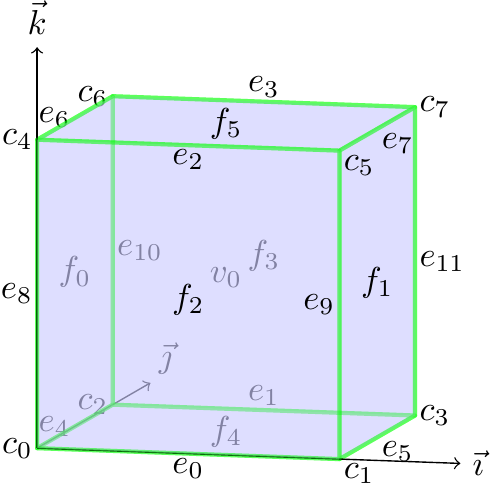}
  \end{center}
  \caption{%
    (left) An illustration of the domains of a root octant $\oc r$, a level-2
    octant $\oc o$, and an atom $\oc a$.  (right) The correlation between the
    boundary indices in $\st B$ (see the definition at the bottom of this
    page) and the lower-dimensional hypercubes---squares, line segments, and
    vertices---in the boundary of an octant, with the central cube labeled
    with the volume index, $v_0$ (adapted with permission from \cite[Fig.\
    2]{BursteddeWilcoxGhattas11}).  These indices are used to define points.%
  }%
  \figlab{boundarysets}%
\end{figure}%

Here we define the octant data type, which we will use in our algorithms, and
some special octants, which are illustrated in \figref{boundarysets}.

  An \emph{octant} $(d=3)$ or \emph{quadrant}
  $(d=2)$ $\oc o$ has the following data fields:
  \begin{itemize}
    \item
      $\oc o.\ix t\in \mathbb{N}$---$\oc o$'s \emph{tree index}, relevant to
      forests of octrees (see \secref{forest});%
      \footnote{In \pforest, the tree index is always available from context,
        not stored with the octant.}
    \item
      $\oc o.\ix l \in \{0,1,\dots,\lmax\}$---$\oc o$'s \emph{level of refinement} (or
      just \emph{level});
    \item
      $\oc o.\ix x\in \mathbb{Z}^d$---$\oc o$'s \emph{coordinates}, whose
      components must be multiples of $2^{\lmax - \oc o.\ix l}$.
  \end{itemize}
  The fields $\oc o.\ix l$ and $\oc o.\ix x$ encode an open cube in
  $\mathbb{R}^d$---$\oc o$'s \emph{domain}---with sides of length
  $2^{\lmax - \oc o.\ix l}$,
  \begin{equation}\eqnlab{octdomain}
    \mathop{\mathrm{dom}}(\oc o):=\{X\in\mathbb{R}^d: \oc o.\ix x_i < X_i <
    \oc o.\ix x_i + 2^{\lmax - \oc o.\ix l},\ 0\leq i < d\}.
  \end{equation}
%
%
  For every tree index $\ix t$, the \emph{root} is the level-0
  octant whose domain is $(0,2^{\lmax})^d$:
  \begin{equation}\eqnlab{root}
    \Root(\ix t).\ix t := t,\quad \Root(\ix t).\ix l := 0,\quad \Root(\ix
    t).\ix x := (0)^d.
  \end{equation}
%
  An \emph{atom} $\oc a$ is a smallest-possible octant, which has $\oc
  a.\ix l= \lmax$ and sides of length 1.


The algorithms we present involve both the hierarchical aspect of octrees and
the topological aspect of their domains.  Here we define a data type, which we
will call a \emph{point}, that encompasses both octants and their interfaces.
We will then define topological and hierarchical operations for points in
\secref{relations}.  We present these definitions in the context of a single
octree.  Minor modifications will be necessary for a forest of octrees, which
we will discuss in \secref{forest} (see \remref{modify}).

  %
  The boundary of a cube in $\mathbb{R}^d$ has a standard partition into
  lower-dimensional hypercubes, which contains $2^{d-n}\binom{d}{n}$
  $n$-dimensional hypercubes for $0\leq n <d$: since $(2 + 1)^d = \sum_{n=0}^d
  \binom{d}{n}2^{d-n}1^n$, there are $3^d - 1$ in all. We
  index these hypercubes with a set of \emph{boundary indices} $\st B$. The
  \emph{boundary domain} $\dom_{\ix b}(\oc o)$, $\ix b\in \st B$, is the
  corresponding hypercube in the boundary of $\dom(\oc o)$.  
%
For $d=3$, $\st B$ is made of eight corner indices $\{\ix c_i\}_{0 \leq i <8}$,
twelve edge indices $\{\ix e_i\}_{0\leq i < 12}$, and six face indices $\{\ix
f_i\}_{0\leq i < 6}$, which are all illustrated in \figref{boundarysets}.  For
convenience, we define one additional index, the \emph{volume index} $v_0$
corresponding to the volume of an octant, which defines an alias of an
octant's domain, $\dom_{\ix v_0}(\oc o) := \dom(\oc o)$.
%


  A \emph{point} is a tuple $(\oc o,\ix b)$, where $\oc o$ is an octant and
  $\ix b \in \st B \cup \{\ix v_0\}$.  The \emph{domain of a point} $\oc
  c=(\oc o,\ix b)$ is
  \begin{equation}\eqnlab{point} \dom(\oc c):=\dom_{\ix b}(\oc o).
  \end{equation}
  Two points are equal if and only if their domains are equal.
%
  The \emph{dimension} $\dim(\oc c)$ of a point $\oc c$ is the topological
  dimension of its domain. If $\dim(\oc c)=n$, $\oc c$ is an \emph{$n$-point}.
%
  The \emph{level} of a point $\oc c$ is the minimum refinement level of all
  octants in the vicinity of $\oc c$ that appear in a point tuple equal to
  $\oc c$,
  \begin{equation}\eqnlab{pointlevel} \level(\oc c):=\min\{\oc o.\ix
    l:\exists\ix b\in\st B \cup \{\ix v_0\}, \ (\oc o,\ix b) = \oc c\}.
  \end{equation}
%
If $\dim(\oc c)>0$, there is no need to use the minimum refinement level in
the definition of $\level(\oc c)$, because all octants that have $\oc c$ as a
boundary point have the same refinement level: $(\oc c = (\oc o, \ix b))
\Leftrightarrow (\level(\oc c) = \oc o.\ix l)$.  A $0$-point, however, may be
a corner point for octants with different refinement levels: by choosing the
minimum we make a $0$-point's level match that of the biggest neighboring
octant.

\begin{remark}\upshape
  We consider octants to be points: when an operation on a point is applied to
  an octant $\oc o$, one should understand $(\oc o,\ix v_0)$.
\end{remark}


\subsection{Hierarchical and topological relationships}
\seclab{relations}

Here we define the hierarchical and topological relationships used in our
algorithms and proofs below.  We show what these relationships look like in
\tabref{sets}.

\begin{table}
  \caption{%
    For $d=3$, illustrations of the boundary sets $(\bound(\oc c))$, children
    $(\child(\oc c))$, child partition sets $(\part(\oc c))$, and support sets
    $(\supp(\oc c))$ of octants and lower-dimensional $n$-points.  The closure
    set, not illustrated, is the union of the point with its boundary set.
    For a $0$-point, the atomic support set $\asupp(\oc c)$ looks like the
    support set, only scaled down.
  }
  \tablab{sets}
  \tikzset{face1/.style={gray}}
  \tikzset{face2/.style={gray!67}}
  \tikzset{face3/.style={gray!33}}
  \tikzset{corner/.style={ circle
                         , fill=black
                         , inner sep=0cm
                         , minimum width=0.1cm
                         }}
  \newcommand{\mytikzbbox}%
  {%
    \path[use as bounding box](-1cm,-1.25cm) (-1cm,1.25cm) (1cm,1.25cm)
    (1cm,-1.25cm);
  }
  \newcommand{\littleoctant}%
  {%
      \fill[face1] (-.4cm,-.4cm) rectangle (-.1cm,-.1cm);
      \draw[draw=black] (-.4cm,-.4cm) rectangle (-.1cm,-.1cm);
      \fill[face2] (-.1cm,-.4cm) -- (-.025cm,-.3cm) --
                        (-.025cm,0cm) -- (-.1cm,-.1cm) -- cycle;
      \draw[draw=black] (-.1cm,-.4cm) -- (-.025cm,-.3cm) --
                        (-.025cm,0cm) -- (-.1cm,-.1cm) -- cycle;
      \fill[face3] (-.4cm,-.1cm) -- (-.1cm,-.1cm) --
                        (-.025cm,0cm) -- (-.325cm,0cm) -- cycle;%
      \draw[draw=black] (-.4cm,-.1cm) -- (-.1cm,-.1cm) --
                        (-.025cm,0cm) -- (-.325cm,0cm) -- cycle;%
  }
  \newcommand{\midoctant}%
  {%
      \fill[face1] (-.9cm,-.4cm) rectangle (-.1cm,.4cm);
      \draw[draw=black] (-.9cm,-.4cm) rectangle (-.1cm,.4cm);
      \fill[face2] (-.1cm,-.4cm) -- (.1cm,-0.133cm) --
                        (.1cm,.667cm) -- (-.1cm,.4cm) -- cycle;
      \draw[draw=black] (-.1cm,-.4cm) -- (.1cm,-0.133cm) --
                        (.1cm,.667cm) -- (-.1cm,.4cm) -- cycle;
      \fill[face3] (-.9cm,.4cm) -- (-.1cm,.4cm) --
                        (.1cm,.667cm) -- (-.7cm,.667cm) -- cycle;%
      \draw[draw=black] (-.9cm,.4cm) -- (-.1cm,.4cm) --
                        (.1cm,.667cm) -- (-.7cm,.667cm) -- cycle;%
  }
  \newcommand{\bigoctant}%
  {%
      \fill[face1] (-.5cm,-.5cm) rectangle (.5cm,.5cm);
      \draw[draw=black] (-.5cm,-.5cm) rectangle (.5cm,.5cm);
      \fill[face2] (.5cm,-.5cm) -- (.75cm,-0.167cm) --
                        (.75cm,.833cm) -- (.5cm,.5cm) -- cycle;
      \draw[draw=black] (.5cm,-.5cm) -- (.75cm,-0.167cm) --
                        (.75cm,.833cm) -- (.5cm,.5cm) -- cycle;
      \fill[face3] (-.5cm,.5cm) -- (.5cm,.5cm) --
                        (.75cm,.833cm) -- (-0.25cm,.833cm) -- cycle;%
      \draw[draw=black] (-.5cm,.5cm) -- (.5cm,.5cm) --
                        (.75cm,.833cm) -- (-0.25cm,.833cm) -- cycle;%
  }
  \begin{center}
  \begin{tabular}[c]{|c|>{\centering\arraybackslash}m{2cm}|%
                        >{\centering\arraybackslash}m{2cm}|%
                        >{\centering\arraybackslash}m{2cm}|%
                        >{\centering\arraybackslash}m{2cm}|%
                        >{\centering\arraybackslash}m{2cm}|}
    \hline
    $n$ & $\oc c$ & $\bound(\oc c)$ & $\child(\oc c)$ & $\part(\oc c)$ &
    $\supp(\oc c)$
    \\ \hline
    $0$ &
    \begin{tikzpicture}
      \mytikzbbox
      \node[corner] {} node [below] {(corner)};
    \end{tikzpicture}
    &
    $\emptyset$
    &
    $\emptyset$
    &
    $\emptyset$
    &
    \begin{tikzpicture}
      \mytikzbbox

      \begin{scope}[shift={(0.025cm,-.467cm)}]
        \midoctant
      \end{scope}
      \begin{scope}[shift={(1.025cm,-.467cm)}]
        \midoctant
      \end{scope}
      \begin{scope}[shift={(-0.225cm,-0.8cm)}]
        \midoctant
      \end{scope}
      \begin{scope}[shift={(0.775cm,-0.8cm)}]
        \midoctant
      \end{scope}
      \begin{scope}[shift={(0.025cm,.533cm)}]
        \midoctant
      \end{scope}
      \begin{scope}[shift={(1.025cm,.533cm)}]
        \midoctant
      \end{scope}
      \begin{scope}[shift={(-0.225cm,0.2cm)}]
        \midoctant
      \end{scope}
      \begin{scope}[shift={(0.775cm,0.2cm)}]
        \midoctant
      \end{scope}
    \end{tikzpicture}
    \\
    \hline
    $1$ &
    \begin{tikzpicture}
      \mytikzbbox
      \draw (0cm,0.5cm) -- (0cm,-0.5cm) node [below] {(edge)};
    \end{tikzpicture}
    &
    \begin{tikzpicture}
      \mytikzbbox
      \path (0cm,0.5cm) node [corner] {} (0cm,-0.5cm) node [corner] {};
    \end{tikzpicture}
    &
    \begin{tikzpicture}
      \mytikzbbox
      \draw (0cm,0.4cm) -- (0cm,0.1cm);
      \draw (0cm,-0.1cm) -- (0cm,-0.4cm);
    \end{tikzpicture}
    &
    \begin{tikzpicture}
      \mytikzbbox
      \draw (0cm,0.4cm) -- (0cm,0.1cm);
      \draw (0cm,-0.1cm) -- (0cm,-0.4cm);
      \path (0cm,0cm) node [corner] {};
    \end{tikzpicture}
    &
    \begin{tikzpicture}
      \mytikzbbox
      \begin{scope}[shift={(0.025cm,0.033cm)}]
        \midoctant
      \end{scope}
      \begin{scope}[shift={(1.025cm,0.033cm)}]
        \midoctant
      \end{scope}
      \begin{scope}[shift={(-0.225cm,-0.3cm)}]
        \midoctant
      \end{scope}
      \begin{scope}[shift={(0.775cm,-0.3cm)}]
        \midoctant
      \end{scope}
    \end{tikzpicture}
    \\
    \hline
    $2$ &
    \begin{tikzpicture}
      \mytikzbbox
      \fill[fill=gray] (-.5cm,-.5cm) rectangle (.5cm,.5cm);
      \draw[draw=black] (-.5cm,-.5cm) rectangle (.5cm,.5cm);
      \path (0cm,-.5cm) node [below] {(face)};
    \end{tikzpicture}
    &
    \begin{tikzpicture}
      \mytikzbbox
      \draw (-0.5cm,0.4cm) -- (-0.5cm,-0.4cm);
      \draw (0.5cm,0.4cm) -- (0.5cm,-0.4cm);
      \draw (-0.4cm,-0.5cm) -- (0.4cm,-0.5cm);
      \draw (-0.4cm,0.5cm) -- (0.4cm,0.5cm);
      \path (-0.5cm,-0.5cm) node [corner] {};
      \path (0.5cm,-0.5cm) node [corner] {};
      \path (-0.5cm,0.5cm) node [corner] {};
      \path (0.5cm,0.5cm) node [corner] {};
    \end{tikzpicture}
    &
    \begin{tikzpicture}
      \mytikzbbox
      \fill[fill=gray] (-.4cm,-.4cm) rectangle (-.1cm,-.1cm);
      \draw[draw=black] (-.4cm,-.4cm) rectangle (-.1cm,-.1cm);
      \fill[fill=gray] (0.1cm,-.4cm) rectangle (0.4cm,-.1cm);
      \draw[draw=black] (0.1cm,-.4cm) rectangle (0.4cm,-.1cm);
      \fill[fill=gray] (-.4cm,0.1cm) rectangle (-.1cm,0.4cm);
      \draw[draw=black] (-.4cm,0.1cm) rectangle (-.1cm,0.4cm);
      \fill[fill=gray] (0.1cm,0.1cm) rectangle (0.4cm,0.4cm);
      \draw[draw=black] (0.1cm,0.1cm) rectangle (0.4cm,0.4cm);
    \end{tikzpicture}
    &
    \begin{tikzpicture}
      \mytikzbbox
      \fill[fill=gray] (-.4cm,-.4cm) rectangle (-.1cm,-.1cm);
      \draw[draw=black] (-.4cm,-.4cm) rectangle (-.1cm,-.1cm);
      \fill[fill=gray] (0.1cm,-.4cm) rectangle (0.4cm,-.1cm);
      \draw[draw=black] (0.1cm,-.4cm) rectangle (0.4cm,-.1cm);
      \fill[fill=gray] (-.4cm,0.1cm) rectangle (-.1cm,0.4cm);
      \draw[draw=black] (-.4cm,0.1cm) rectangle (-.1cm,0.4cm);
      \fill[fill=gray] (0.1cm,0.1cm) rectangle (0.4cm,0.4cm);
      \draw[draw=black] (0.1cm,0.1cm) rectangle (0.4cm,0.4cm);
      \draw (0cm,0.4cm) -- (0cm,0.1cm);
      \draw (0cm,-0.1cm) -- (0cm,-0.4cm);
      \draw (0.4cm,0cm) -- (0.1cm,0cm);
      \draw (-0.1cm,0cm) -- (-0.4cm,0cm);
      \path (0cm,0cm) node [corner] {};
    \end{tikzpicture}
    &
    \begin{tikzpicture}
      \mytikzbbox
      \begin{scope}[shift={(0.525cm,0.033cm)}]
        \midoctant
      \end{scope}
      \begin{scope}[shift={(0.275cm,-0.3cm)}]
        \midoctant
      \end{scope}
    \end{tikzpicture}
    \\
    \hline
    $3$ &
    \begin{tikzpicture}
      \mytikzbbox
      \begin{scope}[shift={(-.125cm,-.167cm)}]
        \bigoctant
      \end{scope}
      \path (0cm,-.667cm) node [below] {(octant)};
    \end{tikzpicture}
    &
    \begin{tikzpicture}
      \mytikzbbox

      \begin{scope}[scale=1.5,shift={(-.125cm,-.167cm)}]
        \path (-.25cm,-.167cm) node [corner] {};
        \path (0.75cm,-.167cm) node [corner] {};
        \path (-.25cm,0.833cm) node [corner] {};
        \path (0.75cm,0.833cm) node [corner] {};
        \draw (-.25cm,0.733cm) -- (-.25cm,-0.067cm);
        \draw (0.75cm,0.733cm) -- (0.75cm,-0.067cm);
        \draw (-0.15cm,-.167cm) -- (0.65cm,-.167cm);
        \draw (-0.15cm,0.833cm) -- (0.65cm,0.833cm);
        \fill[face1] (-.15cm,-.076cm) rectangle (.65cm,.733cm);
        \draw[draw=black] (-.15cm,-.076cm) rectangle (.65cm,.733cm);

        \draw (-.475cm,-0.467cm) -- (-.275cm,-0.2cm);
        \draw (0.525cm,-0.467cm) -- (0.725cm,-0.2cm);
        \fill[face3] (-.375cm,-.467cm) -- (.425cm,-.467cm) --
        (.625cm,-.2cm) -- (-0.175cm,-.2cm) -- cycle;
        \draw[draw=black] (-.375cm,-.467cm) -- (.425cm,-.467cm) --
        (.625cm,-.2cm) -- (-0.175cm,-.2cm) -- cycle;

        \fill[face2] (-.475cm,-.367cm) -- (-.275cm,-0.1cm) --
        (-.275cm,.7cm) -- (-.475cm,.433cm) -- cycle;
        \draw[draw=black] (-.475cm,-.367cm) -- (-.275cm,-0.1cm) --
        (-.275cm,.7cm) -- (-.475cm,.433cm) -- cycle;

        \fill[face2] (.525cm,-.367cm) -- (.725cm,-0.1cm) --
        (.725cm,.7cm) -- (.525cm,.433cm) -- cycle;
        \draw[draw=black] (.525cm,-.367cm) -- (.725cm,-0.1cm) --
        (.725cm,.7cm) -- (.525cm,.433cm) -- cycle;

        \draw (0.525cm,0.533cm) -- (0.725cm,0.8cm);
        \draw (-0.475cm,0.533cm) -- (-.275cm,0.8cm);
        \fill[face3] (-.375cm,.533cm) -- (.425cm,.533cm) --
        (.625cm,.8cm) -- (-0.175cm,.8cm) -- cycle;
        \draw[draw=black] (-.375cm,.533cm) -- (.425cm,.533cm) --
        (.625cm,.8cm) -- (-0.175cm,.8cm) -- cycle;

        \path (-.5cm,-.5cm) node [corner] {};
        \path (0.5cm,-.5cm) node [corner] {};
        \path (-.5cm,0.5cm) node [corner] {};
        \path (0.5cm,0.5cm) node [corner] {};
        \draw (-0.5cm,0.4cm) -- (-0.5cm,-0.4cm);
        \draw (0.5cm,0.4cm) -- (0.5cm,-0.4cm);
        \draw (-0.4cm,-0.5cm) -- (0.4cm,-0.5cm);
        \draw (-0.4cm,0.5cm) -- (0.4cm,0.5cm);
        \fill[face1] (-.4cm,-.4cm) rectangle (.4cm,.4cm);
        \draw[draw=black] (-.4cm,-.4cm) rectangle (.4cm,.4cm);

      \end{scope}
    \end{tikzpicture}
    &
    \begin{tikzpicture}
      \mytikzbbox

      \begin{scope}[scale=1.5]
        \begin{scope}[shift={(0.025cm,0.033cm)}]
          \littleoctant
        \end{scope}
        \begin{scope}[shift={(0.525cm,0.033cm)}]
          \littleoctant
        \end{scope}
        \begin{scope}[shift={(0.025cm,0.533cm)}]
          \littleoctant
        \end{scope}
        \begin{scope}[shift={(0.525cm,0.533cm)}]
          \littleoctant
        \end{scope}
        \begin{scope}[shift={(-.1cm,-.133cm)}]
          \littleoctant
        \end{scope}
        \begin{scope}[shift={(.4cm,-.133cm)}]
          \littleoctant
        \end{scope}
        \begin{scope}[shift={(-.1cm,.367cm)}]
          \littleoctant
        \end{scope}
        \begin{scope}[shift={(.4cm,.367cm)}]
          \littleoctant
        \end{scope}
      \end{scope}
    \end{tikzpicture}
    &
    \begin{tikzpicture}
      \mytikzbbox

      \begin{scope}[scale=2.0]
        \begin{scope}[shift={(0.025cm,0.033cm)}]
          \littleoctant
        \end{scope}
        \begin{scope}[shift={(0.125cm,0.033cm)}]
          \fill[face2] (-.1cm,-.4cm) -- (-.025cm,-.3cm) --
                            (-.025cm,0cm) -- (-.1cm,-.1cm) -- cycle;
          \draw[draw=black] (-.1cm,-.4cm) -- (-.025cm,-.3cm) --
                            (-.025cm,0cm) -- (-.1cm,-.1cm) -- cycle;
        \end{scope}
        \begin{scope}[shift={(0.525cm,0.033cm)}]
          \littleoctant
        \end{scope}
        \begin{scope}[shift={(0.025cm,0.133cm)}]
          \fill[face3] (-.4cm,-.1cm) -- (-.1cm,-.1cm) --
                            (-.025cm,0cm) -- (-.325cm,0cm) -- cycle;%
          \draw[draw=black] (-.4cm,-.1cm) -- (-.1cm,-.1cm) --
                            (-.025cm,0cm) -- (-.325cm,0cm) -- cycle;%
        \end{scope}
        \begin{scope}[shift={(0.125cm,0.133cm)}]
          \draw (-.1cm,-.1cm) -- (-.025cm,0cm);
        \end{scope}
        \begin{scope}[shift={(0.525cm,0.133cm)}]
          \fill[face3] (-.4cm,-.1cm) -- (-.1cm,-.1cm) --
                            (-.025cm,0cm) -- (-.325cm,0cm) -- cycle;%
          \draw[draw=black] (-.4cm,-.1cm) -- (-.1cm,-.1cm) --
                            (-.025cm,0cm) -- (-.325cm,0cm) -- cycle;%
        \end{scope}
        \begin{scope}[shift={(0.025cm,0.533cm)}]
          \littleoctant
        \end{scope}
        \begin{scope}[shift={(0.125cm,0.533cm)}]
          \fill[face2] (-.1cm,-.4cm) -- (-.025cm,-.3cm) --
                            (-.025cm,0cm) -- (-.1cm,-.1cm) -- cycle;
          \draw[draw=black] (-.1cm,-.4cm) -- (-.025cm,-.3cm) --
                            (-.025cm,0cm) -- (-.1cm,-.1cm) -- cycle;
        \end{scope}
        \begin{scope}[shift={(0.525cm,0.533cm)}]
          \littleoctant
        \end{scope}
        \begin{scope}
          \fill[fill=gray] (-.4cm,-.4cm) rectangle (-.1cm,-.1cm);
          \draw[draw=black] (-.4cm,-.4cm) rectangle (-.1cm,-.1cm);
          \fill[fill=gray] (0.1cm,-.4cm) rectangle (0.4cm,-.1cm);
          \draw[draw=black] (0.1cm,-.4cm) rectangle (0.4cm,-.1cm);
          \fill[fill=gray] (-.4cm,0.1cm) rectangle (-.1cm,0.4cm);
          \draw[draw=black] (-.4cm,0.1cm) rectangle (-.1cm,0.4cm);
          \fill[fill=gray] (0.1cm,0.1cm) rectangle (0.4cm,0.4cm);
          \draw[draw=black] (0.1cm,0.1cm) rectangle (0.4cm,0.4cm);
          \draw (0cm,0.4cm) -- (0cm,0.1cm);
          \draw (0cm,-0.1cm) -- (0cm,-0.4cm);
          \draw (0.4cm,0cm) -- (0.1cm,0cm);
          \draw (-0.1cm,0cm) -- (-0.4cm,0cm);
          \path (0cm,0cm) node [corner] {};
        \end{scope}
        \begin{scope}[shift={(-.1cm,-.133cm)}]
          \littleoctant
        \end{scope}
        \begin{scope}[shift={(0cm,-.133cm)}]
          \fill[face2] (-.1cm,-.4cm) -- (-.025cm,-.3cm) --
                            (-.025cm,0cm) -- (-.1cm,-.1cm) -- cycle;
          \draw[draw=black] (-.1cm,-.4cm) -- (-.025cm,-.3cm) --
                            (-.025cm,0cm) -- (-.1cm,-.1cm) -- cycle;
        \end{scope}
        \begin{scope}[shift={(.4cm,-.133cm)}]
          \littleoctant
        \end{scope}
        \begin{scope}[shift={(-.1cm,-.033cm)}]
          \fill[face3] (-.4cm,-.1cm) -- (-.1cm,-.1cm) --
                            (-.025cm,0cm) -- (-.325cm,0cm) -- cycle;%
          \draw[draw=black] (-.4cm,-.1cm) -- (-.1cm,-.1cm) --
                            (-.025cm,0cm) -- (-.325cm,0cm) -- cycle;%
        \end{scope}
        \begin{scope}[shift={(0.cm,-0.033cm)}]
          \draw (-.1cm,-.1cm) -- (-.025cm,0cm);
        \end{scope}
        \begin{scope}[shift={(0.4cm,-.033cm)}]
          \fill[face3] (-.4cm,-.1cm) -- (-.1cm,-.1cm) --
                            (-.025cm,0cm) -- (-.325cm,0cm) -- cycle;%
          \draw[draw=black] (-.4cm,-.1cm) -- (-.1cm,-.1cm) --
                            (-.025cm,0cm) -- (-.325cm,0cm) -- cycle;%
        \end{scope}
        \begin{scope}[shift={(-.1cm,.367cm)}]
          \littleoctant
        \end{scope}
        \begin{scope}[shift={(0cm,.367cm)}]
          \fill[face2] (-.1cm,-.4cm) -- (-.025cm,-.3cm) --
                            (-.025cm,0cm) -- (-.1cm,-.1cm) -- cycle;
          \draw[draw=black] (-.1cm,-.4cm) -- (-.025cm,-.3cm) --
                            (-.025cm,0cm) -- (-.1cm,-.1cm) -- cycle;
        \end{scope}
        \begin{scope}[shift={(.4cm,.367cm)}]
          \littleoctant
        \end{scope}
      \end{scope}
    \end{tikzpicture}
    &
    \begin{tikzpicture}
      \mytikzbbox
      \begin{scope}[shift={(-.125cm,-.167cm)}]
        \bigoctant
      \end{scope}
    \end{tikzpicture}
    \\
    \hline
  \end{tabular}
\end{center}
\end{table}

The hierarchical relationships between points are determined by set inclusion
of their domains.
%
  The \emph{descendants} of a point $\oc c$ are all of the points with the
  same dimension whose domains are contained in the domain of $\oc c$,
  \begin{equation}\eqnlab{pointdesc}
    \desc(\oc c):=\{\oc e:\dim(\oc e)=\dim(\oc c),\ \dom(\oc e) \subseteq
    \dom(\oc c)\}.
  \end{equation}
  The \emph{children} of a point $\oc c$ are descendants that are
  more refined than $\oc c$ by one level,
  \begin{equation}\eqnlab{pointchildren}
    \child(\oc c):=\{\oc h:\oc h\in\desc(\oc c),\ \level(\oc h) = \level(\oc
    c)+1\}.
  \end{equation}

The requirement that an octant $\oc o$'s coordinates must be multiples of
$2^{\lmax- \oc o.\ix l}$ has the consequence that the domains of two distinct
points with the same level do not overlap, and that every point's domain is
tiled by the domains of its children
%
  (a collection $\mathcal{U}$ of subsets of a set $S$ in a topological space
  tiles $S$ if
  %
$
    S \subseteq \bigcup_{U\in \mathcal{U}} \overline{U}
    \text{ and }
    (U,V\in\mathcal{U},U\neq V) \Rightarrow (U\cap V = \emptyset)
    $).

\begin{propn}
  If $\dim(\oc c)>0$, then $|\child(\oc c)| = 2^{\dim(\oc c)}$ and
  $\dom(\child(\oc c))$ tiles $\dom(\oc c)$.
\end{propn}

A point's domain is tiled by its children's domains, but it is not partitioned
by them, as they are open sets.  To define a partition, we must add
lower-dimensional points between them.
%
  The \emph{child partition} is the set of all points whose domains are
  contained in $\dom(\oc c)$ and whose levels are greater by one,
  \begin{equation}\eqnlab{pointpart}
    \part(\oc c) :=
    \{\oc h:\level(\oc h)=\level(\oc c) + 1,
    \ \dom(\oc h)\subset\dom(\oc c)\}.
  \end{equation}

\begin{propn}
  If $\dim(\oc c)>0$, then $|\part(\oc c)| = 3^{\dim(\oc c)}$ and $\part(\oc
  c)$ defines a partition, $\dom(\oc c) = \bigsqcup\dom(\part(\oc c))$.
\end{propn}

The two basic topological sets we need for a point $\oc c$ are the
lower-dimensional points that surround $\oc c$---its boundary points---and the
octants that surround $\oc c$---its support octants.  To define boundary
points, we first define closure points.
%
  The \emph{closure set of an octant} $\oc o$ is the set of all points in which $\oc o$
  may appear in a point tuple,
  \begin{equation}\eqnlab{octclos} \clos(\oc o) := \{(\oc o,\ix b):\ix b\in\st
    B\cup\{\ix v_0\}\}.  \end{equation}
  The \emph{closure set of a point} $\oc c$ is the intersection of all octant
  closure sets containing $\oc c$,
  \begin{equation}\eqnlab{pointclos} \clos(\oc c):=\bigcap \{\clos(\oc o):\oc
    c\in \clos(\oc o)\}.  \end{equation}
%
  The \emph{boundary set} of a point $\oc c$ is its closure less itself,
  \begin{equation}\eqnlab{pointbound} \bound(\oc c):= \clos(\oc
    c)\backslash\{\oc c\}.  \end{equation}

\begin{propn}[Point closure matches $\mathbb{R}^d$ closure]
  \label{prop:closure}
  The domains of points in $\oc c$'s closure set $\clos(\oc c)$
  partition the closure of its domain, $\overline{\dom(\oc c)} = \bigsqcup
  \dom(\clos(\oc c))$.
\end{propn}

  The \emph{support set} of a point $\oc c$ is the set of octants with the
  same refinement level as $\oc c$ whose closures include $\oc c$,
  \begin{equation}\eqnlab{pointsupp}
    \supp(\oc c):= \{\oc o:\oc c\in \clos(\oc o),\oc o.\ix l=\level(\oc c)\}.
  \end{equation}

\begin{propn}[$\mathbb{R}^d$ intersection implies support intersection]
  \label{prop:suppisect}
  If $\oc o$ is an octant, $\oc c$ is a point, and $\overline{\dom(\oc o)}
  \cap \dom(\oc c) \neq \emptyset$, then there exists $\oc s\in\supp(\oc c)$
  such that $\oc s\in\desc(\oc o)$ or $\oc o\in\desc(\oc s)$.
\end{propn}

The following proposition shows the duality between $\clos(\cdot)$ and
$\supp(\cdot)$.

\begin{propn}\label{prop:duality1}
  If $\dim(\oc c)>0$, then $(\oc o\in \supp(\oc c)) \Leftrightarrow
  (\oc c \in \clos(\oc o))$.
\end{propn}

For $0$-points, this duality does not hold because a $0$-point can be in the
closure set of an octant with a more refined level.  In fact, every $0$-point
is in the closure of an atom.
%
  The \emph{atomic support set} of a $0$-point $\oc c$ is the set of atoms
  whose closures include $\oc c$,
  \begin{equation}\eqnlab{pointasupp}
    \asupp(\oc c):= \{\oc a:\oc c\in \clos(\oc a),\oc a.\ix l=\lmax\}.
  \end{equation}
%
The support and atomic support sets of a $0$-point contain and
are contained in all neighboring octants, respectively.

\begin{propn}\label{prop:duality2}
  If $\dim(\oc c)=0$, $\oc o$ is a octant, and $\oc c\in \clos(\oc o)$, then
  there are $\oc a\in\asupp(\oc c)$ and $\oc s\in\supp(\oc c)$ such that $\oc
  a\in\desc(\oc o)$ and $\oc o\in \desc(\oc s)$.
\end{propn}

\subsection{Forests of octrees}
\seclab{forest}

A forest of quadtrees ($d=2$) or octrees ($d=3$) is a mesh of a
$d$-dimensional domain $\Omega$ with two layers, a macro layer and a micro
layer.
%
  The \emph{macro layer} is a geometrically conformal mesh\footnote{%
    By ``geometrically conformal mesh'' we mean that $\{T^t\}_{0\leq t < K}$
    are the cells of a CW complex (a generalization of simplicial complex to
    other polytopes, where C stands for closure-finite and W for weak
    topology; see
    e.g.\ \cite[Chapter 10]{May99}): informally, each $T^{\ix t}$ is open,
    $\{T^{\ix t}\}_{0\leq t < K}$ tiles $\Omega$, and if the intersection
    $\overline{T^{\ix s}} \cap \overline{T^{\ix t}}$ has dimension $(d-1)$,
    then it is equal to a whole face of $T^{\ix s}$ and a whole face of
    $T^{\ix t}$.
  } %
  of $K$ mapped cells (quadrilaterals ($d=2$) or hexahedra ($d=3$)),
  \begin{equation}\eqnlab{macro}
    \cT:=\{(T^{\ix t},\varphi^{\ix t})\}_{0\leq t < K},
  \end{equation}
  where each $T^{\ix t}$ has an associated map $\varphi^{\ix t}:
  \overline{\dom(\Root(\ix t))} \to \overline{T^{\ix t}}$, which is a
  continuous bijection between the domain of the root octant and $T^{\ix t}$.
%
  We define the \emph{mapped domain of a point}
  $\oc c=(\oc o,\ix b)$ by its image under the map for $\oc o$'s tree index,
  \begin{equation}\eqnlab{mapped}
    \Dom(\oc c):=\varphi^{o.\ix t}(\dom_{\ix b}(\oc o)).
  \end{equation}

\begin{remark}[Modifications to definitions for forests]\upshape
  \remlab{modify}
  In \secref{octant} we defined points in the context of a single octree.  In
  the forest-of-octrees context, we consider two points equal if their mapped
  domains are equal.  If one substitutes mapped domains for unmapped domains
  in the definitions and propositions in \secref{relations}, they hold in the
  forest-of-octrees context.  If a point $\oc c$'s mapped domain is on the
  boundary between macro-layer cells, then $\oc c$'s support set $\supp(\oc
  c)$ no longer has the regular shape shown in \tabref{sets}, but depends on
  the macro layer topology.   A face on the boundary of $\Omega$, for example,
  has only one support octant.  We emphasize that the sets defined in
  \secref{relations} do not depend on the exact nature of the maps
  $\{\varphi^{\ix t}\}_{0\leq t < K}$, but can be constructed, in time
  proportional to their sizes, from the point data of their arguments and the
  mesh topology of $\cT$, i.e., which cells are neighbors, which of their
  faces correspond, and how those faces are oriented relative to each other.
  These issues are covered in \cite[Section 2.2]{BursteddeWilcoxGhattas11}.
\end{remark}

  For each $0 \leq t < K$, the \emph{tree-$t$ leaves} $\treeoctants \subset
  \desc(\Root(t))$ are a set of $N^{(t)}$ octants whose 
  domains tile $\dom(\Root(\ix t))$.  The \emph{micro layer} $\cO$ is the
  union of these sets, and its size is $N$,
  \begin{equation}\eqnlab{micro}
    \cO:= \bigsqcup_{0\leq t < K} \cO^{\ix t},
    \ N := |\cO| = \sum_{0\leq t < K}N^{(t)}.
  \end{equation}
%
The mapped domains of the micro layer octants tile $\Omega$, but this tiling
is not a geometrically conforming mesh: when neighboring octants have
different levels, their faces (and edges if $d=3$) do not conform to each
other.  If neighboring octants differ by at most one level, the forest is said
to satisfy a 2:1 balance condition
\cite{SampathAdavaniSundarEtAl08,IsaacBursteddeGhattas12}.

We call $\treeoctants$ the leaves of octree $\ix t$ because one could build a
tree structure, starting with $\Root(\ix t)$ and using the $\child(\cdot)$
operation, whose leaves would be $\treeoctants$.  The \pforest library does
not store this tree structure in memory.  Storing just $\treeoctants$ is an
approach known as a linear octree representation \cite{SundarSampathBiros08}.

\subsection{Distributed forests of octrees}
\seclab{distributed}

In \pforest, the macro layer $\cT$ is static and replicated on each process,
while the micro layer $\cO$ is dynamic---it may be adaptively refined,
coarsened, and repartitioned frequently over the life of a forest---and
distributed, with each process owning a distinct subset of leaves.  We
describe the distribution method here, and illustrate it in \figref{forest}.

\begin{figure}
  \centering
  \usetikzlibrary{shapes.geometric}
\usetikzlibrary{positioning}
\tikzstyle{intnode}=[draw=black,fill=white]
\tikzstyle{intnode0}=[circle,draw=black,fill=white,inner sep=0.09cm]
\tikzstyle{intnode1}=[diamond,draw=black,fill=white,inner sep=0.07cm]
\tikzstyle{leaffirst}=[draw=black,fill=red!15]
\tikzstyle{leaffirst0}=[circle,draw=black,fill=red!15,inner sep=0.09cm]
\tikzstyle{atomfirst0}=[circle,solid,draw=black,fill=red,inner sep=0.06cm]
\tikzstyle{leafsecond}=[draw=black,fill=green!15]
\tikzstyle{leafsecond0}=[circle,draw=black,fill=green!15,inner sep=0.09cm]
\tikzstyle{atomsecond0}=[circle,solid,draw=black,fill=green,inner sep=0.06cm]
\tikzstyle{leafsecond1}=[diamond,draw=black,fill=green!15,inner sep=0.07cm]
\tikzstyle{leafthird}=[draw=black,fill=blue!15]
\tikzstyle{leafthird1}=[diamond,draw=black,fill=blue!15,inner sep=0.07cm]
\tikzstyle{atomthird1}=[diamond,solid,draw=black,fill=blue,inner sep=0.05cm]

\tikzstyle{zcurve}=[thick,draw=black]
\tikzstyle{connect}=[densely dashed,draw=black]
\tikzstyle{partition}=[thick,dotted,draw=black]

\tikzstyle{first}=[very thin,draw=gray,fill=red!15]
\tikzstyle{firstatom}=[very thin,draw=gray,fill=red]
\tikzstyle{second}=[very thin,draw=gray,fill=green!15]
\tikzstyle{secondatom}=[very thin,draw=gray,fill=green]
\tikzstyle{third}=[very thin,draw=gray,fill=blue!15]
\tikzstyle{thirdatom}=[very thin,draw=gray,fill=blue]

\tikzstyle{edge from parent}=[draw,gray]
\tikzstyle{level 1}=[level distance=1cm, sibling distance=11mm]
\tikzstyle{level 2}=[level distance=1cm, sibling distance=5mm]
\tikzstyle{level 3}=[level distance=1cm, sibling distance=5mm]
\tikzstyle{atombranch}=[edge from parent/.style={draw,gray,dashed}]

\begin{tikzpicture}[scale=.65625]

  \draw (0,0)
  node (top0) [left] {} node [above] (root0) {$\mathop{\mathrm{root}}(s)$} node[intnode0,below=0.0cm of root0] (k0e0) {} 
  child { node[leaffirst0] (k0e00) {}
    child [atombranch] {
      child {
        child {
          node[atomfirst0] (atom0) {} node[right] {$\oc f_{\pr p}$}
        }
      }
    }
  }
  child {
    node[intnode0] (k0e01) {}
    child { node[leaffirst0] (k0e010) {} }
    child { node[leaffirst0] (k0e011) {} }
    child { node[leaffirst0] (k0e012) {} }
    child { node[leaffirst0] (k0e013) {} }
  }
  child { node[leaffirst0] (k0e02) {} }
  child {
    node[intnode0] (k0e03) {}
    child { node[leaffirst0] (k0e030) {} }
    child { node[leafsecond0] (k0e031) {}
      child [atombranch] {
        child {
          node[atomsecond0] (atom1) {} node[right] {$\oc f_{\pr q}$}
        }
      }
    }
    child { node[leafsecond0] (k0e032) {} }
    child {
      node[intnode0] (k0e033) {}
      child { node[leafsecond0] (k0e0330) {} }
      child { node[leafsecond0] (k0e0331) {} }
      child { node[leafsecond0] (k0e0332) {} }
      child { node[leafsecond0] (k0e0333) {} }
    }
  }
  ;

  \draw(4.5,0)
  node [right=-1.2cm] (top1) {} node [above] (root1) {$\mathop{\mathrm{root}}(t)$} node[intnode1,below=0cm of root1] (k1e0) {}
  child { node[leafsecond1] (k1e00) {} }
  child { node[leafsecond1] (k1e01) {} }
  child {
    node[intnode1] (k1e02) {}
    child { node[leafthird1] (k1e020) {}
      child [atombranch] {
        child {
          node[atomthird1] (atom2) {} node[right] {$\oc f_{\pr r}$}
        }
      }
    }
    child { node[leafthird1] (k1e021) {} }
    child {
      node[intnode1] (k1e031) {}
      child { node[leafthird1] (k1e0220) {} }
      child { node[leafthird1] (k1e0221) {} }
      child { node[leafthird1] (k1e0222) {} }
      child { node[leafthird1] (k1e0223) {} }
    }
    child { node[leafthird1] (k1e023) {} }
  }
  child { node[leafthird1] (k1e03) {} }
  ;

  \draw[zcurve,->]
  (k0e00.center) --
  (k0e010.center) --
  (k0e011.center) --
  (k0e012.center) --
  (k0e013.center) --
  (k0e02.center) --
  (k0e030.center) --
  (k0e031.center) --
  (k0e032.center) --
  (k0e0330.center) --
  (k0e0331.center) --
  (k0e0332.center) --
  (k0e0333.center);

  \draw[connect,<-] (k1e00.center) -- (k0e0333.center);

  \draw[zcurve,->]
  (k1e00.center) --
  (k1e01.center) --
  (k1e020.center) --
  (k1e021.center) --
  (k1e0220.center) --
  (k1e0221.center) --
  (k1e0222.center) --
  (k1e0223.center) --
  (k1e023.center) --
  (k1e03.center);

  \path (k0e02) -- (k0e03) node[midway] (m0) {};
  \path (k0e030) -- (k0e031) node[midway] (m1) {};
  \draw[partition] (k0e0) -- (m0.center) -- (m1.center) -- +(0,-2.);
  \path (k1e01) -- (k1e02) node[midway] (m2) {};
  \node (m3) [node distance=2mm,left of=k1e020] {};
  \draw[partition] (k1e0) -- (m2.center) -- (m3.center) -- +(0,-2.);

  \begin{scope}[yshift=-0.75cm]
    \begin{scope}
      \draw [fill=white,drop shadow] (-2cm,-5cm) rectangle (0.25cm,-4cm);

      \draw (-1.5cm,-4.5cm) node [leaffirst0] (leg0) {} node [right=0.1cm of leg0,leafsecond0] (leg1) {}
      node [right=0.0cm of leg1] (legtext0) {$\strut\mathcal{O}^{s}$};
    \end{scope}

    \begin{scope}[xshift=2.5cm]
      \draw [fill=white,drop shadow] (-2cm,-5cm) rectangle (0.25cm,-4cm);

      \draw (-1.5cm,-4.5cm) node [leafsecond1] (leg2) {} node [right=0.1cm of leg2,leafthird1] (leg3) {}
      node [right=0.0cm of leg3] (legtext1) {$\strut\mathcal{O}^{t}$};
    \end{scope}

    \begin{scope}[xshift=5cm]
      \draw [fill=white,drop shadow] (-2cm,-5cm) rectangle (-0.25cm,-4cm);

      \draw (-1.5cm,-4.5cm) node [leaffirst0] (leg4) {}
      node [right=0.0cm of leg4] (legtext2) {$\strut\mathcal{O}_{\mathfrak{p}}$};
    \end{scope}

    \begin{scope}[xshift=7cm]
      \draw [fill=white,drop shadow] (-2cm,-5cm) rectangle (0.25cm,-4cm);

      \draw (-1.5cm,-4.5cm) node [leafsecond0] (leg5) {} node [right=0.1cm of leg5,leafsecond1] (leg6) {}
      node [right=0.0cm of leg6] (legtext3) {$\strut\mathcal{O}_{\mathfrak{q}}$};
    \end{scope}

    \begin{scope}[xshift=9.5cm]
      \draw [fill=white,drop shadow] (-2cm,-5cm) rectangle (-0.25cm,-4cm);

      \draw (-1.5cm,-4.5cm) node [leafthird1] (leg7) {}
      node [right=0.0cm of leg7] (legtext4) {$\strut\mathcal{O}_{\mathfrak{r}}$};
    \end{scope}

    \begin{scope}[xshift=11.5cm]
      \draw [fill=white,drop shadow] (-2cm,-5cm) rectangle (0.5cm,-4cm);

      \draw (-1.5cm,-4.5cm) node [atomfirst0] (lega0) {}
      node [right=0.1cm of lega0,atomsecond0] (lega1) {}
      node [right=0.1cm of lega1,atomthird1] (lega2) {}
      node [right=0.0cm of lega2] (legatext) {$\strut\arry f$};
    \end{scope}
  \end{scope}

  \begin{scope}[xshift=82mm,yshift=-38mm,scale=0.5]

  \draw[first] (0,4) rectangle +(4,4);
  \draw[first] (4,6) rectangle +(2,2);
  \draw[first] (6,6) rectangle +(2,2);
  \draw[first] (4,4) rectangle +(2,2);
  \draw[first] (6,4) rectangle +(2,2);
  \draw[first] (0,0) rectangle +(4,4);
  \draw[first] (4,2) rectangle +(2,2);
  \draw[firstatom] (0,7.5) rectangle +(0.5,0.5);
  \draw[second] (6,2) rectangle +(2,2);
  \draw[second] (4,0) rectangle +(2,2);
  \draw[second] (6,1) rectangle +(1,1);
  \draw[second] (7,1) rectangle +(1,1);
  \draw[second] (6,0) rectangle +(1,1);
  \draw[second] (7,0) rectangle +(1,1);
  \draw[secondatom] (6,3.5) rectangle +(0.5,0.5);

  \draw[second] (12,0) rectangle +(4,4);
  \draw[second] (12,4) rectangle +(4,4);
  \draw[third] (10,0) rectangle +(2,2);
  \draw[third] (10,2) rectangle +(2,2);
  \draw[third] (9,0) rectangle +(1,1);
  \draw[third] (9,1) rectangle +(1,1);
  \draw[third] (8,0) rectangle +(1,1);
  \draw[third] (8,1) rectangle +(1,1);
  \draw[third] (8,2) rectangle +(2,2);
  \draw[third] (8,4) rectangle +(4,4);
  \draw[thirdatom] (11.5,0) rectangle +(0.5,0.5);

  \draw[zcurve,->]
  (2,6) -- (5,7) -- +(2,0) -- +(0,-2) -- +(2,-2) -- (2,2) -- (5,3) --
  +(2,0) -- +(0,-2) -- (6.5,1.5) -- +(1,0) -- +(0,-1) -- +(1,-1);
  \draw[connect,<-]
  (14,2) -- (7.5,0.5);
  \draw[zcurve,->]
  (14,2) -- (14,6) -- (11,1) -- +(0,2) -- (9.5,0.5) -- +(0,1) -- +(-1,0)
  -- +(-1,1) -- (9,3) -- (10,6);

  \draw[partition] (4,0) -- (4,2) -- (6,2) -- (6,4) -- (8,4);
  \draw[partition] (8,0) -- (8,8);
  \draw[partition] (12,0) -- (12,8);

  \begin{scope}[scale=4]

  \draw (0,0) rectangle (4,2); 

  \draw (2,2) node [above] {$\Omega$};

  \draw (.5,.95) node [rectangle,draw,fill=white,opacity=.5] {$T^{s}$};
  \draw (3,1.4) node [rectangle,draw,fill=white,opacity=.5] {$T^{t}$};

  \draw[->,thick] (0,2) -- +(.6,0) node[above] {$\varphi^{s}(\vec\imath)$};
  \draw[->,thick] (0,2) -- +(0,-.6) node[left] {$\varphi^{s}(\vec\jmath)$};
  \draw[->,thick] (4,0) -- +(0,.6) node[right] {$\varphi^{t}(\vec\imath)$};
  \draw[->,thick] (4,0) -- +(-.6,0) node[below] {$\varphi^{t}(\vec\jmath)$};

  \end{scope}

  \end{scope}

\end{tikzpicture}

  \caption{%
    (Adapted with permission from \cite[Fig. 2.1]{BursteddeWilcoxGhattas11}.)
    A $d=2$ example of the relationship between the implicit tree structure
    (left) and the domain tiling (right) of a forest of octrees with two trees
    $\ix s$ and $\ix t$.  The bijections $\varphi^{\ix s}$ and $\varphi^{\ix
      t}$ map the implicit coordinate systems of the unmapped octant domains
    onto the cells $T^{\ix s}$ and $T^{\ix t}$.  The left-to-right
    traversal of the leaves (black ``zig-zag'' line) demonstrates the total
    order (\algref{comparison}).  In this example the forest is partitioned
    among three processes $\pr p$, $\pr q$ and $\pr r$ by sectioning the
    leaves into $\cO_{\pr p}$, $\cO_{\pr q}$, and $\cO_{\pr r}$.  Color
    conveys this partition, while the node shapes convey the division of the
    leaves into the trees $\cO^{\df{s}}$ and $\cO^{\ix t}$.  The small,
    brightly colored nodes represent the first atoms located in
    each process's subdomains \eqref{eqn:firstatoms}.%
  }%
  \figlab{forest}
\end{figure}

The partitioning of leaves between processes and the layout of leaves in
memory is determined by the total order induced by the comparison operation in
\algref{comparison}.  The comparison of coordinates in line
\ref{coordcomparison} is defined by a space-filling curve; the \pforest
library uses the so-called $z$-ordering which corresponds to the Morton curve
\cite{Morton66}.

\begin{algorithm}
  \caption{ $\oc o \leq \oc r$ (\Octant $\oc o$, \Octant $\oc r$)}
  \alglab{comparison}

  \lIf%
    (\Comment{If $s<t$, all of octree $s$'s leaves come before $t$'s.})
    {$\oc o.\ix t \neq \oc r.\ix t$}%
    {\Return{$(\oc o.\ix t \leq \oc r.\ix t)$}}%

  \lIf%
    (\Comment{Coordinates are ordered by SFC index.})
    {$\oc o.\ix x \neq \oc r.\ix x$}%
    {\Return{$(\oc o.\ix x \leq \oc r.\ix x)$}}%
    \label{coordcomparison}

  \Return{$(\oc o.l \leq \oc r.l)$}%
  \Comment{Ancestors precede descendants (preordering).}
\end{algorithm}%

\begin{remark}\upshape
  When we index a sorted array of octants ($\arry A[i]$), children
  ($\child(\oc o)[i]$), or a support set ($\supp(\oc c)[i]$), we mean the
  $i$th octant with respect to the total order.  
\end{remark}

Using this total order, each process is assigned a contiguous (with respect to
the total order) section of leaves in MPI-rank order.
%
  For each $0 \leq t < K$ and $0 \leq \pr p < P$, the subset of $\treeoctants$
  assigned to process $\pr p$ is in an array $\localtreeoctants$, which has
  size $N^{(t)}_{\pr p}$; these arrays form a partition,
  $\treeoctants=\bigsqcup_{0 \leq \pr p < P}\localtreeoctants$.  The set of all
  leaves assigned to $\pr p$ is
  \begin{equation}
    \eqnlab{localoctants}
    \localoctants := \bigsqcup_{0 \leq t < K} \localtreeoctants;
  \end{equation}
  its size is $N_{\pr p} := |\localoctants| = \sum_{0 \leq t < K} N^{(t)}_{\pr
  p}>0$.%
  \footnote{%
    In \pforest partitions may be empty, but for simplicity we assume
    here that they are all non-empty.%
  } %
%
  %
  The \emph{subdomain} of $\Omega$ tiled by $\Dom(\localtreeoctants)$ is the
  interior of its closure,
  \begin{equation}\eqnlab{subdomains} \Omega_{\pr p}^{\ix t} :=
    \left(\overline{\bigcup\Dom(\localtreeoctants)}\right)^{\circ};
  \end{equation}
  the subdomains $\Omega^{\ix t}$ and $\Omega_{\pr p}$ are analogously
  defined.
%


Because the leaves are partitioned, process $\pr p$ cannot determine (without
communication) if an octant $\oc o$, whose mapped domain $\Dom(\oc o)$ is
outside $\Omega_{\pr p}$, is a leaf.  We do, however, want $\pr p$ to be able
to locate the subdomains that overlap $\Dom(\oc o)$.  To allow this, each
process in \pforest has some information about other processes' subdomains,
which we describe here.  We start from the fact that an atom's mapped domain
is \emph{located} in the subdomain of only one process's subdomain,
%
  %
  \begin{equation}\eqnlab{locate}
    (\locate(\oc a) := \pr q) \Leftrightarrow (\Dom(\oc a) \subseteq \Omega_{\pr
      q}).
  \end{equation}
%
Note that $\pr q=\locate(\oc a)$ does not imply $\oc a\in\cO_{\pr q}$: $\oc a$
could be a descendant of a leaf in $\cO_{\pr q}$.  To test whether $\pr
q=\locate(\oc a)$ for an arbitrary atom $\oc a$, it is only necessary to
precompute a process's \emph{range}: a tuple of its \emph{first atom} $\oc
f_{\pr q}$ and its \emph{last atom} $\oc l_{\pr q}$ with respect to the total
order of octants,
  \begin{equation}\eqnlab{procrange}
    \begin{aligned}
      \oc f_{\pr q} &:= \min \{\oc a:\oc a.\ix l=\lmax,\ \Dom(\oc a)\subseteq
      \Omega_{\pr q}\},
      \\
      \oc l_{\pr q} &:= \max \{\oc a:\oc a.\ix l=\lmax,\ \Dom(\oc a)\subseteq
      \Omega_{\pr q}\},
      \\
      \range(\pr q) &:= (\oc f_{\pr q},\ \oc l_{\pr q}).
    \end{aligned}
  \end{equation}

\begin{propn}
  $(\pr q = \locate(\oc a)) \Leftrightarrow (\oc f_{\pr q} \leq \oc a \leq
  \oc l_{\pr q}).$
\end{propn}

\begin{propn}[A range describes a subdomain]
  $\overline{\Omega_{\pr q}} = \overline{\bigcup\Dom([\oc f_{\pr q}, \oc l_{\pr
      q}])}$, where $[\oc f_{\pr q}, \oc l_{\pr q}]$ is the set of all atoms
  $\oc a$ such that $\oc f_{\pr q} \leq \oc a \leq \oc l_{\pr q}$.
\end{propn}

We also apply the range operator to individual octants:
%
  the \emph{range of an octant} $\oc o$ is a tuple of the first and last
  atoms, $\oc f_{\oc o}$ and $\oc l_{\oc o}$, in its descendants,
  \begin{equation}\eqnlab{octrange}
    \begin{aligned}
      \oc f_{\oc o} &:= \min \{\oc a:\ \oc a.l=\lmax,\ \oc a\in\desc(\oc o)\},\\
      \oc l_{\oc o} &:= \max \{\oc a:\ \oc a.l=\lmax,\ \oc a\in\desc(\oc o)\},\\
      \range(\oc o) &:= (\oc f_{\oc o},\ \oc l_{\oc o}).
    \end{aligned}
  \end{equation}

\begin{propn}
  $(\Dom(\oc o)\subseteq \Omega_{\pr q})
  \Leftrightarrow (\pr q = \locate(\oc f_{\oc o}) = \locate(\oc l_{\oc o}))$.
\end{propn}

To locate atoms, it is not necessary to store both $\oc f_{\pr q}$ and $\oc
l_{\pr q}$: $\oc l_{\pr q}$ can be computed from $\oc f_{\pr q +1}$. In
\pforest, we store a sorted array called the \emph{first-atoms array} $\arry
f$, where
\begin{equation}\eqnlab{firstatoms}
  \arry f[\pr p] := \oc f_{\pr p},\ 0 \leq \pr p < P,
\end{equation}
and $\arry f[P]$ is a phony ``terminal'' octant whose tree index is $K$.  This
array is shown in \figref{forest}.
The first atom $\oc f_{\pr q}$ is independent of the leaves in $\cO_{\pr q}$,
so $\arry f$ is up-to-date even if other processes have refined or coarsened
their leaves.  Using $\arry f$, a process can compute $\locate(\oc a)$ in
$O(\log P)$ time and test $(\pr q = \locate(\oc a)?)$ in $O(1)$ time. 

For the purposes of this paper, we have described all components of a
\emph{distributed forest of octrees}, which is, for process $\pr p$, the
combination of macro layer \eqref{eqn:macro}, local leaves
\eqref{eqn:localoctants},
and the first-atoms array \eqref{eqn:firstatoms},
\begin{equation}\eqnlab{distributed}
  \oc F_{\pr p} := (\st T,\localoctants,\arry f).
\end{equation}

\begin{remark}\upshape
  $\oc F_{\pr p}$ is an assumed argument of the
  algorithms we present.
\end{remark}


\section{Parallel multiple-item search via array splitting}
\seclab{search}

We can optimize the search for a leaf that matches a given condition if we
begin at the root of an octree and recursively descend to all children that
could possibly be a match.  This is a lazy exclusion principle which is
motivated by a practical consideration: Often an over-optimistic approximate
check can be significantly faster than an exact check, which applies to
bounding-box checks in computational geometry or to checking the surrounding
sphere of a nonlinearly warped octant volume in space.

\subsection{Searching in a single octree and in a forest}

We assume that the user has a set of arbitrary matching queries indexed
by $\st Q$: For each $\ix q\in\st Q$, $\mathrm{match}_{\ix q}()$
returns true or false for any given octant $\oc o\in \cO$.
We also pass a boolean parameter to $\mathrm{match}_{\ix q}()$ that
specifies whether $\oc o$ is a leaf or not.  This may be used to execute
over-optimistic, cheap matches for non-leaves and strict matches for leaves.
This approach is more general than searching for a
single leaf that matches each query: Our framework encompasses, for example,
the search for all the leaves that intersect a set of polytopes (indexed by
$\st Q$) embedded in $\Omega$.  More generally, it is entirely legal that one
leaf matches multiple queries, or one query matches multiple leaves.

In \fxn{Search} (\algref{search}), we use recursion and lazy exclusion to
track multiple simultaneous queries during one traversal.  At each recursion
into children we only retain the queries that have returned a possible match
on the previous level.  We implement this by passing as a callback a
user-defined \emph{lazy matching function} \fxn{Match}, which is a boolean
operator that takes as arguments an octant $\oc o$, a boolean
$\mathrm{isLeaf}$ that indicates if $\oc o\in \cO$, and a query index $\ix
q\in\st Q$ and satisfies the following properties:

\begin{itemize}
  \item
    $\fxn{Match} (\oc o,\mathrm{isLeaf}, \ix q)$ returns true if there is a
    leaf $\oc r\in\cO$ that is a descendant of $\oc o$ such that
    $\mathrm{match}_{\ix q}(\oc r) = \mathrm{true}$, and is allowed to return
    a false positive (i.e., true even if $\mathrm{match}_{q}(\oc r)$ is
    false for all descendant leaves of $\oc o$);
  \item
    if $\mathrm{isLeaf}=\mathrm{true}$, then the return value of $\fxn{Match}$
    is irrelevant.  The functionality of the algorithm resides in the action
    of the user-defined callback, which is expected to execute appropriate code
    for strictly matched leaves.
\end{itemize}


\begin{algorithm}
  \caption{\newline\fxn{Search} (\Octant\ \Array\ $\arry A$, \Octant\ $\oc a$,
                         index set\ $\st Q$, \Callback \fxn{Match})}
  \alglab{search}
  \Input{%
    $\arry A$ is a sorted subset of leaves, $\arry A \subseteq \cO$;\\
    $\oc a$ is an ancestor of $\arry A[j]$ for each $j$, $\arry A\subseteq
    \desc(\oc a)$.\\
  }
  \Result{%
    $\forall \ix q\in\st Q$,
    $\mathrm{match}_{\ix q}(\oc o)$ is called for all leaves $o$\\
    whose ancestors $\oc a$ have also returned true
    from $\mathrm{match}_{\ix q}(\oc a)$.
  }

  \algorule

  \lIf{$\arry A = \emptyset$}{\Return}
  \Boolean\ isLeaf $\leftarrow (\arry A = \{ \oc a \})$

    index set $\st Q_{\mathrm{match}} \leftarrow \emptyset$
      \Comment{queries that pass the lazy criteria at $\oc a$}

      \ForAll{$q\in\st Q$}{%
        \lIf(\Comment{\fxn{Match}
        runs
        user action
        for leaves})
        {\fxn{Match} ($\oc a$, isLeaf, $q$) }{%
          $\st Q_{\mathrm{match}} \leftarrow \st Q_{\mathrm{match}} \cup \{ q
            \}$
          }
        }

        \If{$\st Q_{\mathrm{match}} \neq \emptyset$ \algand \algnot isLeaf}{%
          $\arry H \leftarrow$ \fxn{Split\_array} ($\arry A$, $\oc a$)
            \Comment{divide $\arry A$ between the children of $\oc a$: see
              \secref{split}}

            \lForAll{$0 \leq i < 2^d$}{%
              \fxn{Search} ($\arry H[i]$, $\child(\oc a)[i]$, $\st
              Q_{\rmatch}$, $\fxn{Match}$)
            }
          }
\end{algorithm}%


To extend the action of \fxn{Search} to the whole forest, it can be called
once for each tree index $0\leq t \leq K$ with $\localtreeoctants$ and
$\Root(t)$ as arguments. The resulting algorithm
is communication-free and every leaf is queried on only one process, although
the ancestors of leaves may appear as arguments to \fxn{Match} for multiple
processes.

%

\subsection{Array splitting}
\seclab{split}

\fxn{Search} requires an algorithm \fxn{Split\_array} that we have not yet
specified.  \fxn{Split\_array} takes a sorted array of leaves $\arry A$ and an
octant $\oc a$ such that each leaf $\arry A[j]$ is a descendant of $\oc a$ and
partitions $\arry A$ into sorted arrays $\arry H[0]$, $\arry H[1]$, \dots,
$\arry H[2^d-1]$ such that $\arry H[i]$ contains the descendants of
$\child(\oc a)[i]$.

%
Because $\arry A$ is sorted, the subarrays can be indicated by a
non-decreasing sequence of indices $0=\arry k[0] \leq \arry k[1] \leq ... \leq
\arry k[2^d] = |\arry A|$, such that $\arry H[i] = \arry
A[\arry k[i], \dots, \arry k[i + 1] -1]$.  If $\child(\oc a)[i]$ has no
descendants in $\arry A$, this is indicated by $\arry k[i] = \arry k[i + 1]$.

Let us assume that the children of $\oc a$ have level $l$.  If we know that an
octant $\oc o$ is a descendant of $\child(\oc a)[i]$ for some $i$, then we can
compute $i$ from $\oc o.\ix x$ using \algref{ancestorid}, which works because
we use the Morton order as our space-filling curve.  We call this algorithm
\fxn{Ancestor\_id}, because it is a simple generalization of the algorithm
\fxn{Child\_id} \cite[Algorithm 1]{BursteddeWilcoxGhattas11}.

\begin{algorithm}
  \caption{\fxn{Ancestor\_id} (\Octant\ $\oc o$, \Int\ $l$)}
  \alglab{ancestorid}%
  \Input{%
    $0< l \le \oc o.\ix l$
  }
  \Result{%
    $i$ such that if $\oc a.\ix l=l-1$ and $\oc o\in\desc(\oc a)$,
    then $\oc o\in\desc(\child(\oc a)[i])$
  }

  \algorule

  $h \leftarrow 2^{\lmax-l}$
  \Comment{%
    the $(\lmax-l)$th bits of the coordinates $\oc o.x$ describe the
    ancestor with level $l$%
  }

  $i \leftarrow 0$

  \ForAll{$0 \leq j < d$}{%
    $i \leftarrow i ~\bitwor~ (\ternary{\oc o.\ix x_j ~\bitwand~ h}{2^j}{0})$
    \Comment{%
      ``$\bitwor$'' and ``$\bitwand$'' are bitwise OR and AND%
    }
  }

  \Return $i$
\end{algorithm}

If we applied \fxn{Ancestor\_id} to each octant in $\arry A$, we would get a
monotonic sequence of integers, so if we search $\arry A$ with the key $i$ and
use \fxn{Ancestor\_id} to test equality, the lowest matching index will give
the first descendant of $\child(\oc a)[i]$ in $\arry A$.  The split operation,
however, is used repeatedly, both by \fxn{Search} and by the algorithm
\fxn{Iterate} we will present in \secref{iterate}.  To make this procedure as
efficient as possible, we combine these searches into one algorithm
\fxn{Split\_array} (\algref{split}), which is essentially an efficient binary
search for a sorted list of keys.

\begin{algorithm}
  \caption{\fxn{Split\_array} (\Octant\ \Array\ $\arry A$, \Octant\ $\oc a$)}
  \alglab{split}%

  \Input{%
    $\arry A$ is sorted; 
    $\oc a$ is a strict ancestor of $\arry A[j]$ for each $j$,
    $\arry A \subseteq \desc(\oc a)\backslash\{\oc a\}$.
  }
  \Result{%
    $\forall 0 \leq i < 2^d$, $\arry H[i]$ is a sorted array containing
    $\desc(\child(\oc a)[i])\cap \arry A$.
  }

  \algorule

  $\arry k[0] \leftarrow 0$
  \Comment{%
    \textbf{invariant 1 $\forall i$:} if $j \geq \arry k[i]$, \dots
    \hspace{1.0cm}%
  }

  \lForAll%
  {$1 \leq i \leq 2^d$}%
  {%
    $\arry k[i] \leftarrow |\arry A|$
    \Comment{%
      \dots then \fxn{Ancestor\_id} $(\arry A[j], \oc a.l + 1) \geq i$%
    }%
  }

  \For%
  {$i=1$ \algto $2^d-1$}{%
    $m \leftarrow \arry k[i - 1]$
    \Comment{%
      \textbf{invariant 2:} if $j < m$, then \fxn{Ancestor\_id} $(\arry A[j],
      \oc a.l + 1) < i$%
    }

    \While{$m < \arry k[i]$}{%
      $n \leftarrow m + \lfloor (\arry k[i] - m)/2 \rfloor$
      \Comment{\,$\arry k[i-1] \le m \leq n < \arry k[i]$}

      $c \leftarrow $ \fxn{Ancestor\_id} $(\arry A[n],\ \oc a.l + 1)$
      \Comment{%
        $\arry A[n]\in\desc(\child(\oc a)[c])$%
      }

      \eIf%
      (\Comment{%
        $\arry A[n]$ is a descendant of a previous child%
      })%
      {$c < i$}{%
        $m \leftarrow n + 1$
        \Comment{%
          increase lower bound to maintain invariant 2%
        }
      }
      (\Comment{%
        $\arry A[n]$ is a descendant of the $c$th child, $c\geq i$
      })%
      {%
        \lForAll%
        {$i\leq j \leq c$}{%
          $\arry k[j] \leftarrow n$
          \Comment{%
            decrease $\arry k[j]$ to maintain invariant 1%
          }%
        }
      }
    }
  }

  \lForAll{$0 \leq i < 2^d$}{%
    $\arry H[i]\leftarrow$ alias $\arry A[\arry k[i],\dots,\arry k[i + 1]
      - 1]$
  }

  \Return{$\arry H$}
\end{algorithm}%

%

\section{Constructing ghost layers for unbalanced forests}
\seclab{ghost}

As discussed in \secref{distributed}, there is no a-priori knowledge on any
given process about what leaves might be in a neighboring process's partition.
This knowledge, however, is necessary to determine the local neighborhoods of
leaves that are adjacent to inter-process boundaries, which is crucial to many
application-level algorithms.  If a forest of octrees obeys a 2:1 balance
condition, it is known that a leaf's neighbors in other partitions can differ
by at most one refinement level.  The previously presented algorithm
\fxn{Ghost} \cite[Algorithm 20]{BursteddeWilcoxGhattas11} uses this fact to
identify neighboring processes and communicate leaves between them.
\fxn{Ghost} is short and effective, but not usable for an unbalanced forest.
Here we present an algorithm for ghost layer construction that works for all
forests.  Its key component is a recursive algorithm that determines when a
leaf's domain and a process's subdomain are adjacent to each other.

\subsection{Ghost layer construction using intersection tests}

A leaf $\oc o\not\in \cO_{\pr q}$ is in the
\emph{full ghost layer} for process $\pr q$ if its boundary intersects the
subdomain's closure, $\partial\Dom(\oc o)\cap \overline{\Omega_{\pr q}}\neq
\emptyset$. 
This definition includes leaves whose intersection with
$\overline{\Omega_{\pr q}}$ is a single vertex.
Some applications, such as discontinuous Galerkin finite element methods, only
require a ghost layer to include leaves whose intersections with
$\overline{\Omega_{\pr q}}$ have codimension 1.  The boundary set $\bound(\oc
o)$ \eqref{eqn:pointbound} allows us to define a ghost layer parametrized by
codimension.
%
  For $1\leq k \leq d$, the \emph{$\pr p$-to-$\pr
    q$ ghost layer} $\arry G_{\pr p\to \pr q}^k$ is a sorted array containing
  the subset of the leaves $\localoctants$ whose boundaries intersect $\pr
  q$'s subdomain at a point with codimension less than or equal to $k$,
  \begin{equation}\eqnlab{ptoqghost}
    \arry G_{\pr p\to\pr q}^k := \{\oc o\in \localoctants:\exists\
    \oc c\in \bound(\oc o),\ \dim(\oc c) \geq d - k,\ \Dom(\oc c) \cap
    \overline{\Omega_{\pr q} }\neq \emptyset \}.
  \end{equation}
  The \emph{$k$-ghost layer} $\arry G_{\pr p}^k$ for process $\pr p$ is the
  sorted union of all $\pr q$-to-$\pr p$ ghost layers,
  \begin{equation}\eqnlab{ghost}
    \arry G_{\pr p}^k := \bigsqcup_{0 \leq \pr q < P,\ \pr q\neq \pr p}
    \arry G_{\pr q\to\pr p}^k.
  \end{equation}

To construct all $\arry G_{\pr p\to\pr q}^k$ according to
\eqref{eqn:ptoqghost}, process $\pr p$ would have to perform intersection
tests between the boundary set of every leaf in $\localoctants$ and every other
process's subdomain.  We can reduce the number of intersection tests by noting
that if $\Dom(\oc c)\cap \overline{\Omega_{\pr q}} \neq \emptyset$, then by a
corollary to Proposition~\ref{prop:suppisect}, $\overline{\Omega_{\pr q}}$
must overlap some octant $\oc s$ in the support set $\supp(\oc c)$ that
surrounds $\oc c$.  Locating the first and last process subdomains that
overlap $\Dom(\oc s)$ for $\oc s\in\supp(\oc c)$ takes $O(\log P)$ time and
typically reduces the number of intersection tests per point from $O(P)$ to
$O(1)$.
We give pseudocode for computing which $\pr p$-to-$\pr q$ ghost layers contain a leaf $\oc o$ in \algref{ghost2}.
\vfill
\begin{algorithm}[H]
  \caption{\fxn{Add\_ghost} (octant $\oc o$, integer $k$)}
  \alglab{ghost2}

  \Input{%
    octant $\oc o\in\localoctants$; 
    $1 \leq k \leq d$.
  }

  \Result{%
    the set $\st Q$ of processes $\pr q$ such that $\oc o\in\arry G_{\pr p\to\pr
      q}^k$ \eqref{eqn:ptoqghost}.
  }

  \algorule

  $\st Q\leftarrow\emptyset$

  \ForAll{$\oc c\in \bound(\oc o)$ \textbf{such that} $\dim(\oc c) \geq
    d-k$}{%
    \ForAll{$\oc s\in \supp(\oc c)\backslash\{\oc o\}$}{%
      $(\oc f_{\oc s},\oc l_{\oc s})\leftarrow \range(\oc s)$

      $({\pr q}_{\rfirst}, {\pr q}_{\rlast}) \leftarrow
        (\locate(\oc f_{\oc s}),\ \locate(\oc l_{\oc s}))$
      \Comment{overlapping processes}


      \ForAll{$\pr q_{\rfirst} \leq \pr q \leq \pr q_{\rlast}, \pr q \neq \pr p$}{%

        $(\oc f_{\pr q},\oc l_{\oc q}) \leftarrow \range(\pr q)$
        \Comment{%
          \ $\overline{\Omega_{\pr q}}=\overline{\bigcup\Dom([\oc
          f_{\pr q}, \oc l_{\pr q}])}$%
        }

        $(\oc f, \oc l) \leftarrow (\max \{\oc f_{\oc s},\oc f_{\pr q}\},
          \ \min \{\oc l_{\oc s},\oc l_{\pr q}\})$
        \Comment{%
          \ $\overline{\bigcup\Dom([\oc f,\oc l])}=\overline{\Omega_{\pr
              q}}\cap \overline{\Dom(\oc s)}$
        }


        \lIf{%
          $\overline{\bigcup\Dom([\oc f,\oc l])}\cap \Dom(\oc c) \neq
          \emptyset$\label{intersecttest}
        }{%
          $\st Q\leftarrow \st Q\cup \{\pr q\}$
        }
      }
    }
  }
  \Return{$\st Q$}
\end{algorithm}

\subsection{Finding a range's boundaries recursively}
\seclab{rangebound}

The kernel of \algref{ghost2} is the intersection test on line
\ref{intersecttest},
\begin{equation}
  \left(\overline{\bigcup\Dom([\oc f,\oc l])}\cap \Dom(\oc c) = \emptyset ?
  \right),
\end{equation}
where $\oc c$ is a point, $[\oc f,\oc l]$ is the range between two atoms, and
$\oc f$ and $\oc l$ are both descendants of an octant $\oc s$ in the support
set $\supp(\oc c)$.  We must specify how this intersection test is to be
performed.
Because $\oc s\in\supp(\oc c)$ implies $\oc c\in\bound(\oc s)$ (see
Propositions~\ref{prop:duality1} and \ref{prop:duality2}), the point $\oc c$
must be equal to $(\oc s,\ix b)$ for some boundary index $b\in\st B$, so we
can rephrase the test as $(\overline{\bigcup\Dom([\oc f,\oc l])}\cap \Dom((\oc
s, b)) = \emptyset ?)$.  This test is a specific case of the
problem of constructing the \emph{range-boundary intersection}
$\st B_{\cap}(\oc f,\oc l,\oc s)$, the set of all 
boundary indices $b$ such that $\Dom((\oc s,b))$ intersects the range,
%
  %
  \begin{equation}\eqnlab{rangebound}
    \st B_{\cap}(\oc f,\oc l,\oc s):=\{b\in\st B:\overline{\bigcup\Dom([\oc f,\oc
      l])} \cap
    \Dom((\oc s,b))\neq \emptyset\}.
  \end{equation}

If the range $[\oc f,\oc l]$ contains only descendants of some child of
$\oc s$, say $\child(\oc s)[i]$, then the range-boundary intersection must be
a subset of the \emph{child-boundary intersection} $\st B_{\cap}^i$, which is
the subset of $\st B$ corresponding to points in $\bound(\oc s)$ intersected
by $\overline{\Dom(\child(\oc s)[i])}$,
\begin{equation}\eqnlab{childisect}
  \st B_{\cap}^i := \{b\in\st B:\overline{\Dom(\child(\oc s)[i])}\cap
  \Dom((\oc s,b))\neq \emptyset\}.
\end{equation}
The child-boundary intersections $\{\st B_{\cap}^i\}_{0\leq i<2^d}$
(\figref{boundset}, right) are the same for all (non-atom) octants.
%
%
The following proposition, illustrated in \figref{boundset} (left), shows how
the child-boundary intersection $\st B_{\cap}^i$ relates $\st B_{\cap}(\oc f,
\oc l, \child(\oc s)[i])$ to $\st B_{\cap}(\oc f, \oc l, \oc s)$.

\begin{propn}[Range-boundary intersection recursion]\label{prop:isect}
  If $\oc f$ and $\oc l$ are atoms, $\oc f \leq \oc l$, and both are
  descendants of $\child(\oc s)[i]$, then
  \begin{equation}\label{findrangeeq}
    (b\in \st B_{\cap}(\oc f, \oc l, \oc s)) \Leftrightarrow
    (b\in \st B_{\cap}(\oc f, \oc l, \child(\oc s)[i])\cap \st B_{\cap}^i).
  \end{equation}
\end{propn}%

\begin{figure}
  \raisebox{-0.5\height}{\includegraphics{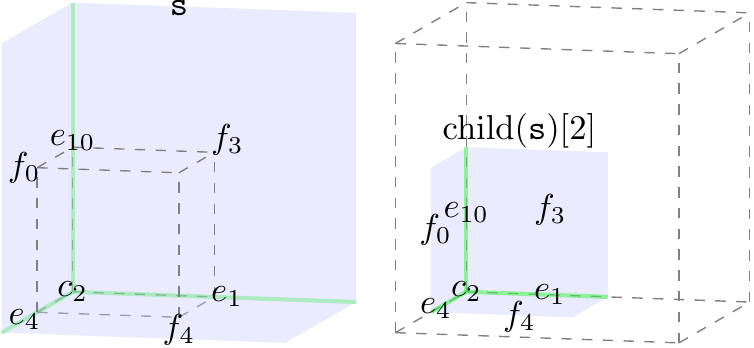}}
  \hfill
  \renewcommand{\arraystretch}{1.2}
  \begin{tabular}{|l|l|}
    \hline
    $\st B_{\cap}^0$ & $\{c_0,e_0,e_4,e_8,f_0,f_2,f_4\}$   \\
    $\st B_{\cap}^1$ & $\{c_1,e_0,e_5,e_9,f_1,f_2,f_4\}$   \\
    $\st B_{\cap}^2$ & $\{c_2,e_1,e_4,e_{10},f_0,f_3,f_4\}$   \\
    $\st B_{\cap}^3$ & $\{c_3,e_1,e_5,e_{11},f_1,f_3,f_4\}$   \\
    $\st B_{\cap}^4$ & $\{c_4,e_2,e_6,e_8,f_0,f_2,f_5\}$   \\
    $\st B_{\cap}^5$ & $\{c_5,e_2,e_7,e_9,f_1,f_2,f_5\}$   \\
    $\st B_{\cap}^6$ & $\{c_6,e_3,e_6,e_{10},f_0,f_3,f_5\}$   \\
    $\st B_{\cap}^7$ & $\{c_7,e_3,e_7,e_{11},f_1,f_3,f_5\}$   \\
    \hline
  \end{tabular}
  \caption{%
    (left) An illustration of a specific instance of Proposition~\ref{prop:isect}.
    If the range $[\oc f,\oc l]$ contains only descendants of $\child(\oc
    s)[2]$, then its domain intersects $\Dom((\oc s,b))$ (left) if and only
    if it intersects $\Dom((\child(\oc s)[i],b))$ (right) and $b \in \st
    B_{\cap}^2$. (right) The child-boundary intersections are enumerated.
  }%
  \figlab{boundset}
\end{figure}%

This result allows us to construct $\st B_{\cap}(\oc f,\oc l,\oc s)$ by
partitioning $[\oc f,\oc l]$ into ranges for all of the overlapping children,
\begin{equation}
  [\oc f,\oc l] = \bigsqcup_{i\in \mathcal{I}} [\oc f_i,\oc l_i],\quad
  \mathcal{I}:=\{i:\exists \oc a\in\desc(\child(\oc s)[i]),\ 
  \oc f\leq \oc a\leq \oc l\},
  \label{runion}
\end{equation}
and constructing the
range-boundary intersection for those children,
%
\begin{equation}
  \st B_{\cap}(\oc f,\oc l,\oc s) = \bigcup_{i\in\mathcal{I}}
  \st B_{\cap}(\oc f_i,\oc l_i,\child(\oc s)[i]) \cap \st B_{\cap}^i.
  \label{dunion}
\end{equation}
This leads to the recursive algorithm \fxn{Find\_\-range\_\-boundaries}
(\algref{findrange}), which computes $\st B_{\cap}(\oc f,\oc l,\oc s)\cap\st
B_{\query}$ for a set $\st B_{\query}\subseteq\st B$.

\begin{algorithm}
  \caption{\fxn{Find\_range\_boundaries} (\Octants\ $\oc f$, $\oc l$, $\oc s$,
    index set $\st B_\query$) }
  \alglab{findrange}

  \Input{%
    $\oc f$ and $\oc l$ are atom descendants of $\oc s$,
    $\oc f \leq \oc l$; 
    $\st B_\query \subseteq \st B$.
  }

  \Result{%
    $\st B_{\cap}(\oc f, \oc l, \oc s) \cap \st
    B_\query$ \eqref{eqn:rangebound}.
  }

  \algorule


  \lIf{$\st B_\query = \emptyset$ \algor $\oc s.l = \lmax$}{%
    \Return{$\st B_\query$}
    \label{maxexit}
  }


  $j \leftarrow \fxn{Ancestor\_id}\ (\oc f, \oc s.l+1)$
  \Comment{index of first child whose range overlaps $[\oc f,\oc l]$}

  $k \leftarrow \fxn{Ancestor\_id}\ (\oc l, \oc s.l+1)$
  \Comment{index of last child whose range overlaps $[\oc f,\oc l]$}

  \lIf%
  {$j = k$}{%
    \Return $\fxn{Find\_range\_boundaries}\ (\oc f, \oc l, \child(\oc s)[j],
    \st B_\query \cap \st B_{\cap}^{j})$
    \label{equalchildren}
  }

  $\st B_\rmatch \leftarrow \mathop{\bigcup}_{j < i < k} \st B_\query \cap
    \st B_{\cap}^i$
  \Comment{boundary touched by wholly-covered children}
  \label{distinctchildrenmid}

  $\st B_\rmatch^j \leftarrow (\st B_\query \cap
    \st B_{\cap}^j)\backslash \st B_\rmatch$

  $(\oc f_j,\oc l_j) \leftarrow \range(\child(\oc s)[j])$
  \Comment{\ $[\oc f,\oc l]\cap[\oc f_j,\oc l_j] = [\oc f,\oc
    l_j]$: if $\oc f\neq \oc f_j$, recursion is needed}

  \lIf%
  {$\oc f \neq \oc f_j$ \label{distinctchildrenfirst}}{%
    $\st B_\rmatch^j \leftarrow  
      \fxn{Find\_range\_boundaries}\ (\oc f, \oc l_j, \child(\oc s)[j], \st
      B_\rmatch^j)$
  }

  $\st B_\rmatch^k \leftarrow ((\st B_\query \cap \st
    B_{\cap}^j)\backslash\st B_{\rmatch}) \backslash B_{\rmatch}^j$

  $(\oc f_k,\oc l_k) \leftarrow \range(\child(\oc s)[k])$
  \Comment{\ $[\oc f,\oc l]\cap[\oc f_k,\oc l_k] = [\oc
    f_k,\oc l]$: if $\oc l\neq \oc l_k$, recursion is needed}

  \lIf%
  {$\oc l \neq \oc l_k$ \label{distinctchildrenlast}}{%
    $\st B_\rmatch^k \leftarrow \fxn{Find\_range\_boundaries}\ (\oc f_k, \oc
      l, \child(\oc s)[k], \st B_\rmatch^k)$
  }

  \Return{$(\st B_\rmatch\cup \st B_{\rmatch}^j \cup \st B_{\rmatch}^k)$
    \label{distinctchildrenreturn}}
\end{algorithm}

To compute the intersection test $(\overline{\bigcup\Dom([\oc f,\oc
  l])}\cap\Dom((\oc s,b))=\emptyset?)$ in \algref{ghost2}, we choose $\st
B_{\query}=\{b\}$ and use \fxn{Find\_\-range\_\-boundaries} to check whether
$(\st B_{\cap}(\oc f, \oc l, \oc s)\cap\st B_{\query}=\emptyset?)$.  A proof
of the correctness of \algref{findrange} is given in \appref{findrangeproof}.
The recursive procedure is also illustrated in \figref{figfindrange}.

\begin{figure}
  \centering
  \begin{tabular}{|>{\centering}p{0.46\columnwidth}|>{\centering}p{0.46\columnwidth}|}
    \hline
    \vspace{-2ex}
    \includegraphics{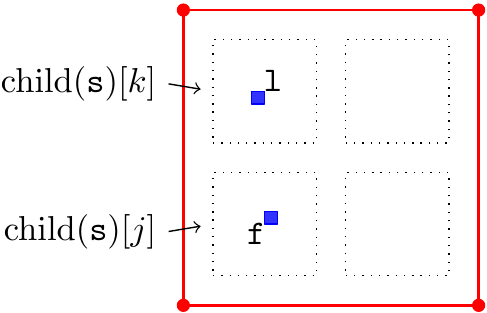}

    $j\leftarrow \fxn{Ancestor\_id}(\oc f,\oc s.l + 1)$

    $k\leftarrow \fxn{Ancestor\_id}(\oc l,\oc s.l + 1)$
    &
    \vspace{-2ex}
    \includegraphics{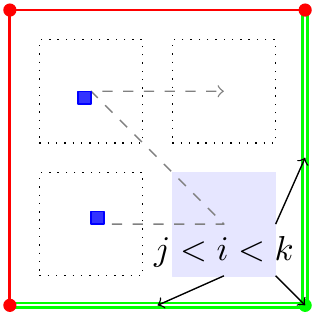}

    $\st B_{\rmatch} \leftarrow \bigcup_{j<i<k} \st B_{\query}\cap \st
      B_{\cap}^i$
    \tabularnewline \hline

    \vspace{-2ex}
    \includegraphics{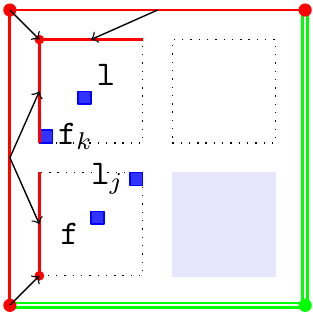}

    $\st B_{\query}^j\leftarrow (\st B_{\query} \cap B_{\cap}^j)
      \backslash \st B_{\rmatch}$

    $\st B_{\query}^k\leftarrow (\st B_{\query} \cap B_{\cap}^k)
      \backslash \st B_{\rmatch}$
    &
    \vspace{-2ex}
    \includegraphics{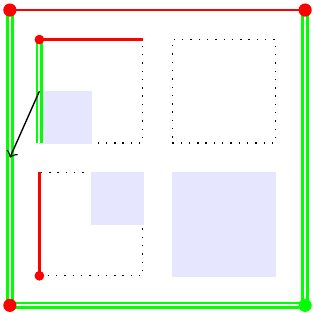}

    $\st B_{\rmatch}^j \leftarrow \st B_{\cap}(\oc f,\oc l_j,
      \child(\oc s)[j])\cap \st B_{\query}^j$

    $\st B_{\rmatch}^k \leftarrow \st B_{\cap}(\oc f_k,\oc l,
      \child(\oc s)[k])\cap \st B_{\query}^k$

    $\mathbf{return}\ \st B_{\rmatch} \cup \st B_{\rmatch}^j \cup \st
    B_{\rmatch}^k$
    \tabularnewline \hline
  \end{tabular}
  \caption{%
    An illustration of \fxn{Find\_range\_boundaries}, listed in
    \algref{findrange}.  (top left) Solid red lines indicate the points in
    $\bound(\oc s)$ indexed by $\st B_\query$; the children containing $\oc f$
    and $\oc l$ are determined. (top right) The contribution to $\st
    B_{\cap}(\oc f,\oc l,\oc s)\cap \st B_{\query}$ (double green lines) of
    children in the middle of the range is calculated; light blue indicates
    that this portion of the range $[\oc f,\oc l]$ has been processed. (bottom
    left) The arguments for the recursive calls are constructed. (bottom
    right) The sets returned by the recursive calls are added to the return
    set.%
  }
  \figlab{figfindrange}
\end{figure}

\subsection{Notes on implementation}

Ghost layer construction in \fxn{Ghost} has a few optimizations not
given in the pseudocode in \algref{ghost2}.  Most leaves in $\localoctants$ do
not touch the boundary of $\overline{\Omega_{\pr p}}$, and so cannot be in
$\arry G^k_{\pr p\to\pr q}$ for any $\pr q\neq \pr p$. To avoid the
intersection tests for these leaves, we first check to see if $\oc o$'s
$3\times 3$ neighborhood (or ``insulation layer''
\cite{SampathAdavaniSundarEtAl08}) is owned by $\pr p$: this can be
accomplished with two comparisons, for the first and last atoms of the
neighborhood, against the first-atoms array $\arry f$ \eqref{eqn:firstatoms}.
We also note that if $\oc c$ is a $0$-point then $\overline{\Omega_{\pr q}}$
intersects $\Dom(\oc c)$ if and only $\pr q = \locate(\oc a)$ for some atom
$\oc a$ in $\oc c$'s atomic support set \eqref{eqn:pointasupp}.  Because this
simpler test is available, we do not call $\fxn{Find\_\-range\_\-boundaries}$
for $0$-points.

If an instance of \fxn{Find\_\-range\_\-boundaries} calls two recursive
copies of itself, all future instances will call only one recursive copy. We
use this fact to take advantage of tail-recursion optimization in our
implementation.  Because $|\st B|<32$ for $d=2$ and $d=3$, we can perform the
set intersection, unions, and differences in \algref{findrange} by assigning
each $b\in\st B$ to a bit in an integer and performing bitwise operations.


\section {A universal topology iterator}
\seclab{iterate}

A forest is first of all a storage scheme for a mesh refinement topology.
Applications use this topological information in ways that we do not wish to
restrict or anticipate.  We focus instead on designing an interface that
conveys this information to applications in a complete and efficient manner,
with the main goal of minimizing the points of contact between \pforest on the
one hand and the application on the other.

As we will see in our discussion of a specific node numbering algorithm in
\secref{lnodes}, some applications need to perform operations not just on
leaves, but also on their boundary points.  Our algorithm that facilitates
this is called \fxn{Iterate}.

\begin{remark}\upshape
  The algorithm \fxn{Iterate} requires that the leaves in the micro layer
  $\cO$ are 2:1 balanced.  This is assumed for the remainder of this section
  and \secref{lnodes}.  This assumption is not very limiting: most
  applications that need topological information (e.g., finite element or fast
  multipole calculations) require 2:1 balance, either for numerical reasons
  or to avoid the complexity of handling general adjacencies.
  On the other hand, extending the algorithm to larger refinement ratios
  when the need arises appears to be a practical option since this case is
  covered by the recursion as well.
\end{remark}

\subsection{Definitions}

The algorithm \fxn{Iterate} is distributed and communication-free (assuming
that the ghost layer $\arry G_{\pr p}^d$ \eqref{eqn:ghost} has already been
constructed): on process $\pr p$, it takes a user-supplied callback and
executes it once for every leaf and leaf-boundary point $\oc c$ that is
relevant to $\localoctants$, supplying information about the neighborhood
around $\oc c$.  We define exactly what this means here.

The union of all leaves with their closure sets, $\bigcup\clos(\cO)$, defines
a covering of $\Omega$, $\overline{\Omega}= \bigcup\Dom(\bigcup\clos(\cO))$.
This covering may not be a partition, because $\bigcup\clos(\cO)$ may contain
\emph{hanging} points: $n$-points of dimension $n< d$ that are in the child
partition sets \eqref{eqn:pointpart} of other points in
$\bigcup\clos(\cO)$.  We can define a partition by removing these hanging
points.
%
  The \emph{global partition} $\partoctants$ is
  \begin{equation}
    \partoctants :=
    \bigcup\clos(\cO) \backslash
    \{\oc c:
    \exists\ \oc e \in \bigcup\clos(\cO),
    \ \oc c\in\part(\oc e)\}.
  \end{equation}

If there is just one process, the function \fxn{Iterate} executes a
user-supplied callback function for every point $\oc c\in\partoctants$.  For a
distributed forest, \fxn{Iterate} as called by process $\pr p$ executes the
callback function only for the subset of $\partoctants$ that is relevant to
$\localoctants$.  In the \pforest implementation, we allow for two definitions
of what is relevant.  The first is the \emph{locally relevant partition},
which is the subset $\localpart\subseteq \partoctants$ that overlaps
$\overline{\Omega_{\pr p}}$,
%
  %
  \begin{equation}
    \localpart :=\{\oc c \in \partoctants:
    \exists\ \oc o\in\localoctants,\ \Dom(\oc c)\cap\overline{\Dom(\oc
      o)}\neq\emptyset\}.
  \end{equation}

One potential problem with $\localpart$ is that, because of hanging points, it
may not be closed: if $\oc c\in\localpart$, there may be $\oc e\in
\clos(\oc c)$ such that $\oc e\not\in\localpart$.  As we will show in
\secref{lnodes}, closedness is necessary for some applications, so we also
define the \emph{closed locally relevant partition} $\overline{\localpart}$,
%
  \begin{equation}\eqnlab{closedpart}
    \overline{\localpart}:=\bigcup\{\clos(\oc e):\ \oc e\in\localpart\}.
  \end{equation}
%
The sets we have defined thus far---$\partoctants$, $\localpart$, and
$\overline{\localpart}$---are illustrated in \figref{hanging}.

\begin{figure}\centering
  \includegraphics{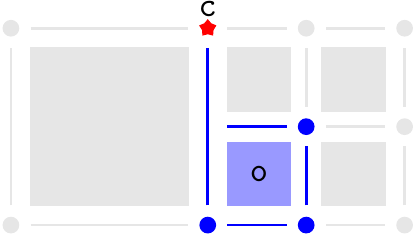}
  \caption{%
    Suppose process $\pr p$ owns only octant $\oc o$ in this two-dimensional
    illustration.  $\partoctants$ is the set of all points shown: note that
    because some points in $\clos(\oc o)$ are hanging, they are not included.
    The set $\localpart$ of points that overlap $\overline{\Omega_p}$ is shown
    in blue.  The $0$-point $\oc c$ shown as a red star is not in
    $\localpart$, but is in its closure, $\overline{\localpart}$.%
  }%
  \figlab{hanging}
\end{figure}

If \fxn{Iterate} only supplied the callback with each relevant point $\oc c$
in isolation, its utility would be limited, because it would say nothing about
the neighborhood around $\oc c$.  We call the neighborhood of adjacent leaves
the \emph{leaf support set} $\lsupp(\oc c)$,
%
  \begin{equation}\eqnlab{lsupp}
    \lsupp(\oc c):=\{\oc o\in \cO : \overline{\Dom(\oc o)} \cap \Dom(\oc c) \neq
    \emptyset\}.
  \end{equation}
%
  Note that the leaf support set $\lsupp(\oc c)$ may differ from the support
  set $\supp(\oc c)$ \eqref{eqn:pointsupp}, which is independent of the leaves
  in the micro layer $\cO$.
%
Because process $\pr p$ does not have access to all leaves in $\cO$, it can
only compute the subset of the leaf support set that is contained in the local
leaves and in the full ghost layer $\arry G_{\pr p}^d$ \eqref{eqn:ghost},
which we call the \emph{local leaf support set} $\lsupp_{\pr p} (\oc c)$,
%
  \begin{equation}\eqnlab{llsupp}
    \lsupp_{\pr p}(\oc c):=\lsupp(\oc c)\cap(\localoctants \cup
    \arry{G}_{\pr p}^d).
  \end{equation}
%

\begin{propn}\label{prop:completeinfo}
  $\lsupp_{\pr p}(\oc c)=\lsupp(\oc c)$ if and only if $\oc c\in\localpart$.
\end{propn}

The local leaf support set, though it may not contain all of the leaf support
set, is what \fxn{Iterate} supplies to the user-supplied callback to describe
the neighborhood around $\oc c$.  The local leaf support set can be used to
determine if $\oc c\in\localpart$ or $\oc c\in\overline{\localpart}$, using a
function \fxn{Is\_relevant} (\algref{relevant}).

\begin{algorithm}[H]
  \caption{\fxn{Is\_relevant} (point $\oc c$, octant set $\lsupp_{\pr p}(\oc c)$)}
  \alglab{relevant}

  \Input{%
    $\lsupp_{\pr p}(\oc c)$ is the local leaf support set of $\oc c$
    \eqref{eqn:llsupp}.
  }

  \Result{%
    true if and only if $\oc c$ is in the relevant set ($\localpart$ or
    $\overline{\localpart}$).
  }

  \algorule

  \ForAll{$\oc s\in\lsupp_{\pr p}(\oc c)$}{%
    \lIf%
    {%
      $\oc s\in\localoctants$
    }{%
      \Return \algtrue
    }%
    \Else(\Comment{%
      this else-statement is only used if $\overline{\localpart}$ is
      relevant
    })%
    {%
      \ForAll%
      {$\oc e\in \bound(\oc s)$ \textbf{such that} $\oc c\in\bound(\oc e)$}{%
        \lIf%
        (\Comment{%
          intersection test from
          \secref{rangebound}
        })%
        {$\Dom(\oc e)\cap\overline{\Omega_{\pr p}}\neq\emptyset$}{%
          \Return \algtrue
        }
      }
    }
  }
  \Return \algfalse
\end{algorithm}

For each octant $\oc o\in\lsupp_{\pr p}(\oc c)$, the \pforest implementation
of \fxn{Iterate} supplements the octant data fields $\oc o.\ix l$, $\oc o.\ix
x$, and $\oc o.\ix t$ with additional data passed to the callback function.
We supply a boolean identifying whether $\oc o\in\localoctants$, so no
searching is necessary to determine if $\oc o$ is local or a ghost.  We also
supply the index of $\oc o$ within either $\localtreeoctants$ for $t=\oc o.\ix
t$ (which is easily converted to $j$ such that $\localoctants[j]=\oc o$) or
within the ghost layer $\arry G^d_{\pr p}$ \eqref{eqn:ghost}.  Keeping track
of this information does not change the computational complexity of
\fxn{Iterate}, but introduces additional bookkeeping that we will omit from
our presentation of the algorithm.

\subsection{Iterating in the interior of a point}

A simple implementation of \fxn{Iterate} might take each leaf $\oc
o\in\localoctants$ in turn and, for each $\oc c\in \bound(\oc o)$, compute
$\lsupp_{\pr p}(\oc c)$ by searching through $\localoctants$ and $\arry
G^d_{\pr p}$ to find $\oc o$'s neighbors that are adjacent to $\oc c$.  A
bounded number of binary searches would be performed per leaf, so the total
iteration time would be $O(N_{\pr p}\log N_{\pr p})$.  This is the strategy
used by the \fxn{Nodes} algorithm in \pforest \cite[Algorithm
21]{BursteddeWilcoxGhattas11} and by other octree libraries
\cite{SundarSampathBiros08}. We note two problems with this approach.  The
first is the large number of independent searches that are performed.  The
second is that this approach needs some way of ensuring that the callback is
executed for each relevant point only once, such as storing the set of points
for which the callback has executed in a hash table.

Instead, the implementation of \fxn{Iterate} that we present proceeds
recursively. We take as inputs to the recursive procedure a point $\oc c$ and
a set of arrays $\arry S$, where $\arry S[i]$ contains all leaves that are
descendants of the support set octant $\supp(\oc c)[i]$. If $\oc c$ is in the
global partition $\partoctants$ and is in the relevant set ($\localpart$ or
$\overline{\localpart}$), then the octants in $\lsupp_{\pr p}(\oc c)$ can be
found in the $\arry S$ arrays and the callback function can be executed.
Otherwise, the points of the global partition $\partoctants$ that are
descendants of $\oc c$ can be divided between the points in the child
partition set $\part(\oc c)$ \eqref{eqn:pointpart}.  Each $\oc e\in \part(\oc
c)$ takes the place of $\oc c$ in a call to the recursive procedure: to
compute the new $\arry S$ arrays for $\oc e$, we use the function
\fxn{Split\_array} (described in \secref{split}) on the original arrays $\arry
S$.  This is spelled out in \fxn{Iterate\_interior} (\algref{iterinterior}).

\begin{algorithm}
  \caption{\fxn{Iterate\_interior} (point $\oc c$, octant arrays $\arry
    S$,
    \fxn{callback})}
  \alglab{iterinterior}

  \Input{%
    $\arry S[i]$ is the sorted array of all leaves $\oc
    o\in\localoctants\cup\arry G_{\pr p}^d$ such that $\oc o\in\desc(\supp(\oc
    c)[i])$.
  }

  \Result{%
    if $\oc c\in\localpart$ (or $\overline{\localpart}$), then $\lsupp_{\pr
      p}(\oc c)$ \eqref{eqn:llsupp} is computed and passed to
    \fxn{callback}; 
    otherwise, \fxn{callback}
    is called for each $\oc e\in\localpart$ (or $\overline{\localpart}$)
    such that $\Dom(\oc e)\subset\Dom(\oc c)$.
  }

  \algorule

  \lIf%
  (\Comment{%
    stop if there are no local leaves
  })%
  {$\bigcup \arry{S} \cap\localoctants = \emptyset$}{%
    \Return
  }

  octant set $\mathcal{L}\leftarrow\emptyset$
  \Comment{%
    if $\oc c\in\partoctants$, then $\mathcal{L}$ will equal $\lsupp_{\pr
      p}(\oc c)$%
  }

  boolean $\mathrm{stop} \leftarrow \algfalse$
  \Comment{\ $\mathrm{stop}$ will be true if $\oc c\in\partoctants$}

  \eIf%
  (\Comment{%
    see \figref{itervolfig}
  })%
  {$\dim(\oc c)>0$}{%
    \ForAll{$0 \leq i < |\supp(\oc c)|$}{%
      $\oc s\leftarrow \supp(\oc c)[i]$

      \eIf%
      (\Comment{%
        if {$\oc s$} is a leaf, \dots
      })%
      {$\arry S[i] = \{\oc s\}$}{%
        $\mathrm{stop}\leftarrow \algtrue$
        \Comment{then $\oc c\in\partoctants$, \dots}

        $\mathcal{L}\leftarrow\mathcal{L} \cup\{\oc s\}$
        \Comment{and $\oc s\in \lsupp(\oc c)$}
         \label{iteraddlarge}
      }{%
        $\arry{H}_i \leftarrow \fxn{Split\_array}(\arry S[i], \oc s)$
        \Comment{$O(\log |\arry S[i]|)$ (see \secref{split})}
        \label{itersplit}

        $\arry h_i\leftarrow \child(\oc s)$
        \label{childrensplit}
        \Comment{%
          if $\oc c\in\partoctants$, then by the 2:1 condition,\ $\dots$%
        }

        $b\leftarrow b\text{ such that }\oc c=(\oc s,b)$
        \Comment{%
          children next to $(\oc s,b)=\oc c$ are in $\lsupp(\oc c)$%
        }

        $\mathcal{L}\leftarrow \mathcal{L} \cup \{\arry h_i[j]: b\in\st
          B_{\cap}^j\}$
        \label{iteraddsmall}
        \Comment{%
          $\st B_{\cap}^j$ are child-boundary intersection sets 
          \eqref{eqn:childisect}%
        }
      }
    }
  }
  (\Comment{%
    a $0$-point, see \figref{figitercorner}%
  })%
  {%
    $\mathrm{stop}\leftarrow \algtrue$
    \Comment{%
      anytime we find a $0$-point between leaves, it is in $\partoctants$.
    }

    \For%
    (\Comment{%
      find leaves surrounding $\oc c$:
      use Proposition~\ref{prop:duality2}%
    })%
    {$0 \leq i < |\supp(\oc c)|$}{%
      $\mathcal{L}\leftarrow \mathcal{L}\cup \{\oc o\in\arry S[i]:\asupp(\oc
        c)[i]\in\desc(\oc o)\}$
      \label{itercorneradd}
      \Comment{requires $O(\log |\arry S[i]|)$ search}
    }
  }%
  \eIf{$\mathrm{stop}$}{%
    \lIf{$\fxn{Is\_relevant}(\oc c,\mathcal{L})$\label{callbackrelevant}}{%
      $\fxn{callback} (\oc c,\mathcal{L})$
      \label{itercallback}
    }
  }{%
    \ForAll%
    (\Comment{%
      set up recursion for each point in child partition set
      \eqref{eqn:pointpart}
    })%
    {$\oc e\in \part(\oc c)$\label{interiorloop}}{%
      \ForAll%
      (\Comment{%
        find descendants of support set octants
      })%
      {$0 \leq i < |\supp(\oc e)|$}{%
        $\arry{S}_{\oc e}[i]\leftarrow \arry H_j[k]$ such that $\arry h_j[k] =
          \supp(\oc e)[i]$
        \Comment{subarrays created on line \ref{itersplit}}
      }
    }
    $\fxn{Iterate\_interior}(\oc e,\arry{S}_{\oc e},\fxn{callback})$
    \label{iterrecurse}
  }
\end{algorithm}

We provide some figures to illustrate the recursion in
\fxn{Iterate\_interior}: \figref{itervolfig} shows the cases when $\dim(\oc
c)=d$ and $0 < \dim (\oc c) < d$ and \figref{figitercorner} shows the case
when $\dim(\oc c)=0$.  The correctness of \algref{iterinterior} is proved in
\appref{iterproof}.

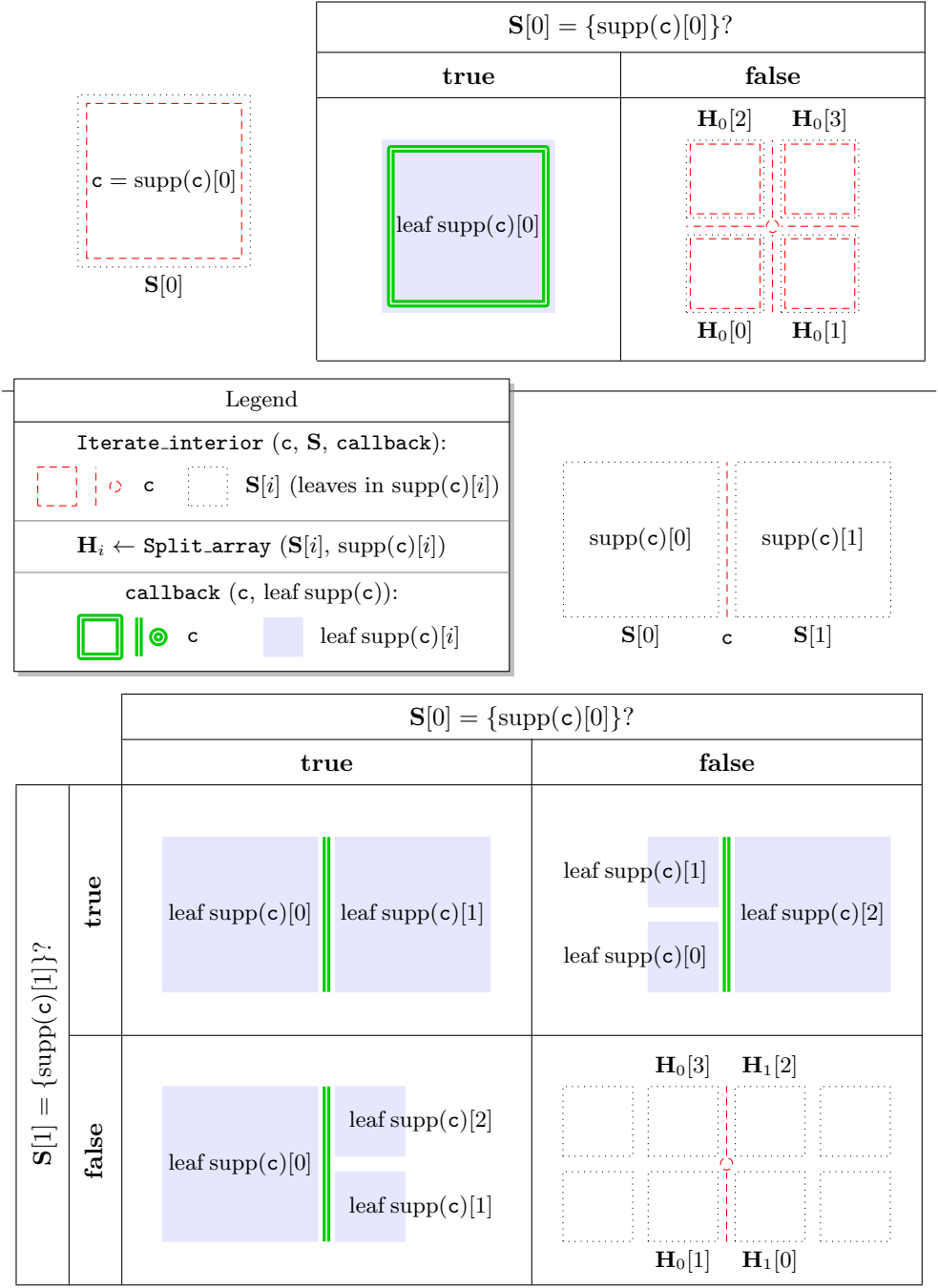
\begin{figure}\centering
  \renewcommand{\arraystretch}{1.5}
  \begin{tabular}{l|c|c|}
                                                                 \cline{2-3}
    \multirow{3}{*}[-2em]{{\small
\usetikzlibrary{trees}
\usetikzlibrary{quadtree}
\begin{tikzpicture}[line join=round,scale=1.2]

  \tikzstyle{unknownstyle}=[dotted]
  \tikzstyle{positionstyle}=[draw=yellow,fill=yellow!20]
  \tikzstyle{leafstyle}=[draw=blue!10,ultra thin,fill=blue!10]
  \tikzstyle{augstyle}=[draw=green,fill=green!20]
  \tikzstyle{branchstyle}=[draw=red,dashed,fill=red!20]
  \tikzstyle{iterline}=[draw=red,densely dashed]
  \tikzstyle{callline}=[very thick,draw=green!80!black,double]
  \def\sp{0.1}
  \def\face{+(0,-1) -- +(0,1)}
  \def\corner{circle(\sp * 0.707 * 2cm)}

  \draw[draw=none,use as bounding box](-1.4cm,-1.4cm) rectangle (1.4cm,1.5cm);
  \draw [unknownstyle]\quadrant node {$\oc c =\supp(\oc c)[0]$} (0,-1) node [below]
  {$\arry S[0]$};
  \draw [scale=(1-\sp)] [iterline]\quadrant;

%

\end{tikzpicture}
}

}
    & \multicolumn{2}{c|}{$\arry S[0] = \{\supp(\oc c)[0]\}?$} \\ \cline{2-3}
    & \textbf{true} & \textbf{false}                           \\ \cline{2-3}
    &
                         {\small
\usetikzlibrary{trees}
\usetikzlibrary{quadtree}
\begin{tikzpicture}[line join=round,scale=1.2]

  \tikzstyle{unknownstyle}=[dotted]
  \tikzstyle{positionstyle}=[draw=yellow,fill=yellow!20]
  \tikzstyle{leafstyle}=[draw=blue!10,ultra thin,fill=blue!10]
  \tikzstyle{augstyle}=[draw=green,fill=green!20]
  \tikzstyle{branchstyle}=[draw=red,dashed,fill=red!20]
  \tikzstyle{iterline}=[draw=red,densely dashed]
  \tikzstyle{callline}=[very thick,draw=green!80!black,double]
  \def\sp{0.1}
  \def\face{+(0,-1) -- +(0,1)}
  \def\corner{circle(\sp * 0.707 * 2cm)}

  \draw[draw=none,use as bounding box](-1.4cm,-1.4cm) rectangle (1.4cm,1.5cm);

  \draw [leafstyle]\quadrant node {$\lsupp(\oc c)[0]$} (0,1);
  \draw [scale=(1-\sp)] [callline]\quadrant;
%

\end{tikzpicture}
}

    &
                         {\small
\usetikzlibrary{trees}
\usetikzlibrary{quadtree}
\begin{tikzpicture}[line join=round,scale=1.2]

  \tikzstyle{unknownstyle}=[dotted]
  \tikzstyle{positionstyle}=[draw=yellow,fill=yellow!20]
  \tikzstyle{leafstyle}=[draw=blue!10,ultra thin,fill=blue!10]
  \tikzstyle{augstyle}=[draw=green,fill=green!20]
  \tikzstyle{branchstyle}=[draw=red,dashed,fill=red!20]
  \tikzstyle{iterline}=[draw=red,densely dashed]
  \tikzstyle{callline}=[very thick,draw=green!80!black,double]
  \def\sp{0.1}
  \def\face{+(0,-1) -- +(0,1)}
  \def\corner{circle(\sp * 0.707 * 2cm)}

  \draw[draw=none,use as bounding box](-1.4cm,-1.4cm) rectangle (1.4cm,1.5cm);

%
  \draw [morton order,quadtree absolute space=\sp] coordinate
    child { [unknownstyle]\quadrant
            { node {} (0,-1) node [below] {$\arry H_0[0]$}}}
    child { [unknownstyle]\quadrant
            { node {} (0,-1) node [below] {$\arry H_0[1]$}}}
    child { [unknownstyle]\quadrant
            { node {} (0,1) node [above] {$\arry H_0[2]$}}}
    child { [unknownstyle]\quadrant
            { node {} (0,1) node [above] {$\arry H_0[3]$}}};
  \draw [morton order,quadtree absolute space=\sp,quadtree scale=(1-\sp)]
    coordinate
    child { [iterline]\quadrant}
    child { [iterline]\quadrant}
    child { [iterline]\quadrant}
    child { [iterline]\quadrant};
  \draw [iterline,scale=0.5*(1-\sp),yshift=-1cm*(1 + \sp)/(1-\sp)] \face;
  \draw [iterline,scale=0.5*(1-\sp),yshift=1cm*(1 + \sp)/(1-\sp)] \face;
  \draw [iterline,scale=0.5*(1-\sp),xshift=-1cm*(1 + \sp)/(1-\sp),rotate=90]
    \face;
  \draw [iterline,scale=0.5*(1-\sp),xshift=1cm*(1 + \sp)/(1-\sp),rotate=90]
    \face;
  \draw [iterline,scale=0.5] \corner;

\end{tikzpicture}
}

    \\ \cline{2-3}
  \end{tabular}
  
  \vspace{\baselineskip}
  \hrule
  \vspace{\baselineskip}

  \begin{tabular}{|c|c|c|c|}
    \multicolumn{3}{l}{{\small
\usetikzlibrary{trees}
\usetikzlibrary{quadtree}
\begin{tikzpicture}[line join=round,scale=1.2]

  \tikzstyle{unknownstyle}=[dotted]
  \tikzstyle{positionstyle}=[draw=yellow,fill=yellow!20]
  \tikzstyle{leafstyle}=[draw=blue!10,ultra thin,fill=blue!10]
  \tikzstyle{augstyle}=[draw=green,fill=green!20]
  \tikzstyle{branchstyle}=[draw=red,dashed,fill=red!20]
  \tikzstyle{iterline}=[draw=red,densely dashed]
  \tikzstyle{callline}=[very thick,draw=green!80!black,double]
  \def\sp{0.1}
  \def\face{+(0,-1) -- +(0,1)}
  \def\corner{circle(\sp * 0.707 * 2cm)}

  \draw [draw=none,use as bounding box] (-0.125cm,-2.75cm) rectangle (5.25cm,0.25cm);

  \draw [fill=white,drop shadow] (-0.5cm,-2.65cm) rectangle (5.25cm,0.75cm);

  \draw (-0.5cm,0.25cm) -- (5.25cm,0.25cm);

  \draw [yshift=0.5cm,xshift=2.375cm,scale=(1-\sp)*0.25] node {Legend};

  \draw [xshift=2.375cm,scale=(1-\sp)*0.25] node 
  {\texttt{Iterate\_interior} ($\oc c$, $\arry S$, \texttt{callback}):};
  \draw [yshift=-0.5cm,scale=(1-\sp)*0.25][iterline]\quadrant (2cm,-1cm) -- (2cm,1cm)
    (3cm,0cm) circle(\sp * 0.707 * 4cm) (4cm,0cm) node [right] {$\oc c$};
  \draw [xshift=1.75cm,yshift=-0.5cm,scale=(1-\sp)*0.25][unknownstyle]\quadrant
    (1.5cm,0cm) node [right] {$\arry S[i]$ (leaves in $\supp(\oc c)[i]$)};

  \draw [draw=gray,thin] (-0.5cm,-0.9cm) -- (5.25cm,-0.9cm);

  \draw [yshift=-1.2cm,xshift=2.375cm,scale=(1-\sp)*0.25]
  node {$\arry H_i \leftarrow$ \texttt{Split\_array} ($\mathbf{S}[i]$,
    $\supp(\oc c)[i]$)};

  \draw [draw=gray,thin] (-0.5cm,-1.5cm) -- (5.25cm,-1.5cm);

  \draw [yshift=-1.75cm,xshift=2.375cm,scale=(1-\sp)*0.25] node 
  {\texttt{callback} ($\oc c$, $\lsupp(\oc c)$):};
  \draw [yshift=-2.25cm,xshift=0.5cm,scale=(1-\sp)*0.25][callline]
    \quadrant (2cm,-1cm) -- (2cm,1cm)
    (3cm,0cm) circle(\sp * 0.707 * 4cm) (4cm,0cm) node [right]
    {$\oc c$};
  \draw [xshift=2.625cm,yshift=-2.25cm,scale=(1-\sp)*0.25][leafstyle] \quadrant
    (1.5cm,0cm) node [right] {$\lsupp(\oc c)[i]$};
  
\end{tikzpicture}
}

}         &
    \multicolumn{1}{c}{\raisebox{.125\height}{{\small
\usetikzlibrary{trees}
\usetikzlibrary{quadtree}
\begin{tikzpicture}[line join=round,scale=1.2]

  \tikzstyle{unknownstyle}=[dotted]
  \tikzstyle{positionstyle}=[draw=yellow,fill=yellow!20]
  \tikzstyle{leafstyle}=[draw=blue!10,ultra thin,fill=blue!10]
  \tikzstyle{augstyle}=[draw=green,fill=green!20]
  \tikzstyle{branchstyle}=[draw=red,dashed,fill=red!20]
  \tikzstyle{iterline}=[draw=red,densely dashed]
  \tikzstyle{callline}=[very thick,draw=green!80!black,double]
  \def\sp{0.1}
  \def\face{+(0,-1) -- +(0,1)}
  \def\corner{circle(\sp * 0.707 * 2cm)}


  \draw [xshift=-1cm,scale=(1-\sp)] [unknownstyle]\quadrant
     node {$\supp(\oc c)[0]$} (0,-1) node [below] {$\arry S[0]$};
  \draw [xshift=1cm,scale=(1-\sp)] [unknownstyle]\quadrant
     node {$\supp(\oc c)[1]$} (0,-1) node [below] {$\arry S[1]$};
  \draw [scale=(1-\sp)] [iterline]\face;
  \path (0,-1) node [below] {$\oc c$};

\end{tikzpicture}
}

}}
    \\ \cline{3-4}
    \multicolumn{2}{c|}{}                                       &
    \multicolumn{2}{c|}{$\arry S[0] = \{\supp(\oc c)[0]\}?$}    \\ \cline{3-4}
    \multicolumn{2}{c|}{}                                       &
    \textbf{true} & \textbf{false}                              \\ \hline
    \parbox[t]{3mm}{\multirow{2}{*}[-2.5em]{\rotatebox[origin=c]{90}%
    {$\arry S[1] = \{\supp(\oc c)[1]\}?$}}} &
    \parbox[t]{3mm}{\rotatebox[origin=c]{90}{\textbf{true}}}    &
    \raisebox{-.5\height}{{\small
\usetikzlibrary{trees}
\usetikzlibrary{quadtree}
\begin{tikzpicture}[line join=round,scale=1.2]

  \tikzstyle{unknownstyle}=[dotted]
  \tikzstyle{positionstyle}=[draw=yellow,fill=yellow!20]
  \tikzstyle{leafstyle}=[draw=blue!10,ultra thin,fill=blue!10]
  \tikzstyle{augstyle}=[draw=green,fill=green!20]
  \tikzstyle{branchstyle}=[draw=red,dashed,fill=red!20]
  \tikzstyle{iterline}=[draw=red,densely dashed]
  \tikzstyle{callline}=[very thick,draw=green!80!black,double]
  \def\sp{0.1}
  \def\face{+(0,-1) -- +(0,1)}
  \def\corner{circle(\sp * 0.707 * 2cm)}

  \draw[draw=none,use as bounding box](-1.9cm,-1.4cm) rectangle (1.9cm,1.5cm);
  \draw [xshift=-1cm,scale=(1-\sp)] [leafstyle]\quadrant
     node {$\lsupp(\oc c)[0]$};
  \draw [xshift=1cm,scale=(1-\sp)] [leafstyle]\quadrant
     node {$\lsupp(\oc c)[1]$};
  \draw [scale=(1-\sp)] [callline]\face;

\end{tikzpicture}
}

}       &
    \raisebox{-.5\height}{{\small
\usetikzlibrary{trees}
\usetikzlibrary{quadtree}
\begin{tikzpicture}[line join=round,scale=1.2]

  \tikzstyle{unknownstyle}=[dotted]
  \tikzstyle{positionstyle}=[draw=yellow,fill=yellow!20]
  \tikzstyle{leafstyle}=[draw=blue!10,ultra thin,fill=blue!10]
  \tikzstyle{augstyle}=[draw=green,fill=green!20]
  \tikzstyle{branchstyle}=[draw=red,dashed,fill=red!20]
  \tikzstyle{iterline}=[draw=red,densely dashed]
  \tikzstyle{callline}=[very thick,draw=green!80!black,double]
  \def\sp{0.1}
  \def\face{+(0,-1) -- +(0,1)}
  \def\corner{circle(\sp * 0.707 * 2cm)}

  \draw[draw=none,use as bounding box](-1.9cm,-1.4cm) rectangle (1.9cm,1.5cm);

  \begin{scope} [xshift=-1cm,scale=(1-\sp)]
    \draw [morton order,quadtree absolute space=\sp] coordinate
      child [missing]
      child { [leafstyle]\quadrant
              { node [left=-.45cm]
                {$\lsupp(\oc c)[0]$}}}
      child [missing]
      child { [leafstyle]\quadrant
              { node [left=-.45cm]
                {$\lsupp(\oc c)[1]$}}};
  \end{scope}
  \draw [xshift=1cm,scale=(1-\sp)] [leafstyle]\quadrant
     node {$\lsupp(\oc c)[2]$};
  \draw [scale=(1-\sp)] [callline]\face;
\end{tikzpicture}
}

}       \\ \cline{2-4}
    &
    \parbox[t]{3mm}{\rotatebox[origin=c]{90}{\textbf{false}}}   &
    \raisebox{-.5\height}{{\small
\usetikzlibrary{trees}
\usetikzlibrary{quadtree}
\begin{tikzpicture}[line join=round,scale=1.2]

  \tikzstyle{unknownstyle}=[dotted]
  \tikzstyle{positionstyle}=[draw=yellow,fill=yellow!20]
  \tikzstyle{leafstyle}=[draw=blue!10,ultra thin,fill=blue!10]
  \tikzstyle{augstyle}=[draw=green,fill=green!20]
  \tikzstyle{branchstyle}=[draw=red,dashed,fill=red!20]
  \tikzstyle{iterline}=[draw=red,densely dashed]
  \tikzstyle{callline}=[very thick,draw=green!80!black,double]
  \def\sp{0.1}
  \def\face{+(0,-1) -- +(0,1)}
  \def\corner{circle(\sp * 0.707 * 2cm)}

  \draw[draw=none,use as bounding box](-1.9cm,-1.4cm) rectangle (1.9cm,1.5cm);

  \draw [xshift=-1cm,scale=(1-\sp)] [leafstyle]\quadrant
     node {$\lsupp(\oc c)[0]$};
  \begin{scope} [xshift=1cm,scale=(1-\sp)]
    \draw [morton order,quadtree absolute space=\sp] coordinate
      child { [leafstyle]\quadrant
              { node [right=-.4cm]
                {$\lsupp(\oc c)[1]$}}}
      child [missing]
      child { [leafstyle]\quadrant
              { node [right=-.4cm]
                {$\lsupp(\oc c)[2]$}}};
  \end{scope}
  \draw [scale=(1-\sp)] [callline]\face;
\end{tikzpicture}
}

}       &
    \raisebox{-.5\height}{{\small
\usetikzlibrary{trees}
\usetikzlibrary{quadtree}
\begin{tikzpicture}[line join=round,scale=1.2]

  \tikzstyle{unknownstyle}=[dotted]
  \tikzstyle{positionstyle}=[draw=yellow,fill=yellow!20]
  \tikzstyle{leafstyle}=[draw=blue!10,ultra thin,fill=blue!10]
  \tikzstyle{augstyle}=[draw=green,fill=green!20]
  \tikzstyle{branchstyle}=[draw=red,dashed,fill=red!20]
  \tikzstyle{iterline}=[draw=red,densely dashed]
  \tikzstyle{callline}=[very thick,draw=green!80!black,double]
  \def\sp{0.1}
  \def\face{+(0,-1) -- +(0,1)}
  \def\corner{circle(\sp * 0.707 * 2cm)}

  \draw[draw=none,use as bounding box](-1.9cm,-1.4cm) rectangle (1.9cm,1.5cm);

  \begin{scope} [xshift=-1cm,scale=(1-\sp)]
    \draw [morton order,quadtree absolute space=\sp] coordinate
      child { [unknownstyle]\quadrant}
      child { [unknownstyle]\quadrant
              { (0,-1) node [below] {$\arry H_0[1]$}}}
      child { [unknownstyle]\quadrant}
      child { [unknownstyle]\quadrant
              { (0,1) node [above] {$\arry H_0[3]$}}};
  \end{scope}
  \begin{scope} [xshift=1cm,scale=(1-\sp)]
    \draw [morton order,quadtree absolute space=\sp] coordinate
      child { [unknownstyle]\quadrant
              {  (0,-1) node [below] {$\arry H_1[0]$}}}
      child { [unknownstyle]\quadrant}
      child { [unknownstyle]\quadrant
              { (0,1) node [above] {$\arry H_1[2]$}}}
      child { [unknownstyle]\quadrant};
  \end{scope}
  \draw [yshift=0.5cm,scale=0.5*(1-\sp)^2] [iterline]\face;
  \draw [yshift=-0.5cm,scale=0.5*(1-\sp)^2] [iterline]\face;
  \draw [scale=0.5] [iterline]\corner;

\end{tikzpicture}
}

}       \\ \hline
  \end{tabular}
  \caption{%
    Illustrations of \fxn{Iterate\_interior} for $\dim(\oc c)=d$ (top) and
    $0 < \dim(\oc c) < d$ (bottom).
    Dashed red lines indicates the argument point $\oc c$ of
    \fxn{Iterate\_interior}.  The dotted squares indicate the arrays $\arry
    S[i]$ of leaves in $\cO$ that are descendants of $\supp(\oc c)[i]$.  If
    $\supp(\oc c)[i]$ is a leaf (which is determined by testing whether $\arry
    S[i] = \{\supp(\oc c)[i]\}$), then it is in $\lsupp(\oc c)$, which is
    shown with the solid color blue;
    otherwise, $\arry S[i]$ is split using \fxn{Split\_array}.  If a leaf has
    been found, the user-supplied callback function is executed, which we
    indicate with double color green lines; otherwise, \fxn{Iterate\_interior} is
    called for each point in the child partition set $\part(\oc c)$: the
    support sets for these points are found in the children of the octants in
    $\supp(\oc c)$, and the arrays of leaves contained in them
    are found in the sets created by \fxn{Split\_array}.%
  }%
  \figlab{itervolfig}%
\end{figure}%
\begin{figure}\centering
  \includegraphics{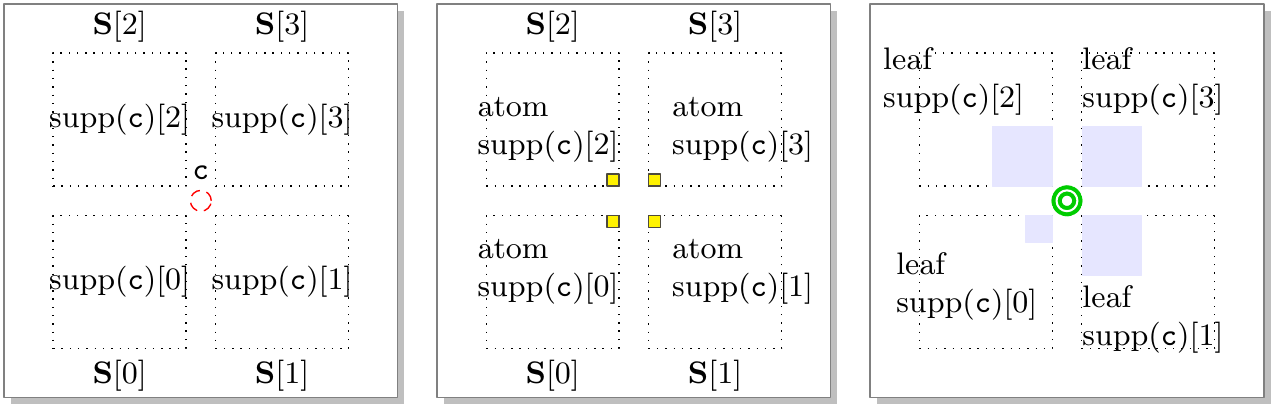}
  \caption{%
    An illustration of \fxn{Iterate\_interior} when $\dim(\oc c)=0$, using the
    same color conventions as \figref{itervolfig}.  The arguments of
    \fxn{Iterate\_interior} are in the left panel.  The small squares (middle
    panel) indicate octants  in $\asupp(\oc c)$: these must be descendants of
    the leaves in $\lsupp(\oc c)$, so we use $\asupp(\oc c)[i]$ as a key to
    search for $\lsupp(\oc c)[i]$ in $\arry S[i]$.  Once $\lsupp(\oc c)$ is
    found, the callback is executed (right panel).
  }\figlab{figitercorner}
\end{figure}

\begin{remark}\upshape
  \remlab{iterorder}
  An instance of \fxn{Iterate\_interior} may call multiple recursive copies of
  itself: one for each point in the child partition set $\part(\oc c)$ (see
  the loop starting on line \ref{interiorloop}).  We have not yet specified an
  order for these recursive calls.  In our implementation, we have chosen to
  order these calls by decreasing point dimension.  This guarantees that, if
  $\oc e\in \bound(\oc c)$, then the callback is executed for $\oc c$ before
  it is executed for $\oc e$.  We take advantage of this order in designing a
  node-numbering algorithm in \secref{lnodes}.
\end{remark}

\subsection{Iterating on a forest}

To iterate on the complete forest, each process must call
\fxn{Iterate\_interior} for the closure set of the root of every tree.  This
is listed in \fxn{Iterate} (\algref{iter}).  Asymptotic analysis of the
performance of \fxn{Iterate} is given in \appref{iteranalysis}: it shows that,
in general, \fxn{Iterate} executes in $\cO(((\max_{\oc o\in \cO} \oc o.\ix l)
+ N_{\pr
  p})\log N_{\pr p})$ time, but if the refinement pattern of the octrees in
the forest is uniform or nearly so, then it executes in $\cO(\log P + N_{\pr
  p})$ time.
%
\begin{algorithm}
  \caption{\fxn{Iterate} (\fxn{callback}, octant array $\arry{G}_{\pr p}^d$)}
  \alglab{iter}

  \Input{%
    $\arry G_{\pr p}^d$ is the full ghost layer \eqref{eqn:ghost}.
  }

  \Result{%
    if $\oc c\in\localpart$ (or $\overline{\localpart}$), $\lsupp_{\pr p}(\oc
    c)$ is computed and passed to \fxn{callback}.
  }

  \algorule

  \ForAll{$0 \leq t < K$}{%
    $\arry G^t \leftarrow \arry{G}_{\pr p}^d \cap \treeoctants$
    \Comment{%
      subset of the ghost layer for tree $t$ ($O(\log |\arry{G}_{\pr p}^d|)$)
    }

    $\arry S^t \leftarrow \localtreeoctants \cup \arry G^t$
    \Comment{%
      \ $\localtreeoctants$ and $\arry G^t$ are already ordered: no sorting is
      necessary%
    }

  }

  \ForAll{$\oc c\in \bigcup_{0\leq t < K} \clos(\Root(t))$}{%
    \ForAll{$0 \leq i < |\supp(\oc c)|$}{%

      $t\leftarrow \supp(\oc c)[i].\ix t$
      \Comment{Each $\supp(\oc c)[i]$ is the root of an octree}

      $\arry {U}[i] \leftarrow \arry S^t$
    }
    \fxn{Iterate\_interior} ($\oc c$, $\arry U$, \fxn{callback})
  }
\end{algorithm}

\subsection{Implementation} 
\seclab{iterimpl}

The implementation of \fxn{Iterate} in \pforest has some differences from the
presented algorithm to optimize performance.  \fxn{Iterate\_interior} is
implemented in while-loop form to keep the stack from growing: all space
needed for recursion (which is $O(\lmax)$) is pre-allocated on the heap.  We
noticed in early tests that \fxn{Split\_array} can be called with the same
arguments multiple times during a call to \fxn{Iterate}.  To avoid some of
this recomputation, we keep an $O(\lmax)$ fixed-size cache of the array
aliases produced by \fxn{Split\_array}.  We allow the user to specify a
separate callback function for each dimension, so that extra recursion can be
avoided.  If, for example, the callback only needs to be executed for faces,
then an instance of \fxn{Interate\_interior} operating on $\oc c$ will only
call a recursive copy for $\oc e\in \part(\oc c)$ if $\dim(\oc
e)\geq d-1$.

\section{A use case for the iterator: higher-order nodal basis construction}

\seclab{lnodes}

Thus far, we have developed algorithms for distributed forests
with no special regard for numerical applications.  In this section, we use
our framework to perform a classic but complex task necessary for finite
element computations, namely the globally unique numbering of degrees of
freedom for a continuous finite element space over hanging-node meshes.  We
call it \fxn{Lnodes} in reference to (Gau\ss-)Lobatto nodes, which means that
some nodes are located on element boundaries and are thus shared between
multiple elements and/or processes, which presents some interesting
challenges.

Hanging-node data structures have been discussed as early as 1980
\cite{RheinboldtMesztenyi80} and adapted effectively for higher-order spectral
element computations \cite{FischerKruseLoth02, SertBeskok06}.
Special-purpose data structures and interface routines have been defined for
many discretization types built on top of octrees, including piecewise linear
tensor-product elements \cite{AkcelikBielakBirosEtAl03,
BursteddeStadlerAlisicEtAl13} and discontinuous spectral elements
\cite{WilcoxStadlerBursteddeEtAl10}.  The \dealii finite element software
\cite{BangerthHartmannKanschat07} uses yet another mesh interface
\cite{BangerthBursteddeHeisterEtAl11}.  In our presentation of \fxn{Lnodes},
we hope to show that the \fxn{Iterate} approach is sufficiently generic that
it could be used to efficiently construct any of these data structures.

\subsection{Concepts}



%
In a hexahedral $n$-order nodal finite element, the Lagrangian basis functions
and the degrees of freedom are associated with $(n + 1)^3$ \emph{$Q^n$-nodes}
located on a tensor grid of locations in the element.  There is one node at
each corner, $(n - 1)$ nodes in the interior of each edge, $(n - 1)^2$ nodes
in the interior of each face, and $(n - 1)^3$ nodes in the interior of the
element, as in \figref{hanging2} (left).  If we endow each leaf in a forest of
octrees with $Q^n$-nodes, we get $N \times (n+1)^3$ \emph{element nodes}.
$Q^n$-nodes are numbered lexicographically, and element-local nodes are then
numbered to match the order of their associated leaves.  The basis functions
associated with the element nodes span a discontinuous approximation space.

\begin{figure}\centering
  \includegraphics{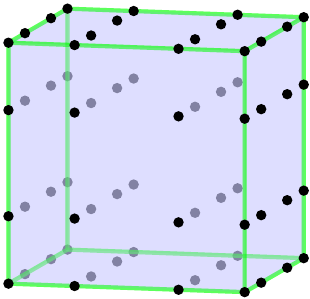}\hfill
  \includegraphics{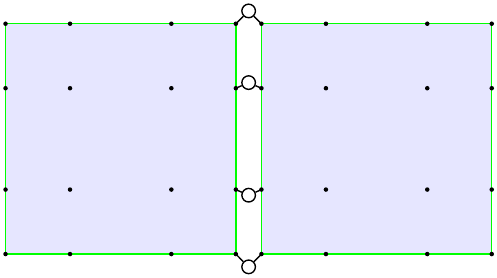}\hfill
  \includegraphics{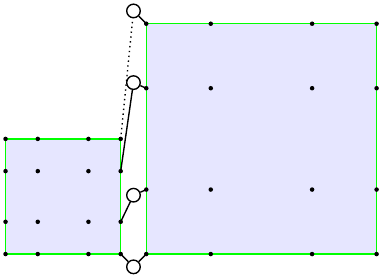}
  \caption{%
    (left) $Q^n$-nodes for $n=3$, with one node at each corner, $n-1$ nodes on
    (the interior of) each edge, $(n-1)^2$ nodes on each face, and $(n-1)^3$
    nodes in each element.
    (middle, right) For both conformal and non-conformal
    interfaces, each element node corresponds to exactly one global node.
    (right) At non-conformal interfaces, an element may reference a global node
    remotely, as the small element references the top node.%
  }%
  \figlab{hanging2}%
\end{figure}%

We want to create a nodal basis for a $C^0$-conformal approximation space on
$\Omega$ such that the restriction of the space to any leaf is spanned by the
$Q^n$-nodes' basis functions.  The nodes for the continuous basis functions
are called \emph{global nodes}.  Each of the element nodes in the interior of
a leaf can be associated with a unique global node, but on the boundary of a
leaf, element nodes from multiple leaves may occupy the same location: in this
case, the two element nodes are associated with the same global node, as in
\figref{hanging2} (middle).  For non-conformal interfaces, the element nodes
of the smaller leaves are not at the same locations as those of the larger
leaf, but they cannot introduce new degrees of freedom, because every function
in the space, when restricted to the non-conformal interface, must be
representable using the larger leaf's basis functions.  Conceptually, we can
place the global nodes at the locations of the larger leaf's nodes and
associate each element node from the smaller leaves with a single global node,
as shown in \figref{hanging2} (right).  In reality, the value of a function at
an element node on a non-conformal interface must be interpolated from the
values at multiple global nodes, but the conceptual one-to-one association
between a leaf's element nodes and global nodes is sufficient, in that it
identifies all global nodes whose basis functions are supported on that leaf.

It is important to note that an element node of a leaf $\oc o$ may reference a
global node that is contained in the domain of a point $\oc c$ that is outside
the closure $\Dom(\oc o)$, and that $\oc o$ is therefore not in the set of
adjacent leaves $\lsupp(\oc c)$ defined in \eqref{eqn:lsupp}.  In this
situation, we say that $\oc o$ \emph{remotely references} the global nodes in
$\oc c$.
%
Formally, a leaf $\oc o$ remotely references a point $\oc c$ in the
global partition set $\partoctants$ if
\begin{equation}\eqnlab{remote}
  \oc c\not\in\lsupp(\oc c)\text{ and }\exists \oc e
  \text{ such that }
  \oc o\in\lsupp(\oc e)\text{ and }\oc c\in\bound(\oc e).
\end{equation}
%
This relationship is also shown in \figref{hanging2} (right). 
From this definition, we can see that the global nodes referenced by leaves in
$\localoctants$ will be contained in the closed locally relevant partition
$\overline{\localpart}$ \eqref{eqn:closedpart}.  We note that a point $\oc
c$ can be remotely referenced only if $\dim(\oc c) < d - 1$.


\subsection{Data structures}

On process $\pr p$, we can represent the global nodes using an array $\arry
{N}_{\pr p}$ and the element nodes using a double array $\arry {E}_{\pr p}$,
where $\arry {E}_{\pr p}[j][k]$ maps the $k$th element node of
$\localoctants[j]$ to its global node in $\arry {N}_{\pr p}$.
$\arry {N}_{\pr p}$ and $\arry {E}_{\pr p}$ only reference locally relevant
global nodes and thus implement fully distributed parallelism.

In presenting the \fxn{Lnodes} algorithm, we consider a global node $\oc g$ to
have the following data fields:
%
  %
  \begin{itemize}
    \item
      $\oc g.i$:
      the globally unique \emph{index} of this node,
    \item
      $\oc g.\pr p$: the \emph{process} that owns $\oc g$ for the
      purposes of scatter/gather communication of node values,
    \item
      $\oc g.\st S_{\rshare}$:
      the \emph{sharer set} of all processes that reference this node.%
  \end{itemize}%
%
We include the sharer set $\st S_{\rshare}$ so that, in addition to the
scatter/gather communication paradigm, the global nodes can also be used in
the share-all paradigm, wherein any process that shares a node can send
information about that node to all other processes that share that node.%
\footnote{%
  In \pforest, the implementation differs: rather than the per-node lists $\oc
  g.\st S_{\rshare}$, we store lists of nodes that are shared by each
  process, which is a more useful layout for filling communication
  buffers.%
} %
If each process generates new information about a node, the former paradigm
requires two rounds of communication for information to disseminate, one
gather and one scatter, while the latter requires one round, but with an
increased number of messages.  Each paradigm can be faster than the other,
depending on communication latency, bandwidth, and other factors.  We tend to
place the highest weight on latency, hence our added support for share-all.

%

\begin{remark}\upshape
  Most applications do not require higher-order finite element nodes, but the
  \fxn{Lnodes} data structure can be used in much more general settings.  In
  particular, the \fxn{Lnodes} data structures for $n=2$ assign one unique
  global index to every point in $\overline{\mathcal{P}_{\pr p}}$, and a map
  from each leaf to the points in its closure.  If one symmetrizes these
  mappings, i.e., if one saves the leaf support sets $\lsupp_{\pr p}(\oc c)$
  for $\oc c\in\overline{\mathcal{P}_{\pr p}}$ generated by \fxn{Iterate},
  then one has essentially converted the forest-of-octrees data structures
  into a graph-based unstructured mesh format with $O(1)$ local topology
  traversal.  This format is typical of generic finite element libraries.
  \fxn{Lnodes} can therefore serve as the initial step in converting a forest
  of octrees into the format of an external library, with the remaining steps
  requiring little or no communication between processes. 
\end{remark}

\subsection{Assigning global nodes}
\seclab{gnumsec}

We want global nodes to be numbered independently of the number of processes
$P$ and the partition of the leaves.
For this reason, it is useful for each global node to be owned by one leaf,
because a partition-independent order is then induced by combining
lexicographic ordering of element nodes with the total octant order (see
\algref{comparison}).  This computation is shown in \algref{simplenum}, which
assumes that we have already determined which leaf $\localoctants[j]$ owns
each global node $\oc g$, and temporarily stored that leaf's index in the
global node's $\oc g.i$ field.

\begin{algorithm}
  \caption{\fxn{Global\_numbering} (\dtype{global node array}\ $\arry{N}_{\pr p}$,
                                    \dtype{double array} $\arry{E}_{\pr p}$)}
  \alglab{simplenum}

  \Input{%
    $\forall\oc g\in\arry{N}_{\pr p}$, $\oc g.\pr p=\pr q$ such that the leaf
    that owns $\oc g$ is in $\cO_{\pr q}$; if $\oc g.\pr p = \pr p$, then $\oc
    g.\st S_{\rshare}$ is correct and $\oc g.i=j$ is (temporarily abused as)
    the index of the leaf that owns $\oc g$, $\localoctants[j]$.%
  }

  \Result{%
    correct global node data for all $\oc g\in\arry{N}_{\pr p}$.
  }

  \algorule 

  integer array $\arry M[|\arry N_{\pr p}|]$
  \Comment{temporarily stores the local indices of global nodes}


  $m \leftarrow 0$
  \Comment{the number of global nodes owned by $p$}

  \For{$j=0$ \algto $|\localoctants| - 1$}{%
    \For{$l = 0$ \algto $(n+1)^d - 1$}{%
      $k\leftarrow\arry{E}_{\pr p}[j][l]$
      \Comment{index in $\arry N_{\pr p}$ of the global node referenced by
        this element node}

      \If{$\arry N_{\pr p}[k].\pr p=\pr p$ \algand $\arry N_{\pr p}[k].i=j$}{%
        $\arry M[k]\leftarrow m$
        \Comment{$\oc g$'s index among the locally owned nodes}

        $m\leftarrow m + 1$
      }
    }
  }

  $\arry t\leftarrow$ \fxn{Prefix\_sums}(\fxn{Allgather}($m$))
  \Comment{%
    \ $\arry t[\pr q]$ is the offset to the first node owned by $\pr q$%
  }

  \ForAll{$0 \leq k < |\arry N_{\pr p}|$}{%
    $\oc g\leftarrow \arry{N}_{\pr p}[k]$

    \eIf{$\oc g.\pr p=\pr p$}{%
      $\oc g.i \leftarrow \arry M[k] + \arry t[\pr p]$
      \Comment{all fields of $\oc g$ are now complete}
      \label{gnumsend}

      send $\oc g$ to each $\pr q\in \oc g.\st S_{\rshare}$
      \Comment{in practice, grouped into one message per process}
    }{%
      receive updated $\oc g$ from $\oc g.\pr p$
    }
  }
\end{algorithm}%

We need a policy for assigning the ownership of nodes to leaves.
%
We assign point $\oc c\in\partoctants$ and its
nodes to the first leaf $\oc o$ in $\lsupp(\oc c)$,
\begin{equation}\eqnlab{owner}
  \owner(\oc c):=\min\ \lsupp(\oc c).
\end{equation}
%
%
In the next subsection, we will show how this assignment policy allows for the
global nodes to be constructed without any more communication between
processes than the communication in \fxn{Global\_numbering}.

\subsection{The \fxn{Lnodes} algorithm}
\seclab{lnodesalg}

The \fxn{Lnodes} algorithm (\algref{lnodes}) creates the global nodes and
element nodes by iterating a callback \fxn{Lnodes\_callback} over all points
in the closed locally relevant partition $\overline{\localpart}$, which sets
up a call to \fxn{Global\_numbering} (\algref{simplenum}).

\begin{algorithm}
  \caption{\fxn{Lnodes} (integer $n$, ghost layer $\arry{G}_{\pr p}^d$)}
    \alglab{lnodes}

  \Input{%
    $n>0$, the order of the nodes; 
    the full ghost layer $\arry{G}_{\pr p}^d$ \eqref{eqn:ghost}.
  }
  \Result{%
    a double array $\arry E_{\pr p}$ of $N_{\pr p}\times (n+1)^d$ indices that
    maps element nodes to global nodes; 
    an array $\arry N_{\pr p}$ of global nodes.
  }

  \algorule

  global node array $\arry N_{\pr p}\leftarrow\emptyset$

  integer double array $\arry E_{\pr p}[|\localoctants|][(n + 1)^d]$

  \fxn{Iterate} (\fxn{Lnodes\_callback}, $\arry G_{\pr p}^d$)
  \Comment{initialize $\arry N_{\pr p}$, $\arry E_{\pr p}$}

  \fxn{Global\_numbering} ($\arry{N}_{\pr p}$, $\arry{E}_{\pr p}$)
  \Comment{finalize $\arry{N}_p$}
\end{algorithm}

Given the assumptions of \fxn{Global\_numbering}, \fxn{Lnodes\_callback} has
to accomplish the following for each global node $\oc g$ located at a point
$\oc c\in\overline{\localpart}$ visited by \fxn{Iterate}:
\begin{enumerate}
  \item
    determine the process $\oc g.\pr p$;
  \item
    if $\oc g.\pr p = \pr p$, determine the index $j$ of $\owner(\oc c)$
    in $\localoctants$;
  \item
    if $\oc g.\pr p = \pr p$, determine the set $\oc g.\st S_{\rshare}$ of
    processes that share $\oc g$;
  \item
    complete the entries in $\arry E_{\pr p}$ that refer to $\oc g$.
\end{enumerate}
We will not give pseudocode for \fxn{Lnodes\_callback} here; we only hope to
convince the reader that the information that \fxn{Iterate} supplies to the
\fxn{Lnodes\_callback}---i.e., the local leaf support set $\lsupp_{\pr p}(\oc
c)$ for each point $\oc c$ in the closed locally relevant partition
$\overline{\localpart}$---is sufficient to accomplish the above listed tasks.

\paragraph{1. Determine $\oc g.\pr p$}
The policy that defines $\owner(\oc c)$ \eqref{eqn:owner} guarantees that $\oc
c$ and its nodes will be owned by a leaf in the partition of the first process
$\pr q$ such that $\overline{\Omega_\pr q}$ intersects $\Dom(\oc c)$.  This
means that each process $\pr p$ that references $\oc c$ can determine the
processes that own all of the nodes it references, even if $\lsupp_{\pr p}(\oc
c)$ is incomplete (see \algref{determine}).

\begin{algorithm}
  \caption{\fxn{Determine\_owner\_process} (point $\oc c$)}
  \alglab{determine}
  \Result{%
    if $\oc c\in\partoctants$, the process that owns $\oc c$ and its global
    nodes.
  }

  \algorule

  \For%
  (\Comment{this set always contains one point})%
  {%
    $\oc e \in \{\oc c\}\cup\part(\oc c)$ \textbf{such that}
    $\dim(\oc e) = 0$}{
    \Return{$\min \{\locate(\oc a):\oc a\in\asupp(\oc e))\}$}
  }
\end{algorithm}

\paragraph{2. If $\oc g.\pr p=\pr p$, determine $\owner(\oc c)$}
Suppose $\owner(\oc c)=\oc o\in\localoctants$: $\oc o$ and $\oc c$ intersect,
$\overline{\Dom(\oc o)}\cap \Dom(\oc c) \neq \emptyset$.  By definition
\eqref{eqn:lsupp}, $\oc c$ must be in $\localpart$, so by
Proposition~\ref{prop:completeinfo}, $\lsupp_{\pr p}(\oc c)=\lsupp (\oc c)$.
Therefore, $\owner(\oc c)$ will be in $\lsupp_{\pr p}(\oc c)$, and its index in
$\localoctants$ was calculated by \fxn{Iterate}, so we can set $\oc g.i$ equal
to that index for each global node $\oc g$ located at $\oc c$.

\paragraph{3. If $\oc g.\pr p=\pr p$, determine $\oc g.\st S_{\rshare}$}%
We use both Proposition~\ref{prop:completeinfo} and the 2:1 balance condition
to design an algorithm called \fxn{Reconstruct\_remote} to reconstruct the
octant data for all leaves that remotely reference \eqref{eqn:remote} the
point $\oc c$ and its nodes (\algref{reconstruct}, \figref{reconstruct}).  By
locating the processes that overlap the octants in $\lsupp_{\pr p}(\oc c)$ and
the octants returned by \fxn{Reconstruct\_remote}($\oc c$, $\lsupp_{\pr
  p}(\oc c))$, process $\pr p$ can determine all processes that reference
$\oc c$'s nodes.

\begin{figure}\centering
  \includegraphics{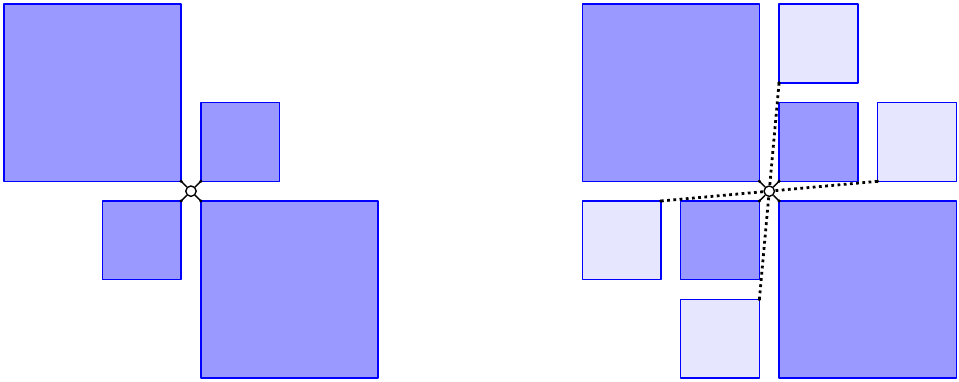}
  \caption{%
    Because of the 2:1 condition, \fxn{Reconstruct\_remote}
    (\algref{reconstruct}) can use the octants in $\lsupp(\oc c)$ (left) to
    reconstruct remotely referencing leaves (right).
  }
  \figlab{reconstruct}
\end{figure}

\begin{algorithm}
  \caption{\fxn{Reconstruct\_remote} (point $\oc c$, octant set $\lsupp(\oc
    c)$)}
  \alglab{reconstruct}
  \Input{%
    the leaf support set $\lsupp(\oc c)$ \eqref{eqn:lsupp}.
  }
  \Result{%
    the set $\st R$ of all octants that are leaves that remotely reference $\oc
    c$.
  }

  \algorule

  $\st R\leftarrow\emptyset$

  \ForAll%
  (\Comment{%
    for every leaf's boundary point \dots
  })%
  {$\oc e\in\bigcup \bound(\lsupp(\oc c))$}{%
    \If%
    (\Comment{%
      \dots that is adjacent to $\oc c$
    })%
    {$\oc c\in\bound(\oc e)$}{%
      \ForAll%
      (\Comment{%
        for every octant $\oc s$ adjacent to a child of $\oc e$
      })%
      {$\oc s\in \bigcup \supp(\child(\oc e))$}{%
        \If%
        (\Comment{%
          if $\oc s$ is not a descendant of a leaf
        })%
        {\algnot $\oc s\in \bigcup\desc(\lsupp(\oc c))$}{%
          $\st R\leftarrow \st R \cup \{\oc s\}$
          \Comment{$s$ remotely references $\oc c$ \eqref{eqn:remote}}
        }
      }
    }
  }

  \Return{$\st R$}
\end{algorithm}

\paragraph{4. Complete the references to $\oc g$ in $\arry{E}_{\pr p}$}%
For each leaf $\oc o\in \lsupp_{\pr p}(\oc c)\cap\localoctants$ \fxn{Iterate}
provides the index $j$ in $\localoctants$ such that $\localoctants[j]=\oc o$,
so determining the values of $k$ such that $\arry{E}_{\pr p}[j][k]$ refers to
the global nodes at point $\oc c$ is a matter of comparing the orientation of
$\oc o$ and $\oc c$ relative to each other.
For each leaf $\oc r\in\localoctants$ that remotely references $\oc c$ and its
nodes, by definition~\eqref{eqn:remote} there must be a point $\oc e$ such
that $\oc c\in\bound(\oc e)$ and $\oc r\in\lsupp(\oc e)$.  The instance of
\fxn{Lnodes\_callback} called for the point $\oc e$ has the index $j$ such
that $\localoctants[j]=\oc r$ and can ``link'' it to $\oc c$, so that the
correct $\arry{E}_{\pr p}[j][\star]$ entries can be completed (see
\remref{iterorder}).

The previously presented algorithm \fxn{Nodes} \cite[Algorithm
21]{BursteddeWilcoxGhattas11} produces data structures equivalent to those
produced by \fxn{Lnodes} for $n=1$.  The ownership rule in
\fxn{Nodes}---associating each node with a unique level-$(\lmax + 1)$ octant
(i.e., allowing the octant data structure to be more refined than an atom
for storing nodes),
and assigning ownership based on the process whose range contains that
octant---is similar in principle to the ownership rule given in
\eqref{eqn:owner}.  \fxn{Nodes} does not have symmetric communication,
however, because it does not construct the neighborhood $\lsupp_{\pr p}(\oc
c)$ when it creates a node at $\oc c$, and so it cannot perform a calculation
like \fxn{Reconstruct\_remote}.  Since it does not deduce the presence of
remotely-sharing processes, \fxn{Nodes} requires a handshaking step, where the
communication pattern is determined.

\section{Performance evaluation}
\seclab{performance}

\begin{figure}
  \begin{center}
  \includegraphics[width=.38\columnwidth]{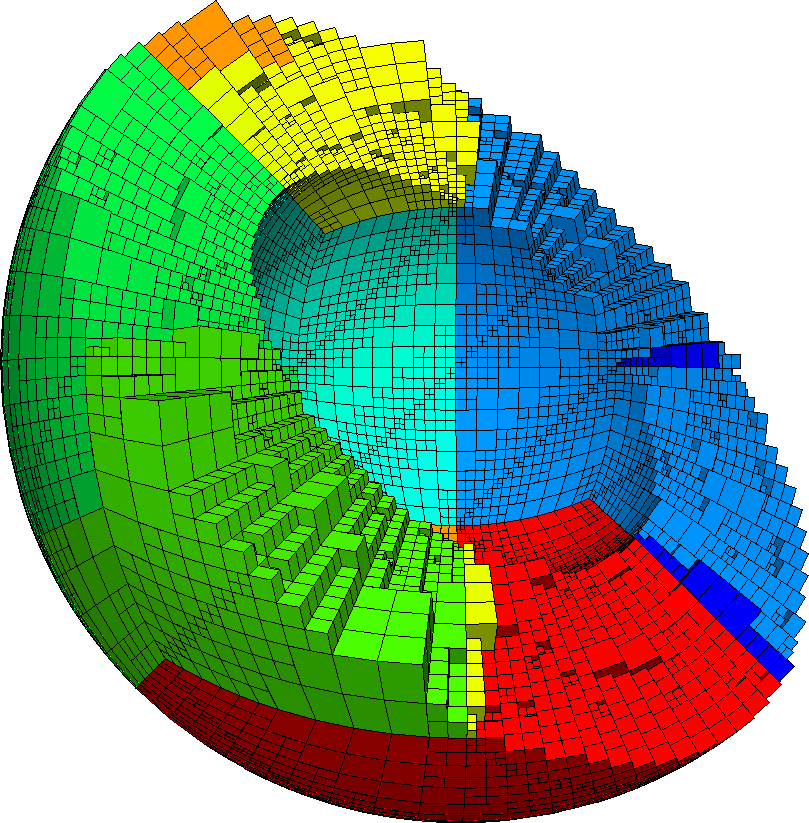}
  \hfill
  \includegraphics[width=.38\columnwidth]{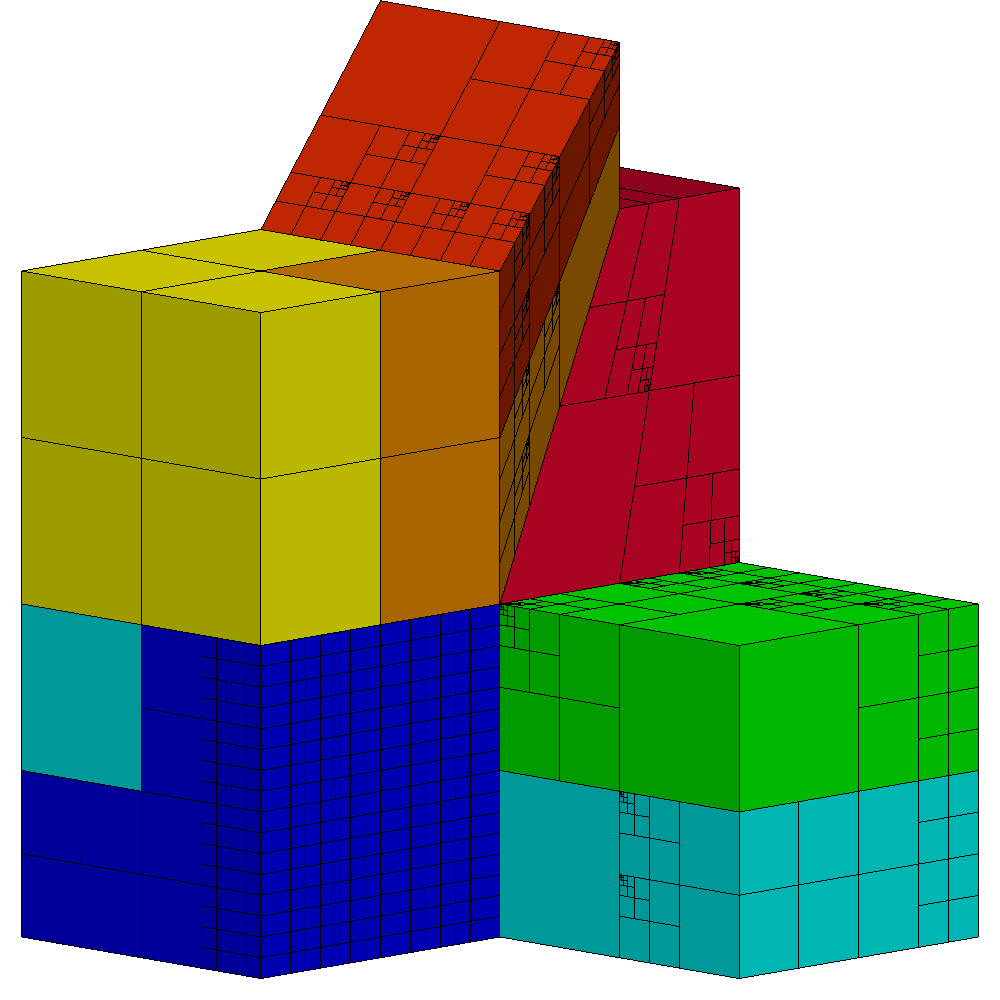}
  \end{center}
  \caption{%
    Example forests of octrees for nontrivial domain topologies.  (left)  A
    cutaway of a shell geometry, composed of 24 mapped octrees.  (right)  A
    collection of six rotated and mapped octrees connected in an irregular
    topology.  Both show adhoc refinement patterns; a 2:1-balance condition
    has been enforced in the left hand plot, but not the right.  Color
    indicates the partitions of different processes.%
  }%
  \figlab{connexamples}%
\end{figure}%

In this section we evaluate the efficiency and scalability of the algorithms
presented in this work as they have been implemented in \pforest.  The
parallel scalability is assessed on the Blue Gene/Q supercomputer JUQUEEN,
which is configured with 28,672 compute nodes, each with 16 GB of memory and
16 cores, for a total of 458,752 cores.  Additional concurrency is available
through simultaneous multithreading: each core has two instruction/integer
pipelines, and can issue one instruction to each of these pipelines per
cycle, as long as they come from different threads.  Where appropriate, we
will compare results for 16, 32, and 64 MPI processes per node (\pforest uses
only MPI for parallelism).
We have compiled the \pforest library and executables with IBM's XL C compiler
in version 12.1.

\subsection{Search}

To test the performance of \fxn{Search} (\algref{search}), we consider the
problem of identifying the leaves that contain a set of randomly generated
points.  We choose a spherical shell domain typical for simulations of earth's
mantle convection, with inner radius $r=0.55$ and outer radius $r=1$, as
illustrated in \figref{connexamples} (left).  For each test, we generate $M$
points at random, independently and uniformly distributed in the cube
containing the shell, and use \fxn{Search} as implemented by the \pforest
function \fxn{p8est\_search} to identify the leaves that contain them. 

The mappings $\varphi^t$ for the octrees \eqref{eqn:macro} are given
analytically for this domain.  In the callback that we provide to
\fxn{Search}, we have two tests to determine whether the mapped domain
$\Dom(\oc o)$ of an octant $\oc o$  contains a point $x$, one fast and
inaccurate in the sense of allowing false positive matches, the other slower
but accurate.  In the accurate test, the mapping $\varphi^t$ is inverted to
get the preimage $\xi$ of $x$, and a bounding box calculation determines
whether $\xi\in \dom(\oc o)$.  In the inaccurate test, the image $x_{\oc o}$
of the octant's center is computed, as well as an upper bound $r_{\oc o}$ on
the radius of the bounding sphere of $\Dom(\oc o)$, and we test whether
$|x-x_o|\leq r_o$.  In \fxn{Search} the accurate test is performed only when
$\oc o$ is a leaf. We perform our tests on a series of forests with increasing
numbers of leaves $N$, but with each forest refined so that the finest leaves
are four levels more refined than the coarsest.

\begin{figure}\centering
  \begin{tikzpicture}[baseline]
    \begin{groupplot}
      [
        group style={group size=2 by 1,horizontal sep=0.17cm,vertical sep=2cm},
        width=6.6cm,
        xmode=log,
        ymode=log,
        ymin=1.e-5,
        ymax=1.e3,
        ytickten={-2,-1,0,1},
        log basis y=64,
        xtickten={-3,-2,-1,0,1,2,3},
        log basis x=64,
        xticklabels={$64^{-3}$,,$64^{-1}$,,$64^1$,,$64^3$},
        xlabel=$M/P$,
        grid=both,
        cycle list name=p4est2search,
      ]

      \nextgroupplot[
        title={$M$ searches for one point,\\ runtime in seconds},
        title style={align=center,text width=6.4cm},
        only marks,
      ]

      \foreach \i/\j in {
        64/1,
        64/2,
        4k/1,
        4k/2,
        4k/3,
        256k/1,
        256k/2,
        256k/3%
      }%
      {
        \addplot table [
            x expr=\thisrowno{2}/\thisrowno{0},
            y index=3,
          ]
          {timings/tsearch-shell/tsrana.data.\i.d\j};
                  
      }
      \addplot[dotted,no markers,sharp plot,update limits=false] coordinates
      {(1.e-7,1.e-6) (1.e7,1.e8)};
      \addplot[dotted,no markers,sharp plot,update limits=false] coordinates
      {(1.e-5,1.e-6) (1.e9,1.e8)};
      \addplot[dotted,no markers,sharp plot,update limits=false] coordinates
      {(1.e-3,1.e-6) (1.e11,1.e8)};

      \nextgroupplot[
        title={One search for $M$ points,\\ runtime in seconds},
        title style={align=center,text width=6.4cm},
        yticklabels={},
        legend cell align=right,
        legend pos=outer north east,
        only marks,
      ]

      \foreach \i/\j/\iname\jname in {
        64/1/$P_0$/$D_0$,
        64/2/$P_0$/$D_1$,
        4k/1/$P_1$/$D_0$,
        4k/2/$P_1$/$D_1$,
        4k/3/$P_1$/$D_2$,
        256k/1/$P_2$/$D_0$,
        256k/2/$P_2$/$D_1$,
        256k/3/$P_2$/$D_2$%
      }%
      {
        \addplot table
          [
            x expr=\thisrowno{2}/\thisrowno{0},
            y index=4,
          ]
          {timings/tsearch-shell/tsrana.data.\i.d\j};
          \addlegendentryexpanded{\iname, \jname}
      }
      \addplot[dotted,no markers,sharp plot,update limits=false] coordinates
      {(1.e-7,1.e-6) (1.e7,1.e8)};
      \addplot[dotted,no markers,sharp plot,update limits=false] coordinates
      {(1.e-5,1.e-6) (1.e9,1.e8)};
      \addplot[dotted,no markers,sharp plot,update limits=false] coordinates
      {(1.e-3,1.e-6) (1.e11,1.e8)};
    \end{groupplot}
  \end{tikzpicture}
  \caption{%
    Scaling results for searching for $M$ points in a shell domain using
    \fxn{Search} as implemented by \fxn{p8est\_search}.  We examine various
    values of $M$, $P$, and $N$.  Three different values of $P$ are used:
    $P_0=64$, $P_1=64^2$, and $P_2=64^3 = 262144$.  Three
    different values of $N/P=D$ are used: $D_0\approx 15\mathrm{k}$, $D_1\approx
    122\mathrm{k}$, and $D_2 \approx 979\mathrm{k}$.  (left) $M$ separate calls of
    \fxn{Search} are used to find the $M$ points.  (right) One call of
    \fxn{Search} is used to find all $M$ points.  The dotted lines symbolize
    linear weak scaling; points on top of each other demonstrate the
    independence of the runtime from the local octant count $D_i$.  The
    largest number of octants reached is $2.568 \times 10^{11}$.
  }%
  \figlab{searchanalysis}
\end{figure}
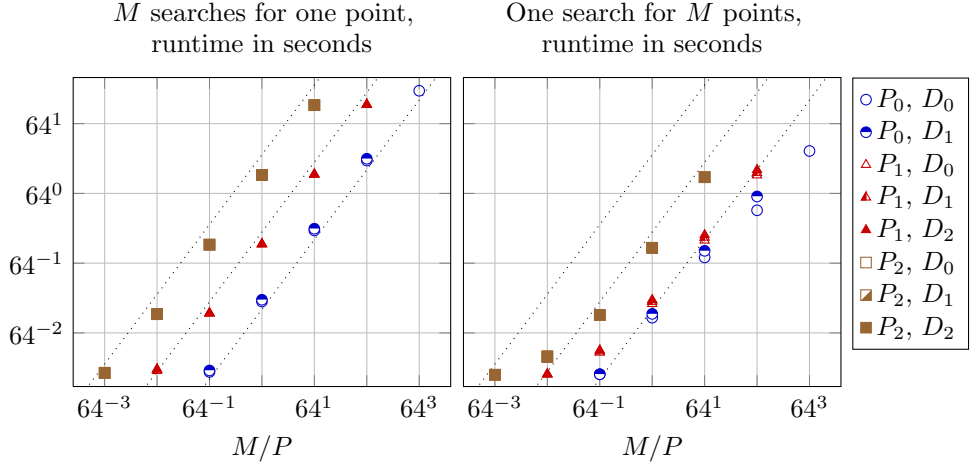

In \figref{searchanalysis}, we present the scaling results for our tests.
Each of the $P$ processes must determine which of the $M$ points are in its
partition.  This means that each process must perform the inaccurate test at
least $M$ times.  This is why, for fixed values of $P$, we see a scaling with
$O(M)$.  Indeed, the fraction of points that fall in a given processes
partition is on average $1/P$, so for large values of $P$ the majority of the
runtime is spent on points that are not in the partition.  This is why, in
\figref{searchanalysis}, the number of leafs on a node $N/P$ has so little
effect on the runtime.
When we take advantage of the algorithm's ability to
search for multiple points simultaneously, however, the setup costs of the
inaccurate test, such as computing the bounding radius $r_o$, can be amortized
over multiple comparisons.  Hence we see significant speedup when searching
for multiple points simultaneously: in \figref{searchanalysis}, we see that for
large values of $P$ and for $M/P\geq 1$ the simultaneous search is roughly $64$
times faster than searching for the same points individually.

\subsection{Ghost}

We test the performance of ghost layer construction as implemented by the
\pforest function \fxn{p8est\_ghost}, on the irregular geometry shown in
\figref{connexamples} (right).  We again create a series of meshes with
increasing $N$ and four levels of difference between the coarsest and finest
leaves.  Our ghost layer construction uses \fxn{Find\_range\_boundaries}
(\algref{findrange}) to determine which processes' partitions border an octant
$\oc o$.  When a partition $\Omega_{\pr p}$ with $N_{\pr p}$ leaves is
well-shaped, we expect $O(N_{\pr p}^{2/3})$ of those leaves to be on the
boundary of $\Omega_{\pr p}$, so we consider $O((N/P)^{2/3})$ to be ideal
scaling.  In the algorithm \fxn{Add\_ghost} (\algref{ghost2}), which is called
for each boundary leaf, $\log(P)$ work is performed to determine a set of
potentially neighboring processes; the remaining leaves in the interior of the
domain are skipped without calling \fxn{Add\_ghost}, and so they should
contribute very little to the runtime of \fxn{Ghost}.  We therefore expect our
performance to be $O((N/P)^{2/3}\log P)$.

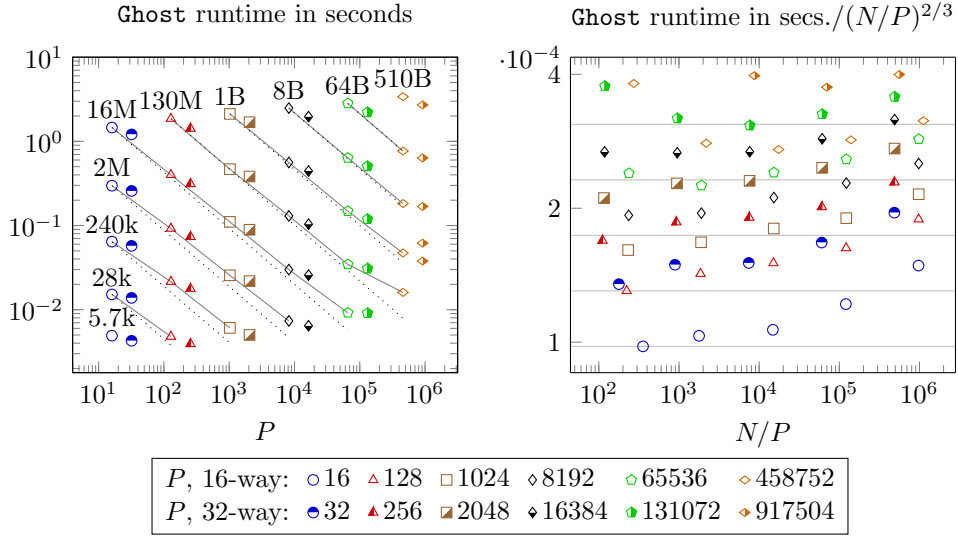
\begin{figure}\centering
  \begin{tikzpicture}[baseline]
    \begin{groupplot}
      [
        group style={group size=2 by 1,horizontal sep=1.5cm},
        width=6.7cm,
        xmode=log,
        ymode=log,
      ]

      \nextgroupplot[xmin=1,xmax=2,ymin=1,ymax=2,xtick={},ytick={},yticklabels={},xticklabels={},tick style={draw=none}]

      \nextgroupplot[
        axis y line*=right,
        axis x line=none,
        xmin=1.e1,
        xmax=1.e7,
        ytick={9.75e-5,1.3e-4,1.734e-4,2.312e-4,3.083e-4},
        ymajorgrids=true,
        yticklabels={},
        tick style={draw=none},
                    ]
      \addplot[draw=none] coordinates {(1.e1,3.81e-4) (1.e7,3.81e-4)};
      \addplot[draw=none] coordinates {(1.e1,9.75e-5) (1.e7,9.75e-5)};
    \end{groupplot}
    \begin{groupplot}
      [
        group style={group size=2 by 1,horizontal sep=1.5cm},
        width=6.7cm,
        xmode=log,
        ymode=log,
        cycle list name=p4est2,
      ]

      \nextgroupplot[
        title=\fxn{Ghost} runtime in seconds\strut,
        xlabel=$P$,
        xmin=4,
        ymax=10,
        legend columns=7,
        legend style={/tikz/every even column/.append style={column sep=5pt}},
        legend style={at={(1.11,-0.27)},anchor=north,legend cell align=left},
                    ]

      \addlegendimage{empty legend}
      \addlegendentry{$P$, 16-way:}

      \foreach \i in {16,128,1024,8192,65536,458752} {%
        \addplot+[only marks] table [x index=0, y index=6]
        {timings/lnodes-rotcubes-b9e7/timana.data.\i};
        \addlegendentryexpanded{$\i$}
      }

      \addlegendimage{empty legend}
      \addlegendentry{$P$, 32-way:}

      \foreach \i in {32,256,2048,16384,131072,917504} {%
        \addplot+[only marks] table [x index=0, y index=6]
        {timings/lnodes-rotcubes-b9e7/timana.data.\i};
        \addlegendentryexpanded{$\i$}
      }

      \addplot+[no markers,gray,solid] table [x expr=\thisrowno{0}, y index=6]
        {timings/lnodes-rotcubes-b9e7/timana.data.strong.ghost};

      \node [above] at (axis cs:16,0.005) {5.7k};
      \node [above] at (axis cs:16,0.015) {28k};
      \addplot[dotted,no markers,sharp plot,update limits=false] coordinates
      {(16,0.0151785) (128,0.0037946)};
      \node [above] at (axis cs:16,0.064) {240k};
      \addplot[dotted,no markers,sharp plot,update limits=false] coordinates
      {(16,0.063772) (1024,0.004024)};
      \node [above] at (axis cs:16,0.297) {2M};
      \addplot[dotted,no markers,sharp plot,update limits=false] coordinates
      {(16,0.29741) (8192,0.004647)};
      \node [above] at (axis cs:16,1.46) {16M};
      \addplot[dotted,no markers,sharp plot,update limits=false] coordinates
      {(16.,1.45934) (65536,0.0057005)};
      \node [above] at (axis cs:128,1.86) {130M};
      \addplot[dotted,no markers,sharp plot,update limits=false] coordinates
      {(128.,1.85936) (458752,0.007939)};
      \node [above] at (axis cs:1024,2.12) {1B};
      \addplot[dotted,no markers,sharp plot,update limits=false] coordinates
      {(1024.,2.11914) (458752,0.0362)};
      \node [above] at (axis cs:8192,2.48) {8B};
      \addplot[dotted,no markers,sharp plot,update limits=false] coordinates
      {(8192.,2.48478) (458752,0.1698)};
      \node [above] at (axis cs:65535,2.8) {64B};
      \addplot[dotted,no markers,sharp plot,update limits=false] coordinates
      {(65536.,2.82067) (458752,0.7708)};
      \node [above] at (axis cs:458752,3.4) {510B};

      \nextgroupplot[
        title=\fxn{Ghost} runtime in secs.$/(N/P)^{2/3}$,
        only marks,
        xlabel=$N/P$,
        ylabel near ticks,
        ylabel=$\cdot 10^{-4}$,
        ytick={1.e-4,2.e-4,4.e-4},
        yticklabels={$1$,$2$,$4$},
        axis y line*=left,
        y label style={rotate=-90,at={(rel axis cs:0.0,1.0)}},
                    ]

      \foreach \i in {16,128,1024,8192,65536,458752} {%
        \addplot table [x expr=\thisrowno{2}/\thisrowno{0},
        y expr=\thisrowno{6}/((\thisrowno{2}/\thisrowno{0})^(2./3.)) ]
        {timings/lnodes-rotcubes-b9e7/timana.data.\i}; } \foreach \i in
      {32,256,2048,16384,131072} {%
        \addplot table [x expr=\thisrowno{2}/\thisrowno{0},
        y expr=\thisrowno{6}/((\thisrowno{2}/\thisrowno{0})^(2./3.)) ]
        {timings/lnodes-rotcubes-b9e7/timana.data.\i}; } \addplot+[update
      limits=false] table [x expr=\thisrowno{2}/\thisrowno{0}, y
      expr=\thisrowno{6}/((\thisrowno{2}/\thisrowno{0})^(2./3.))]
      {timings/lnodes-rotcubes-b9e7/timana.data.917504}; \end{groupplot}
  \end{tikzpicture}
  \caption{%
    The scalability of ghost layer construction.  The meshes used are
    described in the text: the largest number of leaf octants $N$ is $5.1
    \times 10^{11}$.  (left) Runtime as a function of $P$: for 16-way process
    distribution, we compare strong scaling (solid lines) to ideal
    $O((N/P)^{2/3})$ scaling (dotted).  The total number of leaves $N$ in each
    mesh is indicated.  (right) Runtime scaled by $(N/P)^{2/3}$ as a function
    of $N/P$.  Weak-scaling is assessed by comparing the vertical distance
    between points: each grid line represents a $25\%$ loss
    of weak-scaling efficiency relative to the ideal $O((N/P)^{2/3})$.%
  }%
  \figlab{ghostscale}%
\end{figure}%
In \figref{ghostscale}, we give plots for assessing the strong- and
weak-scalability of ghost layer construction, relative to the ideal
$O((N/P)^{2/3})$ scaling. (For almost all problems, assigning 64 processes per
node was slower than 32, so we omit this data from the figure.) We see good
strong-scalability for $16\leq P\leq 65\mathrm{k}$ and $N/P \geq 1\mathrm{k}$.
For the full machine, when $P=459\mathrm{k}$ and $P=918\mathrm{k}$, the
communication latency and the small amount of $O(P)$ workspace and work in the
implementation (two scans of arrays of 32-bit integers) limits the efficiency
for $N/P \leq 10\mathrm{k}$.  In the weak-scaling plot, for the largest values
of $N/P$, we see that the relative efficiency (the efficiency of $(8N,8P)$
relative to $(N,P)$) improves slightly as $N$ and $P$ increase, which suggests
that we are seeing $O((N/P)^{2/3}\log P)$ scaling asymptotically. 



\subsection{Serial comparison of \fxn{Lnodes} and \fxn{Nodes}}

For polynomial degree $n=1$, the data structures constructed by \fxn{Nodes}
\cite[Algorithm 21]{BursteddeWilcoxGhattas11} and \fxn{Lnodes}
(\algref{lnodes}) are essentially equivalent.  For a general forest of octrees
on a single process, both have $\cO(N\log N)$ runtimes.  While
\fxn{Nodes} uses repeated binary searches and hash table queries and
insertions, \fxn{Lnodes} uses \fxn{Iterate} (\algref{iter}) to recursively
split the forest and operates on subsets of leaves.  This divide-and-conquer
approach should make better use of a typical cache hierarchy.  In this
subsection, we present a small experiment that confirms this fact.

The experiment is conducted on a single octree using a single process.  We
again create a series of meshes with increasing $N$ and four levels of
difference between the coarsest and finest leaves.  For each forest in the
series, we have three programs: one that calls \fxn{Nodes}, one that calls
\fxn{Lnodes}, and one that calls neither.  We use the Linux utility
\texttt{perf}\footnote{%
  \url{https://perf.wiki.kernel.org}%
} %
to estimate the number of instructions, cache misses, and branch prediction
misses in each program, calling each program 30 times to compensate for the
noise in \texttt{perf}'s sampling.  The averages of the events from the
program calling neither routine is subtracted from the other two averages,
giving an estimate of the events that can be attributed to the two routines.

The experiment is performed on a laptop with two Intel Ivy Bridge Core
i7-3517U dual core processors.  Each core has a 64 kB on-chip L1 cache, a 256
kB L2 cache, and each processor has a 4 MB L3 cache: \texttt{perf} counts L3
cache misses.  The \pforest library and the executable are compiled by
\texttt{gcc} 4.6.4 with \texttt{-O3} optimization.

The results of the experiment are given in \tabref{serialperf}.  The table
shows that the advantages of \fxn{Lnodes} over \fxn{Nodes} in terms of the
number of instructions and the number of branch misses do not grow much with
$N$, but the advantage in terms of cache misses grows from a factor of 2 on
the smallest problem size to a factor of 11 on the largest.
\begin{table}
  \caption{%
    Serial performance comparison of \fxn{Nodes} (top) and \fxn{Lnodes}
    (bottom) for $n=1$, as implemented by the \pforest functions
    \fxn{p8est\_nodes} and \fxn{p8est\_lnodes}, on a series of single-octree
    forests.%
  }%
  \tablab{serialperf}
  \centering
  \renewcommand{\arraystretch}{1.2}%
  \begin{tabular}{|l|l|l|l|l|}
    \hline
    $N$ & runtime (ms) & instructions & branch misses & cache misses
    \\
    \hline
    $4.6 \times 10^3$ & $9.5 \times 10^0$ & $4.3\times 10^7$    & $2.1 \times 10^5$ & $2.2 \times 10^4$ \\
                      & $1.0 \times 10^1$ & $3.7\times 10^7$    & $5.3 \times 10^4$ & $1.1 \times 10^4$ \\
    \hline                                                                         
    $3.9 \times 10^4$ & $8.6 \times 10^1$ & $4.2\times 10^8$    & $1.7 \times 10^6$ & $2.2 \times 10^5$ \\
                      & $4.0 \times 10^1$ & $3.1\times 10^8$    & $3.6 \times 10^5$ & $5.1 \times 10^4$ \\
    \hline                                                                         
    $3.2 \times 10^5$ & $8.4 \times 10^2$ & $3.7\times 10^9$    & $1.3 \times 10^7$ & $4.8 \times 10^6$ \\
                      & $3.5 \times 10^2$ & $2.5\times 10^9$    & $2.7 \times 10^6$ & $4.5 \times 10^5$ \\
    \hline                                                                         
    $2.6 \times 10^6$ & $8.0 \times 10^3$ & $3.3\times 10^{10}$ & $1.0 \times 10^8$ & $6.1 \times 10^7$ \\
                      & $2.8 \times 10^3$ & $2.0\times 10^{10}$ & $2.2 \times 10^7$ & $5.4 \times 10^6$ \\
    \hline

  \end{tabular}
\end{table}

\subsection{Parallel scalability of \fxn{Lnodes}}

In the previous subsection we compared the per-process efficiency of
\fxn{Lnodes} and \fxn{Nodes}.  Here we compare their parallel scalability on
the same series of test forests used to test \fxn{Ghost} above.

\begin{figure}\centering
  \begin{tikzpicture}[baseline]
    \hspace{0.2cm}
    \begin{groupplot}
      [
        group style={group size=2 by 1,horizontal sep=1.5cm},
        width=6.7cm,
        xmode=log,
        ymode=log,
      ]

      \nextgroupplot[xmin=1,xmax=2,ymin=1,ymax=2,xtick={},ytick={},yticklabels={},xticklabels={},tick style={draw=none}]

      \nextgroupplot[
        axis y line*=right,
        axis x line=none,
        xmin=1.e1,
        xmax=1.e7,
        ytick={2.03e-5,2.72e-5,3.6e-5,4.8e-5,6.4e-5,8.54e-5,1.14e-4,1.52e-4,2.02e-4,2.7e-4},
        ymajorgrids=true,
        yticklabels={},
        tick style={draw=none},
                    ]
      \addplot[draw=none] coordinates {(1.e1,2.61e-4) (1.e7,2.61e-4)};
      \addplot[draw=none] coordinates {(1.e1,2.03e-5) (1.e7,2.03e-5)};
    \end{groupplot}
    \begin{groupplot}
      [
        group style={group size=2 by 1,horizontal sep=1.5cm,vertical sep=2cm},
        width=6.7cm,
        xmode=log,
        ymode=log,
        cycle list name=p4est2,
        legend columns=-1,
        legend style={/tikz/every even column/.append style={column sep=5pt}},
      ]

      \nextgroupplot[
        title={\fxn{Lnodes} ($n=1$) runtime in seconds\strut},
        xmin=4,
        ymax=100,
        xlabel=$P$,
        legend columns=7,
        legend style={/tikz/every even column/.append style={column sep=5pt}},
        legend style={at={(1.11,-0.27)},anchor=north,legend cell align=left},
                    ]

      \addlegendimage{empty legend}
      \addlegendentry{$P$, 16-way:}

      \foreach \i in {16,128,1024,8192,65536,458752} {%
        \addplot+[only marks] table [x index=0, y index=8]
        {timings/lnodes-rotcubes-b9e7/timana.data.\i};
        \addlegendentryexpanded{$\i$}
      }

      \addlegendimage{empty legend}
      \addlegendentry{$P$, 32-way:}

      \foreach \i in {32,256,2048,16384,131072,917504} {%
        \addplot+[only marks] table [x index=0, y index=8]
        {timings/lnodes-rotcubes-b9e7/timana.data.\i};
        \addlegendentryexpanded{$\i$}
      }

      \addlegendimage{empty legend}
      \addlegendentry{$P$, 64-way:}

      \foreach \i in {64,512,4096,32768,262144} {%
        \addplot+[only marks] table [x index=0, y index=8]
        {timings/lnodes-rotcubes-b9e7/timana.data.\i};
        \addlegendentryexpanded{$\i$}
      }

      \addplot+[no markers,gray,solid,mark repeat=3] table [x expr=\thisrowno{0}, y index=8]
        {timings/lnodes-rotcubes-b9e7/timana.data.strong.lnodes};
      \node [above] at (axis cs:16,0.017) {5.7k};
      \node [above] at (axis cs:16,0.058) {28k};
      \addplot[dotted,no markers,sharp plot,update limits=false] coordinates
      {(16,0.05835) (128,0.004729)};
      \node [above] at (axis cs:16,0.355) {240k};
      \addplot[dotted,no markers,sharp plot,update limits=false] coordinates
      {(16,0.3555) (1024,0.005546)};
      \node [above] at (axis cs:16,2.6) {2M};
      \addplot[dotted,no markers,sharp plot,update limits=false] coordinates
      {(16,2.60413) (8192,0.005086)};
      \node [above] at (axis cs:16,19.7) {16M};
      \addplot[dotted,no markers,sharp plot,update limits=false] coordinates
      {(16,19.7192) (65536,0.00481)};
      \node [above] at (axis cs:128,20.3) {130M};
      \addplot[dotted,no markers,sharp plot,update limits=false] coordinates
      {(128,20.2824) (458752,0.004951)};
      \node [above] at (axis cs:1024,20.8) {1B};
      \addplot[dotted,no markers,sharp plot,update limits=false] coordinates
      {(1024,20.7861) (458752,0.0406)};
      \node [above] at (axis cs:8192,21.4) {8B};
      \addplot[dotted,no markers,sharp plot,update limits=false] coordinates
      {(8192,21.3726) (458752,0.3339)};
      \node [above] at (axis cs:65536,21.8) {64B};
      \addplot[dotted,no markers,sharp plot,update limits=false] coordinates
      {(65536,21.8257) (458752,2.7282)};
      \node [above] at (axis cs:458752,25.2) {510B};

      \nextgroupplot[
        title={\hspace{-2em}\fxn{Lnodes} ($n=1$) runtime in secs.$/(N/P)$},
        ylabel=$\cdot 10^{-5}$,
        ylabel near ticks,
        axis y line*=left,
        y label style={rotate=-90,at={(rel axis cs:0.0,1.0)}},
        ytick={2.e-5,4.e-5,8.e-5,1.6e-4},
        yticklabels={$2$,$4$,$8$,$16$},
        xlabel=$N/P$,
        xmin=40.,
        xmax=2.e6,
        only marks,
                    ]

      \foreach \i in {16,128,1024,8192,65536} {%
        \addplot table [x expr=\thisrowno{2}/\thisrowno{0},
                        y expr=\thisrowno{8}/(\thisrowno{2}/\thisrowno{0})
                       ]
        {timings/lnodes-rotcubes-b9e7/timana.data.\i};
      }
      \addplot+[update limits=false] table [x expr=\thisrowno{2}/\thisrowno{0},
                      y expr=\thisrowno{8}/(\thisrowno{2}/\thisrowno{0})
                     ]
      {timings/lnodes-rotcubes-b9e7/timana.data.458752};
      \foreach \i in {32,256,2048,16384,131072} {%
        \addplot table [x expr=\thisrowno{2}/\thisrowno{0},
                        y expr=\thisrowno{8}/(\thisrowno{2}/\thisrowno{0})
                       ]
        {timings/lnodes-rotcubes-b9e7/timana.data.\i};
      }
      \addplot+[update limits=false] table [x expr=\thisrowno{2}/\thisrowno{0},
                      y expr=\thisrowno{8}/(\thisrowno{2}/\thisrowno{0})
                     ]
      {timings/lnodes-rotcubes-b9e7/timana.data.917504};
      \foreach \i in {64,512,4096,32768,262144} {%
        \addplot+[update limits=false] table [x expr=\thisrowno{2}/\thisrowno{0},
                        y expr=\thisrowno{8}/(\thisrowno{2}/\thisrowno{0})
                       ]
        {timings/lnodes-rotcubes-b9e7/timana.data.\i};
      }
    \end{groupplot}
  \end{tikzpicture}

  \vspace{\baselineskip}
  \begin{tikzpicture}[baseline]
    \begin{groupplot}
      [
        group style={group size=2 by 1,horizontal sep=1.5cm},
        width=6.7cm,
        ymin=1,
        ymax=5.e6,
        xmode=log,
        ymode=log,
        xlabel=$N/P$,
        ylabel=$P$,
        ylabel near ticks,
      ]

      \nextgroupplot[
        align=center,
        title={\fxn{Nodes} runtime / \\ \fxn{Lnodes} $(n=1)$ runtime},
        point meta=explicit,
        only marks,
        no marks,
        scatter,
        nodes near coords*={\pgfmathprintnumber\myvalue},
        visualization depends on={\thisrowno{7}/\thisrowno{8}\as\myvalue},
        every node near coord/.style={text=mapped color},
        colormap={lnodesvnodes}{color(0cm)=(red); color(1cm)=(green!50!black); color(16cm)=(blue)},
      ]

      \foreach \i in {2,16,128,1024,8192,65536} {%
        \addplot table [
          x expr=\thisrowno{2}/\thisrowno{0},
          y index=0,
          meta expr=\thisrowno{7}/\thisrowno{8},
                       ]
        {timings/lnodes-rotcubes-b9e7/timana.data.\i};
      }

      \nextgroupplot[
        align=center,
        title={\fxn{Lnodes} $(n=7)$ runtime / \\ \fxn{Lnodes} $(n=1)$ runtime \phantom{/}},
        point meta=explicit,
        only marks,
        no marks,
        scatter,
        nodes near coords*={\pgfmathprintnumber\myvalue},
        visualization depends on={\thisrowno{5}/\thisrowno{3}\as\myvalue},
        every node near coord/.style={text=mapped color},
        colormap={lnodesvnodes}{color(0cm)=(blue); color(1cm)=(green!50!black); color(2cm)=(red)},
      ]

      \foreach \i in {2,16,128,1024,8192,65536,458752} {%
        \addplot table [
          x expr=\thisrowno{2}/\thisrowno{0},
          y index=0,
          meta expr=\thisrowno{5}/\thisrowno{3},
                       ]
        {timings/multiple-lnodes-b9e7/timana.multiple-lnodes.data.\i};
      }
    \end{groupplot}

  \end{tikzpicture}

  \caption{%
    The parallel scalability of the \fxn{Lnodes} algorithm, as implemented by
    the \pforest function \fxn{p8est\_lnodes}.  (top) Runtimes for $n=1$.
    (top left) Runtime as a function of $P$, comparing strong scaling (solid lines)
    to ideal $O(N/P)$ scaling (dotted).  The total number of leaves $N$ in
    each mesh is indicated.  (top right) Runtime scaled by $N/P$ as a function
    of $N/P$.  Weak-scaling is assessed by comparing the vertical distance
    between points: each grid line represents a $25\%$ loss of
    weak-scaling efficiency.  (bottom left) The speedup of \fxn{Lnodes} versus
    \fxn{Nodes} as implemented by \fxn{p8est\_nodes} is shown for the same
    meshes as above.
    The color scale indicates whether \fxn{Lnodes} performs better (blue) or
    worse (red) than \fxn{Nodes}.
    (bottom right) The runtime for $n=7$, scaled by the
    runtime for $n=1$, with an analogue meaning of the colors.
  }%
  \figlab{nodescale}
\end{figure}
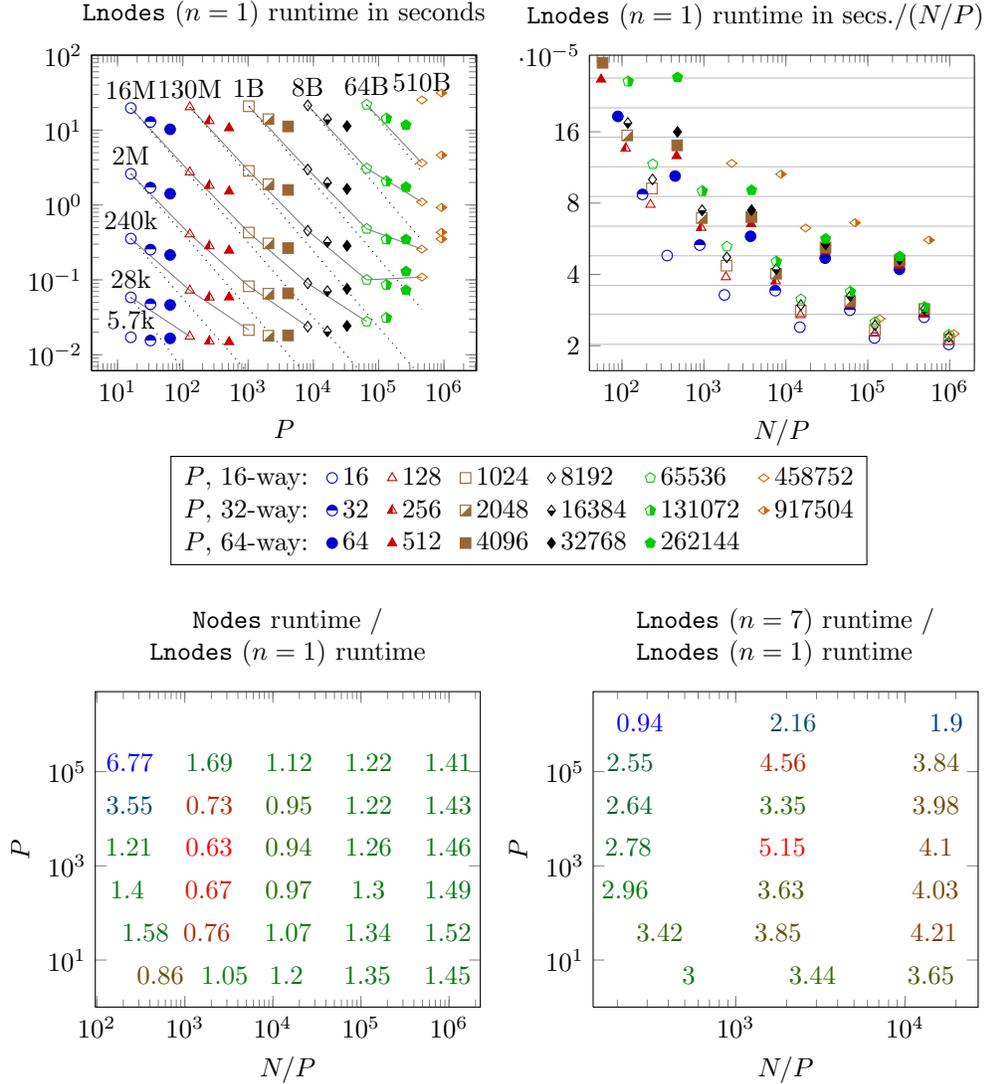

In \figref{nodescale} (top), we show the runtimes of \fxn{Lnodes} for $n=1$.
As discussed in \secref{iterimpl}, the implementation of \fxn{Iterate} has
been optimized for large values of $N/P$: this optimization requires
$O(\lmax)$ workspace and setup time.  The weak-scaling plot shows that the
optimization is effective, in that for $N/P \geq 10\mathrm{k}$ and $P\leq
262\mathrm{k}$ the weak-scalability is nearly ideal, and that the absolute
runtime is small (~20 seconds for $1M$ leafs/process). The optimization
requires redudant work, however, and this affects the efficiency for $N/P <
1\mathrm{k}$.  The strong-scaling plot shows good scalability for $P\leq
262\mathrm{k}$ and $(N/P) > 1\mathrm{k}$, and in this range the algorithm
benefits from 32 and 64 processes per node as well.  As in the scaling for
ghost layer construction, the communication latency and the small amount of
$O(P)$ work in the implementation finally limit the scalability for the
smallest meshes that were timed on the full machine.

In the same figure (bottom left) we compare the runtimes for \fxn{Lnodes} for
$n=1$ to the runtimes of \fxn{Nodes} \cite[Algorithm
21]{BursteddeWilcoxGhattas11}.  For most tests, \fxn{Lnodes} is faster than
\fxn{Nodes}: although the relative advantage is smaller on the Blue Gene/Q
architecture of JUQUEEN than on the Ivy Bridge architecture used in the serial
test, we still see the advantage increasing as $N/P$ increases, which is
suggestive of better cache performance.  The communication pattern of
\fxn{Lnodes}, consisting of one allgather and one round of point-to-point
communication, is more scalable than the communication pattern of \fxn{Nodes},
which includes a handshake component, hence the better performance of
\fxn{Lnodes} for small values of $N/P$ and large values of $P$.


Finally, \figref{nodescale} also compares the scalability of \fxn{Lnodes} for
higher polynomial orders to the scalability for $n=1$ (bottom left).  We see
that the runtime to construct \nth{7}-order nodes is never more than six times
the runtime to construct \nth{1}-order nodes, even though there are 64 times
as many element nodes and roughly 500 times as many global nodes.%
\footnote{%
  The number of global nodes depends on the forest topology and the refinement
  pattern.  For a single octree with uniform refinement, the number of global
  nodes is asymptotically equivalent to $n^3 N$, in which case the number of
  \nth{7}-order nodes would be 343 times the number of \nth{1}-order nodes.
  Because of non-conformal elements, however, we see a higher ratio.%
} %
For large values of $P$ the communication costs, which do not increase
significantly with $n$, dominate the runtime, so that the cost of constructing
high-order nodes is essentially the same as \nth{1}-order nodes.


\section{Conclusion}
\seclab{conclusion}

In this work, we introduce new recursive algorithms that operate on the
distributed forest-of-octrees data structures that the \pforest software
defines and uses to support scalable parallel AMR.  The algorithms developed
here exploit a recursive space partition from a topological point of view.
They constitute \pforest's high-level reference interface, which is designed to
be used directly from third-party numerical applications.

With the \fxn{Search} algorithm, we demonstrate how to efficiently traverse a
linear octree downward from the root, even though the flat storage of leaves
has no explicit tree structure.  This search operation is in some sense purely
hierarchical: a similar search could be performed even if the nodes and leaves
of the tree were not interpreted as a space partition in $\mathbb{R}^d$.

As a component of the \fxn{Ghost} algorithm, we propose a recursive algorithm
for determining the intersections between lower-dimensional boundary cubes and
ranges of leaves that are specified only by the first and last leaves in the
range.  This algorithm is notable in that, while the procedure is recursive on
the implicit octree structure, the result that it computes---a set of
intersections---is purely topological in nature.

In the \fxn{Iterate} algorithm, we present a method of performing
callback-based iteration over leaves and leaf boundaries that construct local
topological information for the callback on the fly.  This procedure combines
aspects of the two previous algorithms: it involves recursion over the octree
hierarchy and recursion over topological dimension.  The divide-and-conquer
nature of the algorithm makes better use of the cache hierarchy than
approaches to iteration that rely on repeated searches through the array of
leaves, as we demonstrate in practice.

We use \fxn{Iterate} in the construction of fully-distributed higher-order
$C^0$ finite element nodes in the algorithm \fxn{Lnodes}.  The topological
information provided by \fxn{Iterate} simplifies the handling of non-conformal
interfaces, and provides sufficient information to allow for node assignments
to be made without communication, and for the communication pattern between
referencing processes to be determined without handshaking.  In practice, this
gives us good scalability, which we have demonstrated to nearly a half million
processes on the JUQUEEN supercomputer. The implementation has been tuned for
granularities of a thousand leaves per MPI process and above, and in this
range we see good scalability, although room for improvement remains for
smaller granularities.

The scalability of \fxn{Lnodes} that we have demonstrated is important for
more applications than just higher-order finite element nodes, because the
data structures returned by the \fxn{Lnodes} algorithm can also serve as the
basis for converting a linear forest of octrees into an unstructured mesh
adjacency graph.  \fxn{Lnodes} includes all of the communication necessary for
this conversion, so the same scalability should be achievable by third-party
numerical codes that use \fxn{Lnodes} (or a similar approach based on
\fxn{Iterate}) to interface \pforest with their own mesh formats.

\section*{Reproducibility}

The algorithms presented in this article are implemented in the \pforest
reference software \cite{Burstedde10}.  \pforest, including the programs used
in the performance analysis presented above, is free and freely downloadable
software published under the GNU General Public License version 2, or (at your
option) any later version.

\section*{Acknowledgments}

The first author thanks the U.S.\ Department of Energy for support by the
Computational Science Graduate Fellowship (DOE CSGF) and by the Office of
Science (DOE SC), Advanced Scientific Computing Research (ASCR), Scientific
Discovery through Advanced Computing (SciDAC) program, under award number
DE-FG02-09ER25914.
The second author is supported by the Hausdorff Center for Mathematics (HCM) at
Bonn University and the Transregio 32 research collaborative, both funded by
the German Research Foundation (DFG).

The authors gratefully acknowledge the Gau\ss{} Centre for Supercomputing (GCS)
for providing computing time through the John von Neumann Institute for
Computing (NIC) on the GCS share of the supercomputer JUQUEEN at J\"ulich
Supercomputing Centre (JSC).  GCS is the alliance of the three national
supercomputing centres HLRS (Universit\"at Stuttgart), JSC (Forschungszentrum
J\"ulich), and LRZ (Bayerische Akademie der Wissenschaften), funded by the
German Federal Ministry of Education and Research (BMBF) and the German State
Ministries for Research of Baden-W\"urttemberg (MWK), Bayern (StMWFK) and
Nordrhein-Westfalen (MIWF).

The authors are indebted to three anonymous reviewers, whose remarks led to
significant improvements in the final form of this paper, and to Jose A.\
Fonseca and Johannes Holke for their editorial help.

\bibliographystyle{siam}
\bibliography{p4est-2,ccgo}

\appendix

\section{Proof of the correctness of \fxn{Find\_\-range\_\-boundaries}
  (\algref{findrange})}
\applab{findrangeproof}

\begin{theo}%
  Given a range $[\oc f,\oc l]$, where $\oc f$ and $\oc l$ are atoms 
  with a common ancestor $\oc s$, and given a set of boundary indices $\st
  B_\query \subseteq \st B$, \algref{findrange}
  returns the set $\st B_{\cap}(\oc f,\oc l,\oc s)\cap \st B_\query$ 
  \eqref{eqn:rangebound}.
\end{theo}
\begin{proof}
  The proof is inductive on the refinement level, $\oc s.l$. 

  If $\oc s.l = \lmax$, then the only descendant of $\oc s$ is itself, so $\oc
  s=\oc f=\oc l$.  Therefore $\overline{\bigcup\Dom([\oc f,\oc l])} =
  \overline{\Dom(\oc s)}$ and $\st B_{\cap}(\oc f,\oc l, \oc s) = \st B$, so
  $\st B_{\cap}(\oc f, \oc l, \oc s)\cap \st B_{\query} = B_\query$. This is
  correctly returned on line
  \ref{maxexit}.

  Now suppose that \fxn{Find\_\-range\_\-boundaries} returns correctly
  if $\lmax \ge \oc s.l \geq m$, and suppose $\oc s.l = m - 1$.  Let
  $\oc f\in\desc(\child(\oc s)[j])$ and $\oc l\in\desc(\child(\oc s)[k])$: the
  set $\mathcal{I}$ defined in \eqref{runion} is equal to $\{j,\dots,k\}$.
  
  If $j=k$, then by Proposition \ref{prop:isect}, $\st B_{\cap}(\oc f,\oc
  l,\oc s) = \st B_{\cap}(\oc f, \oc l, \child(\oc s)[j])\cap \st B_{\cap}^j$.
  By the inductive assumption, line \ref{equalchildren} returns
  \begin{equation}
    \st
    B_{\cap}(\oc f, \oc l, \child(\oc s)[j])\cap (\st B_{\cap}^j \cap \st
    B_{\query})=\st B_{\cap}(\oc f, \oc l, \oc s)\cap \st B_{\query}.
  \end{equation}

  Now suppose $j<k$.  If the range $[\oc f,\oc l]$ overlaps all of $\child(\oc
  s)[i]$, then the set $\st B_{\cap}(\oc f_i,\oc l_i,\child(\oc s)[i])$ is
  equal to $\st B$.  This is the case if $j<i<k$, so the set $\st
  B_{\rmatch}^1$ computed on line
  \ref{distinctchildrenmid} is
  \begin{equation}
    \st B_{\rmatch} =
    \bigcup_{j<i<k} \st B_\query \cap \st B_{\cap}^i =
    \bigcup_{j<i<k} \st B_\query \cap (\st B_{\cap}(\oc f_i,\oc
    l_i,\child(\oc s)[i]) \cap B_{\cap}^i).
  \end{equation}
  This is also the case for $i=j$ if $\oc f=\oc f_j$, so on each branch 
  of the condition on line \ref{distinctchildrenfirst} the set $\st
  B_{\rmatch}^j$ is computed as
  \begin{equation}
    \begin{aligned}
      \st B_{\rmatch}^j
      &= \st B_{\cap}(\oc f_j,\oc l_j,\child(\oc s)[j]) \cap ((\st B_\query
      \cap \st B_{\cap}^j)\backslash \st B_{\rmatch})
      \\
      &= (\st B_{\query} \cap (\st B_{\cap}(\oc f_j,\oc l_j,\child(\oc s)[j])
      \cap \st B_{\cap}^j))\backslash \st B_{\rmatch}.
    \end{aligned}
  \end{equation}
  By the same reasoning, on each branch of the conditional on line
  \ref{distinctchildrenlast}, the set $\st B_{\rmatch}^k$ is computed as
  \begin{equation}
    \st B_{\rmatch}^k
    = (\st B_{\query} \cap (\st B_{\cap}(\oc f_k,\oc l_k,\child(\oc s)[k])
    \cap \st B_{\cap}^k))\backslash \st B_{\rmatch} \backslash \st
    B_{\rmatch}^j.
  \end{equation}
  The union $\st B_{\rmatch}\cup \st B_{\rmatch}^j \cup \st B_{\rmatch}^k$ is 
  therefore equal to
  \begin{equation}
    \st B_{\query} \cap
    \bigcup_{j\leq i \leq k}
    \st B_{\cap}(\oc f_i,\oc l_i,\child(\oc s)[i])
    \cap \st B_{\cap}^i =
    \st B_{\query} \cap \st B_{\cap}(\oc f, \oc l, \oc s).
  \end{equation}
  By induction, the proof is complete.
  \hfill
\end{proof}

\section{Proof of the correctness of \fxn{Iterate\_interior}
  (\algref{iterinterior})}
\applab{iterproof}

Let the definitions in \secref{iterate} be given.  We prove the correctness of
\fxn{Iterate\_\-interior} (\algref{iterinterior}) when the relevant set is
$\overline{\localpart}$.  The proof for the case when
$\localpart$ is the relevant set is very similar.

\begin{theo}
  Assume that the requirements for the arguments of \algref{iterinterior} are
  met.  If $\oc c\in\overline{\localpart}$, then $\lsupp_{\pr p}(\oc c)$ is
  correctly computed.  If there is a subset of $\overline{\localpart}$ whose
  domain is contained in $\Dom(\oc c)$, then the callback function is executed
  for all points in that subset.
\end{theo}
\begin{proof}
  We first assert that if $\oc c\in\overline{\localpart}$, then its local leaf
  support $\lsupp_{\pr p}(\oc c)$ is a subset of $\bigcup_i \arry S[i]$.  By
  the definition of $\lsupp(\oc c)$, $\overline{\Dom(\oc o)}\cap \Dom(\oc c)
  \neq \emptyset$. By Proposition \ref{prop:closure}, there is a point $\oc e$
  in the closure set of $\oc o$, $\oc e\in\clos(\oc o)$, such that $\Dom(\oc
  e)\subseteq \Dom(\oc c)$ or $\Dom(\oc c)\subseteq\Dom(\oc e)$.  By the
  definition of the global partition set $\partoctants$, the former must be
  true: otherwise, $\oc c$ could not be in $\overline{\localpart}\subseteq
  \partoctants$.  Therefore $\oc o$ cannot be less refined than $\oc c$,
  $\level(\oc o)\geq \level(\oc c)$.  By Proposition~\ref{prop:suppisect},
  there must be some support octant $\oc s=\supp(\oc c)[i]$ such that $\oc
  o\in\desc(\oc s)$.  By the definition of $\arry S[i]$, it must contain $\oc
  o$.  This proves the first assertion.

  From here, we split the proof into two
  cases, $\dim(\oc c)=0$, and $\dim(\oc c)>0$.

  Suppose $\dim(\oc c)=0$. If $\oc o\in\lsupp_{\pr p}(\oc c)$, then there is
  $i$ such that $\oc o\in \arry S[i]$.  By Proposition~\ref{prop:duality2},
  $\oc o$ must be an ancestor of the atom $\asupp(\oc c)[i]$.  Therefore $\oc
  o$ is added to $\mathcal{L}$ on line \ref{itercorneradd}.  Conversely, if
  $\oc o$ is added to $\mathcal{L}$ on line \ref{itercorneradd}, then
  $\asupp(\oc c)[i]$ is a descendant of $\oc o$, and by definition its domain
  is in $\oc o$'s domain, $\overline{\Dom(\asupp(\oc c)[i])}\subseteq
  \overline{\Dom(\oc o)}$.  Because $\overline{\Dom(\asupp(\oc c)[i])}\cap
  \Dom(\oc c)\neq \emptyset$,
  it must be that $\overline{\Dom(\oc o)}\cap \Dom(\oc c)\neq\emptyset$.
  Therefore $\oc o$ is a leaf in $\arry S[i]\subset\localoctants\cup
  \arry{G}_{\pr p}^d$ whose closure intersects $\oc c$, which matches the
  definition of $\lsupp_{\pr p}(\oc c)$.  Thus, if $\oc c\in\partoctants$,
  the set $\mathcal{L}$ computed is equal to $\lsupp_{\pr p}(\oc c)$, and the
  callback will be executed on line \ref{itercallback} if and only if $\oc
  c\in \overline{\localpart}$.

  Now suppose $\dim(\oc c) > 0$.  Let $L$ be the minimum level of a leaf
  $\oc o\in\cup_{i} \arry S[i]$.  The remainder
  of the proof is inductive on the difference $\delta=L-\level(\oc c)$.

  Suppose $\delta = 0$, and let $\oc o\in\arry S[i]$ be a leaf with level
  $L=\level(\oc c)$.  Because $\oc o\subseteq \supp(\oc c)[i]$ and because
  $\level(\supp(\oc c)[i])=\level(\oc c)$ by definition, $\oc o=\supp(\oc
  c)[i]$.  Therefore $\overline{\Dom(\oc o)}\cap\Dom(\oc c)\neq\emptyset$ and
  $\oc o\in\lsupp_{\pr p}(\oc c)$.  Because leaves do not overlap, it must be
  that $\arry S[i] = \{\oc o\}$.  Therefore $\oc o$ is added to $\mathcal{L}$
  on line \ref{iteraddlarge}.

  Because of the 2:1 condition, all remaining leaves in $\lsupp_{\pr p}(\oc
  c)$ have level $L + 1$.  Let $\oc o\in\arry S[j]$ be a leaf with level $L +
  1$.  This implies that $\arry S[j]\neq \{\supp(\oc c)[i]\}$, so the children
  of $\supp(\oc c)[i]$ are assigned to $\arry h_i$ on line
  \ref{childrensplit}: $\oc o$ must be one of these children.  On line
  \ref{iteraddsmall}, $\oc o$ is added to $\mathcal{L}$ if and only if
  $\overline{\Dom(\oc o)}\cap\Dom(\oc c)\neq\emptyset$, which matches the
  definition of $\lsupp_{\pr p}(\oc c)$.  Therefore, if $\oc
  c\in\partoctants$, the constructed set $\mathcal{L}$ matches $\lsupp_{\pr
  p}(\oc c)$, and the callback executes on line \ref{itercallback} if and only
  if $\oc b\in \overline{\localpart}$.

  Now suppose the algorithm is correct for $0\leq \delta < k$, and suppose
  $\delta = k$.  There can be no $i$ such that $\arry S[i]=\{\supp(\oc
  c)[i]\}$, so the arrays $\arry H_i$ and octants $\arry h_i$ are computed on
  lines \ref{itersplit} and \ref{childrensplit} for every $i$.  Let $\oc e$ be
  in the child partition set $\part(\oc c)$: $\oc e$ has level $\level(\oc c)
  + 1$.  By definition, each octant in the support set $\supp(\oc e)$ also has
  level $\level(\oc c) + 1$ and $\overline{\Dom(\supp(\oc e)[i])}\cap \Dom(\oc
  e)\neq\emptyset$, which implies $\overline{\Dom(\supp(\oc e)[i])}\cap \Dom
  (\oc c) \neq \emptyset $.  Proposition~\ref{prop:suppisect} implies that
  there must be $j$ and $k$ such that $\supp(\oc e)[i] =  \child(\supp(\oc
  c)[j])[k]$.  Therefore $\supp(\oc e)[i] =  \arry h_j[k]$ and the set $\arry
  {S}_{\oc e}[i]=\arry H_j[k]$ is equal to $(\localoctants \cup \arry{G}_{\pr
    p}^d)\cap \desc(\supp(\oc e)[j])$.  This means that the arguments of the
  recursive call on line \ref{iterrecurse} are correct for each $\oc e\in
  \part(\oc c)$.  By the inductive assumption, the callback function is
  executed for the subset of $\overline{\localpart}$ whose domains are in
  $\Dom(\oc c)=\bigsqcup\Dom(\part(\oc c)$.  By the principle of induction,
  the proof is complete.
  \hfill
\end{proof}

\section{Asymptotic analysis of \fxn{Iterate} (\algref{iter})}
\applab{iteranalysis}

We first present the asymptotic analysis of the complexity of the algorithm
in a single-process, single-octree setting.

\begin{theo}
Ignoring the time taken by the callbacks,
\fxn{Iterate} executes in the worst case in $O(N\log N)$ time.
\end{theo}
\begin{proof}
  The only operations in each instance of \fxn{Iterate\_interior} that are not
  $O(1)$ are the $O(\log |\arry S[i]|)$  terms for the input arrays $\arry
  S[i]$.  Each of these arrays is associated with an octant $\supp(\oc
  c)[i]$ that is an ancestor of a leaf.  An octant $\oc o$ can only be in
  $\supp(\oc c)$ if $\oc c\in \bound(\oc o)$ and $\oc c=(\oc o,b)$ for some $b
  \in \st B$.  Therefore each ancestor octant can be associated with at most
  $|\st B|$ terms with $O(\log |\arry S[i]|)$ complexity.  An octree has
  $O(N)$ ancestors that are not leaves, so $O(N)$ searches are conducted.
  Each array $\arry S[i]$ contains a subset of leaves, so each $O(\log |\arry
  S[i]|)$ is $\cO(\log N)$.  We conclude that an upper bound on the running
  time is $O(N \log N)$.
  \hfill
\hfill\end{proof}

\begin{theo}\label{theo:linear}
Ignoring callbacks, \fxn{Iterate} executes in $O(N)$ time on a uniformly
refined octree.
\end{theo}
\begin{proof}%
  The leaves are all at the same level $L$, so $N=2^{d L}$, and there are
  $2^{d l}$ nodes in level $l$ of the tree.  Because leaves are evenly
  distributed, each node at level $l$ has $2^{d(L - l)}$ leaf descendants.
  Each node is associated with a bounded number of binary searches and calls
  to \fxn{Split\_array}, each with logarithmic complexity in the number of
  leaves beneath it.  So, ignoring leading coefficients, the time complexity
  is
  \begin{equation}
    \begin{aligned}
      \sum_{ l  = 0}^{ L  - 1}2^{d l } \log 2^{d( L  -  l )}
      &=
      d\sum_{ l  = 0}^{ L  - 1}2^{d l } ( L  -  l ) \\
      &=
      d\sum_{ l  = 0}^{ L  - 1}\frac{2^{d L }}{2^{d( L  
      -  l )}}
      ( L  -  l ) \\
      [\hat l  =  L  -  l ]\quad &=
      d2^{d L }\sum_{\hat l  = 
      1}^{ L }\frac{\hat l }{2^{d\hat l }} = d2^{d L }O(1) = d O(N).
    \end{aligned}
  \end{equation}
  Because the dimension $d$ is fixed, \fxn{Iterate} runs in $O(N)$ time.
\hfill\end{proof}

A uniformly refined octree is just a regular grid, so the indices of neighbors
follow a predictable rule: a linear-time algorithm can be achieved without a
recursive algorithm and without searching through the leaf arrays.  We outline
a class of octrees which has no rule for neighboring indices,
but for which \fxn{Iterate} still runs in linear time.

\begin{defn}[$\Delta$-uniform octrees]
  A class of octrees is \emph{$\Delta$-uniform} if the difference
  $(\max_{\oc o\in\cO} \oc o.l - \min_{\oc o\in\cO} \oc o.l)$
  is uniformly bounded by $\Delta$ for all octrees in the class.
\end{defn}%

\begin{theo} \fxn{Iterate} executes in $O(N)$ time on a class of
  $\Delta$-uniform octree.
\end{theo}
\begin{proof}
  Let $L=\max_{\oc o\in\cO}\oc o.l$ and $l_{\min}=\min_{\oc o\in \cO}\oc o.l$.
  For $l \min_{\oc o\in\cO} \oc o.l$,
  $2^{d\ell}$ is now an upper bound on the number of nodes at level $l$, and
  for every $l$, $2^{d(L-l)}$ is an upper bound on the number of descendant
  leaves of a level $l$ node.  Therefore the $O(2^{d L})$ runtime for a uniform
  octree is an upper bound on the runtime of \fxn{Iterate}, while a lower
  bound on $N$ is $2^{d l_{\min}} = 2^{d(L-\Delta)}$.  Therefore
  $2^{d L}\leq 2^{d\Delta}N$, so the runtime of \fxn{Iterate} is
  $O(2^{d\Delta}N)=O(N)$.  \hfill
\end{proof}

We now consider the \fxn{Iterate} algorithm in the multiple process, single
octree setting, and derive bounds in terms of the local number of leaves
$N_{\pr p}$ and the number of processes $P$.  A key component of the above
analysis for the serial runtime, that the number of ancestor nodes is $O(N)$,
is no longer true in a parallel setting: the number of ancestors of the leaves
in $\localoctants\cup\arry G_{\pr p}^d$ is not necessarily $O(N_{\pr p})$.
Suppose $\oc a$ is the smallest common ancestor of every leaf in
$\localoctants\cup\arry G_{\pr p}^d$ and $\oc a.l=\hat{l}$.  The
number of branches below $\oc a$ must be $O(N_{\pr p})$, so the analysis for
the runtime after level $\hat{l}$ is the same as for a single process,
substituting $\oc a$ for the root, so the time spent below $\oc a$ is
$O(N_{\pr p} \log N_{\pr p})$ in general or $O(N_{\pr p})$ for a
$\Delta$-uniform tree.  Thus an upper bound for the runtime is to add
$O(\hat{l}\log N_{\pr p})$ to that time.  We can bound $\hat{l}$ by
$L=\max_{\oc o\in\cO}\oc o.l$, and in the $\Delta$-uniform case $L\in O(\log
N)$. If we assume an even partitioning of the leaves, $N= P N_{\pr p}$, then
$L\in O(\log P + \log N_{\pr p})$.  The runtime for \fxn{Iterate} on an evenly
distributed octree is thus $O((L + N_{\pr p}) \log N_{\pr p})$ in general and
$O(\log P + N_{\pr p})$ for $\Delta$-uniform octrees.

Introducing multiple trees does not affect the analysis significantly:
maintaining separate arrays for each tree can only reduce the sizes of the
subarrays that are split by \fxn{Split\_array}.  Some time is taken to set up
the calls to \fxn{Iterate\_interior} for the interfaces between octrees, but
this time is negligble, especially if the forest realizes the common use case
$K \ll N$.

\end{document}